\documentclass{article}
\usepackage[utf8]{inputenc}
\usepackage{csquotes}
\usepackage[greek,english]{babel}
\usepackage{polski}

\usepackage{subcaption}
\usepackage{soul}

\usepackage{graphicx}
\usepackage[margin=1.45in]{geometry}

\usepackage{amsthm}
\usepackage{amssymb}
\usepackage{amsfonts}
\usepackage{amsmath}
\usepackage{upgreek} 

\usepackage{algorithm}
\usepackage{algpseudocode}

\usepackage{hyperref}
\usepackage[noabbrev, capitalize]{cleveref} 
\crefalias{algo}{Algorithm}

\usepackage{xcolor}

\newtheorem{theorem}{Theorem}[section]
\newtheorem{corollary}[theorem]{Corollary}
\newtheorem{lemma}[theorem]{Lemma}
\newtheorem{proposition}[theorem]{Proposition}

\theoremstyle{definition}

\theoremstyle{remark}
\newtheorem{remark}[theorem]{Remark}
\newtheorem{example}[theorem]{Example}

\crefname{equation}{equation}{equation} 

 \DeclareMathOperator{\diag}{diag}

\DeclareFontFamily{U}{mathx}{\hyphenchar\font45}
\DeclareFontShape{U}{mathx}{m}{n}{
      <5> <6> <7> <8> <9> <10>
      <10.95> <12> <14.4> <17.28> <20.74> <24.88>
      mathx10
      }{}
\DeclareSymbolFont{mathx}{U}{mathx}{m}{n}
\DeclareMathSymbol{\bigtimes}{1}{mathx}{"91}

\providecommand{\keywords}[1]
{
  \small	
  \emph{\phantom{xx}Keywords:  #1 }
}

\title{\Large{ 
Computation of Robust Option Prices via \\ Structured Multi-Marginal Martingale Optimal Transport\\[10pt]
}}

\author{Linn Engström\footnote{KTH Royal Institute of Technology; linneng@kth.se} \and Sigrid Källblad\footnote{KTH Royal Institute of Technology; sigrid.kallblad@math.kth.se} \and Johan Karlsson\footnote{KTH Royal Institute of Technology; johan.karlsson@math.kth.se} }

\date{}

\begin{document}

\maketitle

\begin{abstract} 

We introduce an efficient computational framework for solving a class of multi-marginal martingale optimal transport problems, which includes many robust pricing problems of large financial interest.
Such problems are typically computationally challenging due to the martingale constraint, however,  by extending the state space we can identify them with problems that exhibit a certain sequential martingale structure. 
Our method exploits such structures in combination with entropic regularisation, enabling fast computation of optimal solutions and allowing us to solve problems with a large number of marginals. 
We demonstrate the method by using it for computing robust price bounds for different options, such as lookback options and Asian options.
\end{abstract}


\section{Introduction} \label{sec:intro}

A fundamental problem in mathematical finance is to find fair prices of financial claims. 
The main goal of this article is to compute robust and model independent bounds on prices of exotic options
when there is uncertainty about the true underlying market model. 
More specifically, we develop computationally efficient methods for obtaining such bounds when restricting to market models that are consistent with given market data and respect fundamental market modelling principles. 
For market data consisting of prices of liquidly traded call options for multiple maturities, the problem can be formulated as a multi-marginal martingale optimal transport (MOT) problem, a problem which has been extensively studied over the last decade. 
The approach taken here leverages recent computational results for solving multi-marginal optimal transport problems using entropic regularisation. 
More pertinently, we will combine entropic regularisation with methods utilising specific structures in the payoff function so as to obtain an efficient method for MOT problems in the presence of multiple marginal constraints.

The classical approach to pricing of financial derivatives is to start by postulating a market model. A model effectively consists of a filtered probability space, say $(\Omega, \mathcal{F},\mathbb{F},\mathbb{P})$, supporting a stochastic process $S$ modelling the underlying future stock prices. We here restrict to discrete time and suppose that the risk-free interest rate is zero.
Such a model is then a valid model, in the sense that it does not allow for arbitrage-opportunities, if there exists a so-called risk-neutral measure $\mathbb{Q}$ 
rendering the price process $S$ a martingale. 
For an exotic option ensuring the holder a payoff given by $\phi(S_0, \dots, S_T)$, for some function $\phi:\mathbb{R}^{T+1}\to\mathbb{R}$, it turns out that any fair price thereof --- that is, any price which does not introduce arbitrage opportunities to the market --- must be given by 
\begin{equation}\label{eq:intro_fair_price}
    \mathbb{E}_{\mathbb{Q}} \left[\phi(S_0, \dots, S_T) \right]
\end{equation}
for such a risk-neutral pricing measure $\mathbb{Q}$. 
Fair prices are thus highly dependent on the original choice of market model, as well as on the choice of pricing measure. 
In reality, however, neither of the two is known. The acknowledgment of this fact has led to an intense stream of research aiming at quantifying the impact of the modelling assumptions. 

A natural question to ask is whether it is possible to obtain robust bounds on fair prices. Naïvely, such bounds could be obtained by optimising \eqref{eq:intro_fair_price} over all possible market models and pricing measures; that is, over all probability spaces $(\Omega,\mathcal{F},\mathbb{Q})$ supporting a stochastic process $S$ which is a martingale in its own filtration under $\mathbb{Q}$. Such a procedure leads however to bounds which are too wide to be of any practical use.
Consensus has therefore been reached on also acknowledging that certain financial products are traded to such a great extent that their market prices can be viewed as `correct' and should contain information about the market participants' beliefs about the future. Put differently, the prices of liquidly traded options carry information about the pricing measure effectively used by the market. It is thus sensible to further restrict to market models and pricing measures for which the computed prices of such liquidly traded options agree with actual market prices. 
The most common choice in the literature is to strive for consistency with market prices of European call options, the perhaps most commonly traded option. Breeden and Litzenberger \cite{Breeden-Litzenberger1978} observed that if one has access to call prices for a continuum of strikes, at given maturities $\mathcal{T}\subseteq\{0,1,\dots,T\}$, one may recover the marginal distributions of the underlying price process at these dates. The thus obtained family of marginal distributions, say $\{\mu_t \}_{t \in \mathcal{T}}$, may therefore be considered to be known and fixed; it reflects what we know about the market-implied pricing measure.
The problem of solving for the associated price bound, here formulated in terms of a lower bound, takes the following form:
\begin{subequations} \label{eqn:intro_mot_general}
\begin{align}
    \underset{(\Omega, \mathcal{F}, \mathbb{Q}, S)}{\inf}  \;\quad &\mathbb{E}_{\mathbb{Q}} [\phi(S_0, \dots, S_T)] \label{eqn:intro_mot_general_objfcn} \\
    \text{subject to} \quad & \mathbb{E}_{\mathbb{Q}}[S_t | \sigma( S_0,\dots,S_{t-1}) ] = S_{t-1}, &&t \in \{1, \dots, T \} \label{eqn:intro_mot_general_mtgconstr} \\
    & S_t \sim_{\mathbb{Q}} \mu_t, &&t \in \mathcal{T} . \label{eqn:intro_mot_general_margconstr}
\end{align}
\end{subequations}
Note that if $\{\mu_t\}_{t\in\mathcal{T}}$ are obtained from consistent call prices, they are marginal distributions of a martingale and are thus in convex order. By Strassen \cite{Strassen1965}, it turns out that the latter is both a necessary and sufficient condition for problem \eqref{eqn:intro_mot_general} to be well posed.
The robust pricing problem was first introduced by Hobson \cite{Hobson1998}. Subsequently a rigorous duality theory has been established motivating the problem also from a pricing-hedging perspective; see Cheridito, Kupper and Tangpi \cite{CheriditoKupperTangpi2017} and Dolinsky and Soner \cite{DolinskySoner2014}.

About a decade ago, the insightful observation was made that problem \eqref{eqn:intro_mot_general} can be viewed as an optimal transport problem with an additional martingale constraint --- an MOT problem; see Beiglböck, Henry-Labordère and Penkner \cite{Beiglbock-HL-Penkner2013} and Galichon, Henry-Labordère and Touzi \cite{Galichon-HL-Touzi2013}. 
The optimal transport (OT) problem --- that is, problem \eqref{eqn:intro_mot_general} without condition \eqref{eqn:intro_mot_general_mtgconstr} --- is a classical problem that has received renewed attention in recent years; we refer to Villani \cite{Villani2003} for an overview. It was originally formulated for two marginals, and then amounts to finding a transport plan moving the mass of one distribution to another while minimising the associated cost, but has since been extended to transportation problems over multiple marginals; see, e.g., Gangbo and Swi\k{e}ch \cite{GangboSwiech1998}, Rüschendorf \cite{Rüschendorf1996} and Pass \cite{Pass2015}.
Equipped with the additional martingale constraint, the MOT problem turned out to be an intriguing problem which triggered intense research and led to many new developments; for multi-marginal results, we refer to Nutz, Stebegg and Tan \cite{NutzStebeggTan2020} and Sester \cite{Sester2023}. Prior to the contributions of \cite{Beiglbock-HL-Penkner2013} and \cite{Galichon-HL-Touzi2013}, the robust pricing problem was typically addressed via the so-called Skorokhod embedding problem (SEP). This approach was initiated by \cite{Hobson1998}, who considered the problem with one marginal constraint and a claim depending on the past maximum of a \emph{continuous} price process, and linked it to the well-known Azéma--Yor solution of the SEP. 
The connection to the SEP remained for long the predominant approach for addressing (continuous) robust pricing problems and extensive research in this area led to new contributions in terms of both robust price bounds and solutions of the SEP; for multi-marginal results in this direction, we refer, e.g., to 
Beiglb{\"o}ck, Cox and Huesmann \cite{BeiglboeckCoxHuesmann2020}, 
Cox, Obłój and Touzi \cite{CoxOblojTouzi2019}, and
Henry-Labordère et al. \cite{Henry-LabordèreOblojSpoidaTouzi2016}, see also \cite{KällbladTanTouzi2017}.

Recently, there has been a rapid development of computational theory and algorithms for OT problems.
The OT problem can be formulated as a linear programming problem when restricting the marginals to have support on a pre-specified finite grid.
However, the number of variables is often large, making the problem computationally intractable for many problems of practical interest. In particular, this is the case for multi-marginal problems, since the number of variables grows exponentially in the number of marginals. 
One of the dominant approaches for handling this has been to utilise entropic regularisation, inspired by Cuturi \cite{Cuturi2013} and Benamou et al. \cite{BenamouCarlierCuturiNennaPeyre2015}. 
The regularised problem can then be addressed by use of duality theory. In particular, the dual problem contains significantly fewer variables than the primal problem, since the number of dual variables is linear in the number of marginals. Further, the dual problem has a certain decomposable structure, which allows for obtaining explicit expressions characterising the optimal dual variables.
These expressions allow for deriving efficient update formulas for blocks of dual variables, which makes coordinate dual ascent --- that is, to cyclically optimise over subsets of the dual variables --- a particularly suitable method for solving the regularised problem.  
In the context of OT, this method is often referred to as Sinkhorn's algorithm for the bi-marginal problem. It was originally utilised for matrix balancing in Sinkhorn \cite{Sinkhorn1967}, but has since been reinvented several times. It can be derived, for example, using Bregman projections or Dykstra’s algorithm; see \cite{BenamouCarlierCuturiNennaPeyre2015} and Peyr{\'e} and Cuturi \cite{PeyreCuturi2019} for further details.
It can be shown that that the dual iterates generated by this method converge linearly, see Luo and Tseng \cite{LuoTsengQ1992} (cf. Franklin \cite{franklin1989} for the bi-marginal case), and for bi-marginal OT problems, entropic regularisation thus yields an efficient method for solving the problem. 
For multi-marginal problems, however, the computational complexity also increases exponentially as a function of the number of marginals, see Lin et al. \cite{LinHoCuturiJordan2022}, and entropic regularisation alone is thus not enough for efficient computation of the problem.
Fortunately, for many multi-marginal problems the cost function has graph-structures that can be exploited in combination with entropic regularisation. 
These are referred to as \emph{graph-structured multi-marginal optimal transport problems} and have been studied, e.g., in 
\cite{BenamouCarlierCuturiNennaPeyre2015, ElvanderHaaslerJakobssonKarlsson2020,HaaslerKarlssonRingh2021,HaaslerRinghChenKarlsson2021,HaaslerSinghZhangKarlssonChen2021,MascherpaHaaslerAhlgrenKarlsson2023}; see also
Altschuler and Boix-Adserà \cite{AltschulerBoix-Adsera2023} and Beier et al. \cite{BeiervonLindheimNeumayerSteidl2023}. In order to carry out the coordinate dual ascent method, one effectively needs to compute certain matrix projections; the key point utilised in those papers is that by suitably exploiting the additional structure of the cost function, said projections can be computed rapidly, thus rendering an efficient algorithm also for very large numbers of marginal constraints.

For MOT problems, entropic regularisation was first introduced by de March \cite{deMarch2018}, who derived a method based on Sinkhorn's algorithm and Newton's method for solving bi-marginal problems. As for early numerical results on robust pricing problems, we also mention Henry-Labordere \cite{HenryLabordere2013}, Davis, Ob{\l}{\'o}j and Raval \cite{DavisOblojRaval2014} and Tan and Touzi \cite{TanTouzi2013}.
Computationally solving multi-marginal MOT problems has since been addressed with neural network approaches in Eckstein and Kupper \cite{EcksteinKupper2021a} and Eckstein et al. \cite{EcksteinGuoLimObloj2021}, and by linear programming approaches in \cite{EcksteinGuoLimObloj2021} and Guo and Ob{\l}{\'o}j \cite{GuoObloj2019}, where the latter relies on a relaxation of the martingale constraint. 
However, due to the curse of dimensionality, it is computationally challenging to utilise these methods for problems with a large number of marginals. Markovian assumptions in a multi-marginal MOT context were introduced in Henry-Labordère \cite[Remark 2.7]{HenryLabordere2017} and further discussed in Sester \cite{Sester2020}, who considered a problem subject to an explicit Markov constraint. Such problems are not convex, which motivated Eckstein and Kupper \cite{Ecksteinkupper2021b} to consider multi-marginal MOT problems under a time-homogeneity assumption. Other related works include Alfonsi et al. \cite{AlfonsiCoyaudEhrlacherLombardi2021} and Lindheim and Steidl \cite{LindheimSteidl2023}, who have developed methods for OT problems featuring moment and affine constraints, respectively.

To the best of our knowledge, there are no previous methods that allows for computing multi-marginal MOT problems with a \emph{large} number of marginals. This article aims to fill this gap, and a main contribution is a method that can efficiently compute solutions to such problems. We combine entropic regularisation with the exploitation of structures inherent in the problem, in order to computationally solve the multi-marginal MOT problem \eqref{eqn:intro_mot_general}. 
In line with the above-mentioned literature on graph-structured multi-marginal OT problems, the idea is to exploit specific structures in the cost function. However, the presence of the martingale constraint introduces dependencies between the marginals and as a result, a multi-marginal MOT problem is not sparse --- this is a fundamental difficulty when it comes to addressing such problems computationally.
To get around this problem we first restrict ourselves to financial derivatives whose payoff is of the form
\begin{equation}  \label{eqn:intro_intro_phi}
    \phi(S_0, \dots, S_T) = \sum_{t=1}^T \phi_{t}(S_{t-1},X_{t-1},S_{t}, X_t),
\end{equation}
where $X_0= h_0(S_0)$ and $X_t=h_t(S_{t},S_{t-1},X_{t-1})$, for some functions $h_0: \mathbb{R}\to\mathbb{R}$, $h_t:\mathbb{R}^3\to\mathbb{R}$ and $\phi_t:\mathbb{R}^4\to\mathbb{R}$, $t \in \{1,\dots,T \}$.
We then utilise this structure in a twofold manner. First, leveraging on the fact that for a given stochastic process, one can construct a Markov process with the same pair-wise joint marginal distributions onto adjacent marginals, for payoffs of the form \eqref{eqn:intro_intro_phi}, we manage to relax condition \eqref{eqn:intro_mot_general_mtgconstr} to $\mathbb{E}_{\mathbb{Q}}[S_t |\sigma(S_{t-1},X_{t-1})] = S_{t-1}$, $t\in\{1,\dots,T\}$. Note that this is done without introducing any explicit Markovian assumption on the price process or any assumption regarding the convexity of the payoff function. In effect, we are facing a problem where each constraint is only connecting adjacent marginals and we may thus employ methods from the theory on graph-structured multi-marginal OT problems to find an efficient algorithm for solving problem \eqref{eqn:intro_mot_general}. We mention that we, in one of our numerical examples (see \cref{sec:sol_lateearly}), address an MOT problem with more than 50 marginals, where each marginal is supported on a large number of points. To the best of our knowledge, this is significantly larger than what has been reported for other methods. We illustrate the usefulness of our method by applying it on a number of examples: First, we apply our algorithm to problems for which the theoretical solution is known; we use it to recover the left-monotone coupling, late and early transport behavior, as well as optimal models for digital options and lookback options. We then use it to compute the robust price bound for an Asian option, subject to several marginal constraints, where as far as we know, no theoretical solution is known.

The rest of the article is organised as follows: In \cref{sec:gb}, we characterise our class of multi-marginal MOT problems in terms of properties inherent in the payoff function; we also provide examples of derivatives of practical financial interest belonging to this class, and provide some background on structured OT. \Cref{sec:equiv} is devoted to the reformulation of the problem onto a form that allows for efficient computations: We show that for our class of payoff functions, the MOT problem is equivalent to an OT problem where the complexity of the martingale constraint has been reduced to involve only adjacent marginals; the latter problem is then discretised and formulated on tensor form. In \cref{sec:algo}, relying on the structure inherent in our reformulated problem, we derive efficient formulas for the optimal dual variables corresponding to its entropy regularised version. In \cref{sec:sol}, we apply our computational framework in a number of different examples and compare our results with existing analytical solutions.

\subsection*{Notation and conventions}

We let $\odot$ and $\oslash$ denote the elementwise product and division of two tensors (or vectors/matrices), while $\exp(\cdot)$ denotes elementwise application of the exponential function. The Kronecker product is indicated by $\otimes$. The notation $\diag(v)$ represents a diagonal matrix with the vector $v$ on the main diagonal, while $1_n$ and $0_n$ are the vectors of length $n \in \mathbb{N}$ consisting of ones and zeros, respectively, and $\mathbf{1}_{n^m}$ for $m \in \mathbb{N}$ denotes the tensor in $\mathbb{R}^{n^m}$ where all elements are equal to one. The one-vector or the one-tensor should not be mixed up with the indicator function  $\chi_A : \mathbb{R}^n \rightarrow \mathbb{R}$  associated with some set $A \subset \mathbb{R}^n$, 
\begin{align*}
    \chi_A(x) :=
    \begin{cases}
        1, & x \in A \\
        0, & x \notin A.
    \end{cases}
\end{align*}
We let $[m]:=\{0, \dots, m\}$ for any $m \in \mathbb{N}$ and write $[m] \backslash k$ for short for $[m] \backslash \{k \}$ for $k \in [m]$. 
The inner product $\langle \cdot, \cdot \rangle$ refers to the Frobenius inner product, that is, 
\begin{equation*}
    \langle \mathbf{A}, \mathbf{B} \rangle
    = \sum_{i_0, \dots, i_n  \in [m]} \mathbf{A}(i_0, \dots, i_n) \mathbf{B}(i_0, \dots , i_n), \quad \mathbf{A}, \mathbf{B} \in \mathbb{R}^{n^m},
\end{equation*}
and we let $P_j: \mathbb{R}^{n^m} \rightarrow \mathbb{R}^n$ denote the projection operator that projects a tensor onto its $j$:th dimension, defined by 
\begin{equation*}
    P_j(\mathbf{A})(i_j) := \sum_{i_{\ell}: \ell\in [m]\backslash j } \mathbf{A}(i_0, \dots i_m), \quad i_j = 1, \dots , n, \quad\mbox{ where } \mathbf{A} \in \mathbb{R}^{n^m}.
\end{equation*}
Similarly, we let $P_{j,k}: \mathbb{R}^{n^m} \rightarrow \mathbb{R}^{n^2}$ be the projection operator that projects the tensor jointly onto the two dimensions $j$ and $k$ (see \cite[p. 7]{ElvanderHaaslerJakobssonKarlsson2020}).

Every real-valued function is assumed to be Borel measurable and we let the regular conditional distribution be equal to some real-valued constant if conditioning on a null event. Zero is assumed to be included in the set of positive real numbers and we use the convention of $0 \cdot \infty = 0$.

\section{Problem formulation and background} \label{sec:gb}
In this section we describe our main problem of interest, along with some examples demonstrating its relevance. We then provide a brief introduction to optimal transport with emphasis on the computational aspects of the problem; in particular, we review how a certain type of structure inherent in some transportation problems enables fast computation of optimal solutions.

\subsection{Problem formulation for structured payoff functions} \label{sec:intro_problemformulation}
We will now specify our class of problems, which consists of MOT problems where the payoff function is of a certain form.  
Let a probability space $(\Omega, \mathcal{F}, \mathbb{Q})$ and a price process $S : \Omega \times [T] \rightarrow \mathbb{R}$ be given such that $S$ is a martingale under $\mathbb{Q}$. Then define a second stochastic process $X : \Omega \times [T] \rightarrow \mathbb{R}$ by
\begin{equation} \label{eqn:intro_x}
    \begin{aligned}
    X_t =
    \begin{cases}
        h_0(S_0), & t = 0\\
        h_t(S_{t}, S_{t-1}, X_{t-1}), & t \in [T]\backslash 0
    \end{cases}
    \end{aligned}
\end{equation}
where $h_0 : \mathbb{R} \rightarrow \mathbb{R}$ and $h_t :\mathbb{R}^3 \rightarrow \mathbb{R}$ for $t \in [T] \backslash 0 $. By defining $X$ in this way we are able to isolate important features of the path of $S$, thus enabling the pricing of many exotic options. We are here interested in pricing path-dependent derivatives whose payoff function $\phi: \mathbb{R}^{T+1} \rightarrow \mathbb{R}$ is finite and can be pairwise decoupled, that is, 
\begin{equation}  \label{eqn:intro_phi}
    \phi(S_0, \dots, S_T) = \sum_{t \in [T] \backslash 0 } \phi_{t}(S_{t-1}, X_{t-1}, S_t, X_t),
\end{equation}
where $\phi_t : \mathbb{R}^4 \rightarrow \mathbb{R}$ for $t \in [ T ] \backslash 0$ and $X$ is obtained from (\ref{eqn:intro_x}) via a given family of functions $h_t$, $t \in [T]$. For payoff functions of this form the MOT problem (\ref{eqn:intro_mot_general}) becomes
\begin{subequations} \label{eqn:intro_mot}
\begin{align}
    \underset{(\Omega, \mathcal{F}, \mathbb{Q}, S)}{\inf}  &\mathbb{E}_{\mathbb{Q}} \bigg[ \sum_{t \in [T] \backslash 0} \phi_{t}(S_{t-1}, X_{t-1}, S_t, X_t) \bigg] \label{eqn:intro_mot_objfcn} \\
    \text{subject to} \quad & \mathbb{E}_{\mathbb{Q}}[S_t | \sigma( S_0, \dots S_{t-1})] = S_{t-1}, &&t \in [T] \backslash 0 \label{eqn:intro_mot_mtgconstr} \\
    & S_t \sim_{\mathbb{Q}} \mu_t, &&t \in \mathcal{T}. \label{eqn:intro_mot_margconstr}
\end{align}
\end{subequations}
Here we assume that the given marginals $\mu_t$, for $t\in\mathcal{T}\subseteq\{0,1,\dots,T\}$, are in convex order. Such problems are well posed, since it follows by Strassen's theorem that their feasible sets are nonempty \cite{Strassen1965}. We remark that problems on the above given form are convex, albeit not strictly convex, in $\mathbb{Q}$; however, showing the existence of optimal solutions requires further assumptions on the payoff function which we do not want to introduce here (see, for example, \cite[Theorem 1.1]{Beiglbock-HL-Penkner2013}).

The contribution of this work is the development of a computational method that allows for approximately solving this class of multi-marginal MOT problems for $T$ large. As we will see, many problems of practical financial interest belong to this class.

\begin{remark}
    We here limit ourselves to the case when the auxiliary process $X$ takes values on $\mathbb{R}$, but it is our belief that our framework could be extended to the case when the process $X$ takes values on $\mathbb{R}^d$ for $d > 1$. Note that this would allow us to take more information about the path of the price process $S$ into account, but only at the cost of the size of the problem increasing --- we emphasise that our computational framework is useful only when $d$ is significantly smaller than $T$.
\end{remark}

\subsection{Financial examples} \label{sec:examples}

Below follows a few examples of financial derivatives whose payoff functions satisfy the assumptions of our framework; the first four of them will be revisited later in \cref{sec:sol}, while the remaining ones illustrate the versatility of our class of problems.

\begin{example}[Lookback and barrier options] \label{ex:examples_max}
Consider the case when $X$ is the maximum process of the price process $S$, that is, when $X_t = \max_{r \in [t]} S_r $ for $t \in [T]$; since it can be written
\begin{equation*} \label{eqn:intro_examples_max}
\begin{aligned}
    X_t = 
    \begin{cases}
        S_0, & t = 0\\
        \max\{S_t, X_{t-1}\}, & t \in [T] \backslash 0
    \end{cases}
\end{aligned}
\end{equation*}
the maximum process is of the form (\ref{eqn:intro_x}) (the minimum process can be considered analogously). Related derivatives are single asset, single period maximum lookback options such as the \emph{floating strike lookback put}, whose payoff function is $\phi(S_T, X_T) = \left( X_T - S_T \right)^+$ (see, for example, \cite[pp. 623--525]{Hull}). Certain barrier options can also be priced by considering the maximum process, for example the \emph{up-and-in barrier call} with barrier $b$ and strike $K$, which at maturity pays out $\phi(S_T, X_T) = \chi_{[b, \infty)}(X_T ) (S_T - K)^+$ (see, for example, \cite[pp. 143--150]{Buchen}). In \cref{sec:sol} we will solve an upper bound MOT problem where the payoff is equal to the terminal value of the maximum. 
\end{example}

\begin{example}[Asian options] \label{ex:examples_mean}
Let the process $X$ be the arithmetic mean of the price process $S$, that is, let $X_t = (t+1)^{-1} \sum_{j \in [t]} S_j$ for $t \in [T]$. Since an equivalent way of defining $X$ is
\begin{equation*} \label{eqn:intro_examples_mean}
\begin{aligned}
    X_t = 
    \begin{cases}
        S_0, & t = 0 \\
        \frac{1}{t+1} S_t + \left(1 - \frac{1}{t+1} \right) X_{t-1}, & t \in [T] \backslash 0,
    \end{cases}
\end{aligned}
\end{equation*}
it is of the form (\ref{eqn:intro_x}). Derivatives whose payoff is a function of the arithmetic mean are Asian options with arithmetic averaging, such as the \emph{average price call} and the \emph{average strike call}. Their payoffs are $\phi(X_T) = (X_T - K)^+ $ and $\phi(S_T, X_T) = (S_T - X_T)^+ $, respectively, where $K >0$ denotes a fixed strike price (see, for example,  \cite[pp. 626--627]{Hull}). In \cref{sec:sol} we will provide an example where we compute the robust fair price of an \emph{Asian straddle}, given two, three and four known marginals. 
\end{example}

\begin{example}[Barrier options] \label{ex:examples_ind}
    Let $A_0, \dots, A_T \subseteq \mathbb{R}$ be given subsets and define a stochastic process via $X_t = \chi_{A_0 \times \dots \times A_t}(S_0, \dots, S_t)$ for $t \in [T]$; this $X$ can be written  
    \begin{equation*} \label{eqn:intro_examples_ind}
    \begin{aligned}
        X_t = 
        \begin{cases}
            \chi_{A_0}(S_0), & t = 0 \\
            \chi_{A_t}(S_t) X_{t-1}, & t \in [T] \backslash 0.
        \end{cases}
    \end{aligned}
    \end{equation*}
    Its terminal value, $X_T$, indicates whether the price $S$ has passed through all the sets $A_0, \dots, A_T$ during the duration of the contract. By setting all but two of the sets $A_0, \dots, A_T$ equal to $\mathbb{R}$, it is possible to consider \emph{second-order binary options}. Indeed, let $T_0 \in [T-1]\backslash 0$ and define $A_{T_0} := \{ x \in \mathbb{R} : x \ge b_0 \}$ and $A_{T} := \{ x \in \mathbb{R} : x \ge b \}$ for some barrier levels $b_0$ and $b$.  Let $A_t = \mathbb{R}$ for $t \in [T] \backslash \{T_0, T\}$. Then the payoff of an up-up type second-order asset binary option can be written $\phi(S_{T}, X_{T}) = S_{T} X_{T}$ (see, for example, \cite[pp. 109--110]{Buchen}). A \emph{double knockout option}, on the other hand, expires worthless if the price of the underlying crosses any of the two barriers $b_l < S_0 < b_u$ (see, for example, \cite[p. 274]{Panjer}). Defining $A := \{ x \in \mathbb{R}: b_l < x < b_u \}$ and setting $A_t = A$ for all $t \in  [T]$ allows us to also  consider such options. In \cref{sec:sol} we will compute the robust price of a \emph{digital option}.
\end{example}

\begin{example}[Variance swaps] \label{ex:examples_var}
    Consider the case when the process $X$ is the realised variance of the price process $S$, that is, when $X_t = t^{-1} \sum_{j=1}^{t} ( \log (  S_{j}/S_{j-1} ) )^2 $ for $t \in [T] \backslash 0$. Set $X_0 := 0 $; then $X$ can be written as
    \begin{equation*}  \label{eqn:intro_examples_var}
    \begin{aligned}
        X_t 
        =
        \begin{cases}
            0, & t= 0 \\
             \big( 1 - \frac{1}{t} \big) X_{t-1} + \frac{1}{t}  \big( \log \big(  \frac{S_t}{S_{t-1}}\big) \big)^2, & t \in [T] \backslash 0,
        \end{cases}
    \end{aligned}
    \end{equation*}
    which shows that the realised variance is indeed of the form (\ref{eqn:intro_x}).
    The payoff of a \emph{variance swap} is then given by $\phi(X_T) = X_T - \sigma_{\text{fixed}}^2$, where $\sigma_{\text{fixed}}^2$ denotes the fixed variance (see, for example, \cite[p. 629]{Hull}). Needless to say, the realised volatility of the price $S$ can easily be obtained as $\sqrt{X_t}$ at any time point $t \in [T]$. We note that the optimal solution of the corresponding MOT problem is known in the case when $T = 1$; it is the left-monotone transport plan \cite{Beiglbock-Juillet2016, HL-T2016}, a solution that will be revisited in \cref{sec:sol}. 
\end{example}

\begin{example}[Parisian options] \label{ex:examples_count}
In a similar fashion to \cref{ex:examples_ind}, the payoff can be defined in terms of the duration during which the process is in a given interval; let $A \subset \mathbb{R}$ and set $X_t = \sum_{j \in [t]} \chi_A (S_t)$ for $t \in [T]$. This process can also be written
\begin{equation*}
\begin{aligned}
    X_t = 
    \begin{cases}
        \chi_A (S_0), & t = 0\\
         \chi_A(S_t) + X_{t-1}, & t \in [T] \backslash 0
    \end{cases}
\end{aligned}
\end{equation*}
and its terminal value, $X_T$, corresponds to the number of, say, days that the price $S$ has spent within $A$ during the duration of the contract, a property that allows for pricing \emph{Parisian options}. In order for such an option --- written with a $D \in \mathbb{N}$ day window --- to activate, the price of the underlying must spend at least $D$ days within $A$ (see, for example, \cite[p. 622]{Hull}). As an example, an up-and-in Parisian put with a $D$ day window, a strike price $K$ and a barrier $b$ has a payoff that can be written $\phi(S_T, X_T ) = (K - S_T)^+ \chi_{ [D, \infty) }(X_T)$,  with $X_T$ defined as above for $A = \{ x \in \mathbb{R} : x \ge b\}$.
\end{example}

\begin{example}[Dual expiry options] \label{ex:examples_memory}
Consider a setting with two times of expiry, $T$ and $T_0 \in [T-1]\backslash 0$, and claims whose payoff depend on the value of the underlying at those two specific times only. By defining the process $X$ as
\begin{equation*}
X_t = 
\begin{cases}
     0,  & t  \in [T_0 -1] \\
    S_{T_0},  & t \in [T ] \backslash  [T_0 -1]
\end{cases}
\end{equation*}
we are able to `remember' the value of the underlying price process $S$ at time $T_0$ throughout the duration of the contract; an equivalent definition is
\begin{equation*} \label{eqn:intro_examples_dualexp}
\begin{aligned}
    X_t 
    =
    \begin{cases}
        0, & t = 0 \\
        X_{t-1} \chi_{ [T] \backslash T_0 }(t) + S_t \chi_{ \{T_0 \}}(t), & t \in [T] \backslash 0 ,
    \end{cases}
\end{aligned}
\end{equation*}
which is of the form (\ref{eqn:intro_x}). This choice of $X$ allows for pricing dual expiry options --- the payoff of a \emph{forward start call} is given by $\phi(S_{T}, X_{T}) = (S_{T} - X_{T})^+$, while a \emph{ratchet call option} with strike price $K$ has a payoff $\phi(S_{T}, X_{T}) = \max \left \{ (S_{T} - K)^+ , (X_{T} - K)^+   \right \} $ (see, for example,  \cite[pp. 107, 119]{Buchen}).
\end{example}

\begin{example}[Cliquet options] \label{ex:examples_caps}
Consider the sum of the truncated relative returns; let $C_l > 0$ be the local cap and let $X_t = \sum_{j=1}^t \max  \{ \min \{ (S_j - S_{j-1})/S_{j-1}, C_l \}, 0 \}$ for $t \in [T] \backslash 0$. Set $X_0 = 0$. Then $X_t$ can be written on the form from \cref{eqn:intro_x}; indeed,
\begin{equation*}
\begin{aligned} \label{eqn:intro_examples_cliq}
    X_t 
    =
    \begin{cases}
        0, & t = 0 \\
        \max \left \{ \min \left \{ \frac{S_t - S_{t-1}}{S_{t-1}}, C_l \right \}, 0  \right \} + X_{t-1}, &  t  \in [T] \backslash 0.
    \end{cases}
\end{aligned}
\end{equation*}
The payoff function of a globally floored, locally capped \emph{cliquet option} is then given by $\phi(X_T) = \max \{X_T, F_g \}$, where $F_g$ is the global floor (see, for example, \cite[p. 379]{Wilmott2005}). 
\end{example}

\subsection{Structured multi-marginal optimal transport} \label{sec:bg_ot}
The OT problem was introduced in \cref{sec:intro} as the problem of minimising the total cost for transporting the mass of one distribution to another one. The transport can occur in one or several steps --- we refer to the former as bi-marginal transport, while the latter is multi-marginal transport. In order to formalise this, consider $T+1$ probability spaces $(\mathcal X_0, \mathcal{F}_{\mathcal X_0},  \mu_0), \dots, (\mathcal X_T, \mathcal{F}_{\mathcal X_T},  \mu_T)$ and let $\mathcal{M}(\mathcal X_j)$ denote the set of probability measures on $\mathcal X_t$ for $t \in [T]$. Similarly, let $\mathcal{M}(\mathcal X)$ denote the set of probability measures on the product space $\mathcal X := \mathcal X_0 \times \dots \times \mathcal X_T$ equipped with the product $\sigma$-algebra. Let the projection operators $\mathcal{P}_t : \mathcal{M}( \mathcal X) \rightarrow \mathcal M(\mathcal X_t)$ be defined as $\mathcal{P}_t (\pi) := \int_{\mathcal X_0} \dots \int_{\mathcal{X}_{t-1}} \int_{\mathcal{X}_{t+1}} \dots \int_{\mathcal{X}_{T}} \mathrm{d} \pi$, for $\pi \in \mathcal M(\mathcal X)$ and $t \in [T]$. A measure $\pi \in \mathcal M(\mathcal X)$ is said to be a coupling of $(\mu_0, \dots, \mu_T)$ if it satisfies $\mathcal{P}_t(\pi) = \mu_t$ for $t \in [T]$; in the context of optimal transport, a coupling is often referred to as a transport plan. Given a cost function $c : \mathcal X \rightarrow \mathbb{R}$, where $c(x_0,\dots, x_T)$ is the cost associated with the path $(x_0,\dots, x_T) \in \mathcal X$, the optimal transport problem is the problem of finding a coupling that minimises the total associated cost, i.e. that attains
\begin{subequations} \label{eqn:bg_multi_ot}
\begin{align}
    \underset{\pi \in \mathcal{M}(\mathcal{X})}{\inf}  \quad &\int_{\mathcal X_0 \times \dots \times \mathcal X_T} c(x_0,\dots, x_T) \mathrm{d} \pi(x_0, \dots, x_T) \label{eqn:bg_multi_ot_objfcn}\\
    \text{subject to} \quad &  \mathcal{P}_t (\pi) = \mu_t , \quad t \in [T]. \label{eqn:bg_multi_ot_margconstr} 
\end{align} 
\end{subequations}
A coupling of $(\mu_0, \dots, \mu_T)$ that attains the minimum in problem (\ref{eqn:bg_multi_ot}) is said to be an optimal coupling or an optimal transport plan. 
Problem (\ref{eqn:bg_multi_ot}) can be generalised by relaxing the constraint (\ref{eqn:bg_multi_ot_margconstr}) to only hold for a subset $\Gamma \subseteq [T]$ of the marginals or by introducing unbalanced problems \cite{BeiervonLindheimNeumayerSteidl2023, ChizatPeyreSchmitzerVialard2018,ElvanderHaaslerJakobssonKarlsson2020} (see also \cite{georgiou2008metrics,karlsson2017generalized, liero2018optimal, piccoli2014generalized} for bi-marginal problems).

In order to solve problems on the form (\ref{eqn:bg_multi_ot}) computationally, we assume that each of the sets $\mathcal X_0, \dots, \mathcal X_T$ is finite, that is, $\cup_{t \in [T]} \mathcal X_t$ is contained within $n \in \mathbb{N}$ points. Then the marginals can be represented by vectors $m_t \in \mathbb{R}^{n}$, $t \in [T]$, and the cost function and the transport plan by tensors $\mathbf{C} \in \mathbb{R}^{n^{T+1}}$ and $\mathbf{Q} \in \mathbb{R}_+^{n^{T+1}}$, respectively. The discrete analogue to problem (\ref{eqn:bg_multi_ot}) is then 
\begin{subequations} \label{eqn:bg_discrete_ot}
\begin{align}
    \underset{\mathbf{Q} \in \mathbb{R}_+^{n^{T+1}}}{\min}  \quad &\langle \mathbf{C}, \mathbf{Q} \rangle \label{eqn:bg_discrete_ot_objfcn} \\
    \text{subject to} \quad &  P_t (\mathbf{Q}) = m_t , \quad t \in [T]. \label{eqn:bg_discrete_ot_margconstr}
\end{align} 
\end{subequations}
Problem (\ref{eqn:bg_discrete_ot}) is a linear programming problem in $n^{T+1}$ variables. Since the number of variables grows exponentially as a function of the number of marginals, the problem quickly becomes intractable for standard methods --- even problems with two marginals ($T=1$) are not straightforward to solve when the number of points $n$ is large. One way to alleviate this difficulty is to regularise the problem with an entropy term, as proposed in \cite{BenamouCarlierCuturiNennaPeyre2015, Cuturi2013, PeyreCuturi2019}; we will see that this leads to
a dual problem in $n(T+1)$ variables and explicit formulas for the optimal choice thereof.

The entropy regularised problem is given by
\begin{subequations} \label{eqn:bg_discrete_ot_reg}
\begin{align}
    \underset{\mathbf{Q} \in \mathbb{R}_+^{n^{T+1}}}{\min}  \quad &\langle \mathbf{C}, \mathbf{Q} \rangle+\varepsilon D(\mathbf{Q}) \label{eqn:bg_discrete_ot_objfcn_reg} \\
    \text{subject to} \quad &  P_t (\mathbf{Q}) = m_t , \quad t \in [T] ,\label{eqn:bg_discrete_ot_margconstr_reg}
\end{align} 
\end{subequations}
where $\varepsilon > 0$ is a regularisation parameter and
\begin{align*}
    D(\mathbf{Q}) := \sum_{i_0, \dots, i_T \in [n-1]}   \mathbf{Q}(i_0, \dots, i_T) \log \mathbf{Q} (i_0, \dots, i_T) - \mathbf{Q}(i_0, \dots, i_T)  
\end{align*}
is an entropy term. It is solved by utilising duality theory. The Lagrangian function corresponding to problem (\ref{eqn:bg_discrete_ot_reg}) reads
\begin{align*}
    \mathcal{L}(\mathbf{Q}, \lambda) = 
    \langle \mathbf{C} , \mathbf{Q} \rangle + \varepsilon D(\mathbf{Q}) 
    + \sum_{t \in [T]}  \lambda_t^{\top} \left( m_t - P_t(\mathbf{Q}) \right),
\end{align*}
where $ \lambda = \{\lambda_t \}_{t \in [T]}$ are Lagrangian multipliers in $\mathbb{R}^{n}$
corresponding to the constraints \eqref{eqn:bg_discrete_ot_margconstr_reg}. The Lagrangian multipliers are also referred to as dual variables, while the dual functional of problem (\ref{eqn:bg_discrete_ot_reg}) is defined as $\varphi(\lambda) := \min \{  \mathcal{L}(\mathbf{Q}, \lambda) : \mathbf{Q} \in \mathbb{R}_+^{n^{T+1}} \}$.

Since the regularisation term $D(\mathbf{Q})$ ensures that the Lagrangian function is strictly convex in $\mathbf{Q}$, the minimum is attained by the unique $\mathbf{Q}>0$ for which the gradient of the Lagrangian function is zero; that is, for the transport plan $\mathbf{Q}$ satisfying
\begin{align*}
0=
\partial_{\mathbf{Q}(i_0,\ldots , i_T)} \mathcal{L}(\mathbf{Q}, \lambda)
=
\mathbf{C}(i_0,\ldots, i_T)+\varepsilon \log \mathbf{Q}(i_0,\ldots, i_T)-\sum_{t \in [T]} \lambda_t(i_t)
\end{align*}
for $i_0, \dots , i_T  \in \{  1, \dots , n \}$. Solving for $\mathbf{Q}$ yields that the stationary point is given by $\mathbf{Q}^{\lambda} = \mathbf{K} \odot \mathbf{U}^{\lambda}$, where 
\begin{equation*}
    \begin{aligned}
        \mathbf{K}(i_0, \dots, i_T) := e^{-\mathbf{C}(i_0, \dots, i_T) / \varepsilon}, 
        \quad 
        \mathbf{U}^{\lambda}(i_0, \dots, i_T) := \prod_{t \in [T]} e^{\lambda_t(i_t)/\varepsilon}, 
        \quad 
        i_0, \dots, i_T = 1, \dots , n.
    \end{aligned}
\end{equation*}
 Inserting this solution $\mathbf{Q}^{\lambda}$ into the Lagrangian function yields the dual functional $\varphi$. The dual of problem \eqref{eqn:bg_discrete_ot_reg} becomes
\begin{align} \label{eqn:struct_dual}
    \max_\lambda  \, \varphi(\lambda)=    \max_\lambda \, \sum_{t\in [T]} \lambda_t^{\top}m_t-\varepsilon \langle \mathbf{K}, \mathbf{U}^{\lambda}\rangle
\end{align}
--- see this by noting that $\langle \mathbf{C}, \mathbf{Q}^{\lambda}\rangle + \varepsilon \langle  \mathbf{Q}^{\lambda} , \log  \mathbf{Q}^{\lambda} \rangle - \sum_{t \in [T]} \lambda_t^{\top} P_t( \mathbf{Q}^{\lambda}) = 0$. Note that the dual variables are unconstrained since the corresponding constraints of the primal problem are equality constraints (see, for example, \cite{Nash-Sofer}). One can show that for a primal-dual pair like the one above, strong duality holds (see, for example, \cite{Nash-Sofer}), meaning that the respective problem values of problems \eqref{eqn:bg_discrete_ot_reg} and \eqref{eqn:struct_dual} coincide. Moreover, the optimal solution of the primal problem \eqref{eqn:bg_discrete_ot_reg} is $\mathbf{Q}^{\lambda^*}$, where $\lambda^*$ is the optimal solution of the dual problem \eqref{eqn:struct_dual}. The optimal solution $\mathbf{Q}^{\lambda^*}$ of the primal problem \eqref{eqn:bg_discrete_ot_reg} can therefore be obtained by first solving the dual problem \eqref{eqn:struct_dual}.

Since the dual functional $\varphi$ is concave, the dual problem is a convex problem. Its maximisers are thus found as stationary points of $\varphi$. 
The optimal $\lambda$ therefore satisfies  $0=\partial_{\lambda_t(i_t)} \varphi(\lambda)=m_t(i_t)-P_t(\mathbf{K} \odot \mathbf{U}^{\lambda})(i_t)$ for $i_t  \in \{ 1, \dots, n \}$ and $t \in [T]$. In order to simplify the expressions, the solution is reparameterised in terms
of $u = \{u_t^{\lambda_t}\}_{t \in [T]}$ where $u_t^{\lambda_t}:=\exp(\lambda_t/\varepsilon)$ for $t\in [T]$. It follows that the tensor $\mathbf{U}^{\lambda}$ can be written
$\mathbf{U}^{\lambda}(i_0, \dots, i_T) = \prod_{\ell \in [T]} u^{\lambda_{\ell}}_\ell(i_\ell)$, for $i_0, \dots , i_T \in \{ 1, \dots, n \}$.
Since $\mathbf{U}^{\lambda}$ is a rank one tensor, it holds that
$P_t(\mathbf{K} \odot \mathbf{U}^{\lambda})=u_t^{\lambda_t} \odot P_t(\mathbf{K} \odot \mathbf{U}^{\lambda}_{-t})$,  where  $\mathbf{U}^{\lambda}_{-t}(i_0, \dots, i_T) := \prod_{\ell\in [T]\backslash t} u^{\lambda_{\ell}}_\ell(i_\ell)$. The optimal $u^\lambda$ therefore satisfies 
\begin{align}\label{eqn:bg_u}
    u_t^{\lambda_t}=m_t \oslash P_t(\mathbf{K} \odot \mathbf{U}^{\lambda}_{-t}), \quad t \in [T]. 
\end{align}

The computational bottleneck of the formula \eqref{eqn:bg_u} is evaluating the projection, which requires manipulating an $n^{T+1}$-tensor. Consequently, the computational complexity increases exponentially as the number of marginals $T+1$ grows \cite{LinHoCuturiJordan2022}, and entropic regularisation alone is thus not enough for efficient computation of the problem. Fortunately, for many such multi-marginal problems the cost function has graph-structures that can be exploited in combination with entropic regularisation \cite{ AltschulerBoix-Adsera2023,BeiervonLindheimNeumayerSteidl2023, BenamouCarlierCuturiNennaPeyre2015, ElvanderHaaslerJakobssonKarlsson2020, HaaslerKarlssonRingh2021, HaaslerRinghChenKarlsson2021, HaaslerSinghZhangKarlssonChen2021,MascherpaHaaslerAhlgrenKarlsson2023}.
These are referred to as \emph{graph-structured multi-marginal optimal transport problems}. For cost functions that decouples sequentially, i.e. are of the form
\begin{equation} \label{eqn:bg_C}
    \mathbf{C}(i_0, \dots, i_T) = \sum_{t \in [T] \backslash 0 } C_t(i_{t-1}, i_t), \quad i_0, \dots, i_T = 1, \dots, n 
\end{equation}
 for some matrices $C_1, \dots , C_T \in \mathbb{R}^{n \times n}$, the projection can be computed efficiently using matrix-vector multiplications. This is described in the following proposition.%
\footnote{In \cite{ElvanderHaaslerJakobssonKarlsson2020}, \cref{prop:algo_filipisabell} was stated in the special case when the matrices $K_t$, $t \in [T] \backslash 0 $, do not depend on the index $t$. However, it is straightforward to generalise it to the version given here.}

\begin{proposition}[\cite{ElvanderHaaslerJakobssonKarlsson2020}, Proposition 2] \label{prop:algo_filipisabell}
    Let the elements of the tensors $\mathbf{K}$ and $\mathbf{U}$ be of the form 
    \begin{equation*}
        \mathbf{K}(i_0, \dots, i_T) = \prod_{t \in [T] \backslash 0   } K_t(i_{t-1}, i_t)
  \qquad \mbox{ and }\qquad
        \mathbf{U}(i_0, \dots, i_T) = \prod_{t \in [T]} u_t(i_t)
    \end{equation*}
    for matrices $K_t \in \mathbb{R}^{n \times n}$, $t \in [T] \backslash 0 $, and vectors $u_t \in \mathbb{R}^n$, $t \in [T]$. Then, 
    \begin{equation*}
        \begin{aligned}
            P_{t} \big(  \mathbf{K} \odot \mathbf{U}  \big) 
            = & \big(  u_0^{\top} K_1 \diag(u_1) K_2 \dots K_{t-1} \diag(u_{t-1}) K_t \big)^{\top} \odot u_t \\ 
            & \odot 
            \big( K_{t+1} \diag(u_{t+1}) K_{t+2} \dots K_{T-1} \diag(u_{T-1}) K_T  u_T \big), \quad t \in [T].
        \end{aligned}
    \end{equation*}
\end{proposition}In \cref{sec:algo} we will use this and similar results to develop an efficient method for solving our class of multi-marginal MOT problems.

\section{Reformulation as a structured OT problem} \label{sec:equiv}
 We now show how an MOT problem of the form (\ref{eqn:intro_mot}) can be identified with a structured multi-marginal OT problem whose optimal solution can be computed efficiently. This is carried out in two steps, where the first one is to reduce the path dependency of the martingale constraint. Once this is done, we will transform the resulting problem to tensor form. A key insight throughout is to leverage on the sequential structure inherent in the problem.

\subsection{Markovian reformulation of the problem}
When comparing the form (\ref{eqn:intro_phi}) of our class of payoff functions with the form (\ref{eqn:bg_C}) of the cost functions characterising the sequentially structured OT problems, we note that they are consistent. However, the presence of the martingale constraint introduces a dependency of each marginal to all previous marginals, ruining the structure of the problem; consequently, our class of MOT problems cannot directly be solved via the techniques described in \cref{sec:bg_ot}, despite the payoff function being of the appropriate form. Instead we consider the modified multi-marginal OT problem

\begin{subequations} \label{eqn:equiv_mot}
\begin{align}
    \underset{(\Omega, \mathcal{F}, \mathbb{Q}, S)}{\inf}  &\mathbb{E}_{\mathbb{Q}} \bigg[ \sum_{t \in [T] \backslash 0 } \phi_{t}(S_{t-1}, X_{t-1}, S_t, X_t) \bigg] && \label{eqn:equiv_mot_objfcn}\\
    \text{subject to} \quad & \mathbb{E}_{\mathbb{Q}}[S_t | \sigma(S_{t-1}, X_{t-1})] = S_{t-1}, && t \in [T] \backslash 0  \label{eqn:equiv_mot_mtgconstr} \\
    & S_t \sim_{\mathbb{Q}} \mu_t, && t \in \mathcal{T} \label{eqn:equiv_mot_margconstr}
\end{align}
\end{subequations}
where $X$ is given in \cref{eqn:intro_x}. Note that since the objective function (\ref{eqn:equiv_mot_objfcn}) and the constraints (\ref{eqn:equiv_mot_mtgconstr}) and (\ref{eqn:equiv_mot_margconstr}) are linear in $\mathbb{Q}$, the above is a convex problem. The key idea is that for every model $(\Omega, \mathcal{F}, \mathbb{Q}, S)$ that satisfies equations (\ref{eqn:equiv_mot_mtgconstr}) and (\ref{eqn:equiv_mot_margconstr}), where the process $X$ is given by \cref{eqn:intro_x}, there exists a Markov process $(\tilde S , \tilde X)$, supported on some probability space $(\tilde{\Omega}, \tilde{\mathcal F}, \tilde{\mathbb{Q}})$, with similar properties as the joint process $(S,X)$; in particular, this $(\tilde{\Omega}, \tilde{\mathcal{F}}, \tilde{\mathbb{Q}}, \tilde{S})$ is feasible to both problems (\ref{eqn:intro_mot}) and (\ref{eqn:equiv_mot}) and it yields the same value of the objective function as the model  $(\Omega, \mathcal{F}, \mathbb{Q}, S)$.

We now show that any optimal solution to problem (\ref{eqn:equiv_mot}) can be used to find an optimal solution to problem (\ref{eqn:intro_mot}).

\begin{theorem}\label{thm:equiv_mots}
    The martingale optimal transport problem  (\ref{eqn:intro_mot}) and the optimal transport problem (\ref{eqn:equiv_mot}) are equivalent in the sense that any optimal solution of problem (\ref{eqn:intro_mot}) is also an optimal solution of problem (\ref{eqn:equiv_mot}), while any optimal solution of problem (\ref{eqn:equiv_mot}) can be used to construct an optimal solution to problem (\ref{eqn:intro_mot}). Moreover, the optimal objective values coincide.
\end{theorem}

\begin{proof}[Proof of \cref{thm:equiv_mots}]
Let a tuple $(\Omega, \mathcal{F}, \mathbb{Q}, S)$ refer to a probability space $(\Omega, \mathcal{F}, \mathbb{Q})$ supporting a stochastic process $S : \Omega \times [T]  \rightarrow \mathbb{R}$. Define two sets of tuples $\mathfrak{F}$ and $\mathfrak{F}'$ as
    \begin{equation*}
        \begin{aligned}
            \mathfrak{F}
            :=
             \{ 
               (\Omega, \mathcal{F}, \mathbb{Q}, S) 
               :   S_t \sim_{\mathbb{Q}} \mu_t \quad \forall t \in \mathcal{T} 
               \quad 
            \text{and} \quad  \mathbb{E}_{\mathbb{Q}}\left[ S_t | \sigma(S_0, \dots, S_{t-1}) \right] = S_{t-1}  \quad \forall t \in [T] \backslash 0  
             \}
        \end{aligned}
    \end{equation*}
    and
    \begin{equation*}
        \begin{aligned}
            \mathfrak{F}'
            :=
             \{ 
               (\Omega, \mathcal{F}, \mathbb{Q}, S) 
               :   S_t \sim_{\mathbb{Q}} \mu_t \quad \forall t \in \mathcal{T} 
               \quad 
            \text{and} \quad  \mathbb{E}_{\mathbb{Q}}\left[ S_t | \sigma(S_{t-1}, X_{t-1}) \right] = S_{t-1}  \quad \forall t \in [T] \backslash 0 
             \},
        \end{aligned}
    \end{equation*}
    respectively, where $X$ is defined out of $S$ as in \cref{eqn:intro_x}. We note that $\mathfrak{F}$ is the feasible set of problem (\ref{eqn:intro_mot}) and that $\mathfrak{F}'$ is the feasible set of problem (\ref{eqn:equiv_mot}). Take any tuple $(\Omega, \mathcal{F}, \mathbb{Q}, S) \in  \mathfrak{F}$. Then, by the Tower property, for $t \in [T] \backslash 0$,
    \begin{equation*}
    \begin{aligned}
        \mathbb{E}_{\mathbb{Q}} \big[ S_t | \sigma( S_{t-1}, X_{t-1} ) \big]
        = S_{t-1}  \quad \text{a.s.}
    \end{aligned}
    \end{equation*}
    It follows that $(\Omega, \mathcal{F}, \mathbb{Q}, S) \in  \mathfrak{F}'$. Therefore $\mathfrak{F}\subseteq \mathfrak{F}'$.

    Take any tuple $(\Omega, \mathcal{F}, \mathbb{Q}, S) \in \mathfrak{F}'$ and define $X$ out of $S$ as given in \cref{eqn:intro_x}. Then note that there exists an $\mathbb{R}^2$-valued Markov process $(\tilde S, \tilde X)$, supported on some probability space $(\tilde{\Omega}, \tilde{\mathcal{F}}, \tilde{\mathbb{Q}})$, that has the same initial distribution and conditional distributions as the joint process $(S,X)$. Indeed, let $\nu : \mathcal{B}(\mathbb{R}^2) \rightarrow [0,1]$ and $\kappa_t : \mathbb{R}^2 \times \mathcal{B}(\mathbb{R}^2) \rightarrow [0,1]$, for $t \in [T] \backslash 0$, be defined by $ \nu(A) := \mathbb{Q} ( (S_0, X_0) \in A )$ and $\kappa_t(y, A) := \mathbb{Q} (  (S_t, X_t) \in A | (S_{t-1}, X_{t-1}) = y )$ for $A \in \mathcal{B}(\mathbb{R}^2)$ and $y \in \mathbb{R}^2$.
    Since $(\mathbb{R}^2, \mathcal{B}(\mathbb{R}^2))$ is a Borel space, the regular conditional distribution given above exists and is almost everywhere equal to a kernel (see, for example, \cite[Theorem 8.5]{Kallenberg2021}). The Kolmogorov extension theorem then guarantees the existence of an $\mathbb{R}^2$-valued Markov process $(\tilde S, \tilde X)$, defined on some probability space $(\tilde{\Omega}, \tilde{\mathcal{F}}, \tilde{\mathbb{Q}})$\footnote{It can be taken as the canonical process on the probability space $(\tilde{\Omega}, \tilde{\mathcal{F}}, \tilde{\mathbb{Q}})$ given by $(\tilde{\Omega}, \tilde{\mathcal{F}}) = (\mathbb{R}^2, \mathcal{B}( \mathbb{R}^2))^{T+1}$ and $\tilde{\mathbb{Q}} = \nu \otimes \kappa_1 \otimes \dots \otimes \kappa_T$.}, with initial distribution $\nu $ and transition kernels $\{\kappa_t\}_{t\in [T] \backslash 0}$ (see, for example, \cite[Theorem 11.4]{Kallenberg2021}). It follows that
    \begin{equation} \label{eqn:mots_equiv_pf1}
        \text{Law}((\tilde{S}_{t-1}, \tilde{X}_{t-1}), (\tilde{S}_{t}, \tilde{X}_{t}))
        =
        \text{Law}(( S_{t-1}, X_{t-1}), (S_{t}, X_{t})), \quad t \in [T] \backslash 0.
    \end{equation}
    Indeed, it is immediate that $\text{Law}((\tilde S_0, \tilde X_0)) = \nu = \text{Law}((S_0, X_0))$. Combining this with the definition of the kernel $\kappa_1$ yields that $\text{Law}( (\tilde{S}_0, \tilde{X}_0), (\tilde{S}_1, \tilde{X}_1)) = \text{Law}( (S_0, X_0), (S_1, X_1) )$, from where it follows that $\text{Law}(  (\tilde{S}_1, \tilde{X}_1)) = \text{Law}((S_1, X_1) )$. \Cref{eqn:mots_equiv_pf1} then follows by induction. Consequently, the law of $\tilde S_t$ under $\tilde{\mathbb{Q}}$ is equal to the law of $S_t$ under $\mathbb{Q}$ for every $t \in [T]$ and therefore $\tilde S_t \sim_{\tilde{\mathbb{Q}}} \mu_t$ for $t \in \mathcal T$. By similar arguments,  
    \begin{equation*}
        \mathbb{E}_{\tilde{\mathbb{Q}}}\big[ \tilde{S}_t | \tilde{S}_{t-1} = s, \tilde{X}_{t-1} = x \big] = \mathbb{E}_{\mathbb{Q}}\big[ S_t | S_{t-1} = s, X_{t-1} = x \big] = s, 
        \quad s,x \in \mathbb{R}, 
        \quad t \in [T] \backslash 0.
    \end{equation*}
     Combining this with the Markov property of $(\tilde S, \tilde X)$ yields that $\tilde S$ is a martingale with respect to its own filtration under $\tilde{\mathbb{Q}}$. Therefore, $(\tilde{\Omega}, \tilde{\mathcal{F}}, \tilde{\mathbb{Q}}, \tilde{S}) \in \mathfrak{F}$. Moreover, it follows from \cref{eqn:mots_equiv_pf1} that 
    \begin{equation*} \label{eqn:mots_equiv_pf2}
        \mathbb{E}_{\tilde{\mathbb{Q}}} \left[ \sum_{t \in [T] \backslash 0} \phi_{t}(\tilde S_{t-1}, \tilde X_{t-1}, \tilde S_t, \tilde X_t) \right]
        =
        \mathbb{E}_{\mathbb{Q}} \left[ \sum_{t \in [T] \backslash 0} \phi_{t}(S_{t-1}, X_{t-1}, S_t, X_t) \right],
    \end{equation*}
    where the left-hand side is the payoff associated with the model $(\tilde{\Omega}, \tilde{\mathcal{F}}, \tilde{\mathbb{Q}}, \tilde{S})$ since the process $(\tilde{S}, \tilde{X})$ satisfies \cref{eqn:intro_x} almost surely. Indeed, recall that the functions $h_t$ for $t \in [T]$ are measurable and hence it follows from \cref{eqn:mots_equiv_pf1} that $\tilde{\mathbb{Q}}(\tilde{X}_t = h_t( \tilde{S}_t, \tilde{S}_{t-1}, \tilde{X}_{t-1} )) = \mathbb{Q}(X_t = h_t(S_t, S_{t-1}, X_{t-1}) )$ for $t \in [T] \backslash 0$, where the right-hand side equals 1, and similarly for the case when $t=0$. We have thus shown that for any $(\Omega, \mathcal{F}, \mathbb{Q}, S) \in  \mathfrak{F}'$ we can find $(\tilde{\Omega}, \tilde{\mathcal{F}}, \tilde{\mathbb{Q}}, \tilde S)  \in  \mathfrak{F}$ such that their corresponding values of the objective function are equal. This implies that the values of the two problems coincide.

    It follows that if $(\Omega^*, \mathcal{F}^*, \mathbb{Q}^*, S^*) \in \mathfrak{F}'$ is an optimal solution of problem (\ref{eqn:equiv_mot}) there exists $(\tilde{\Omega}^*, \tilde{\mathcal{F}}^*, \tilde{\mathbb{Q}}^*, \tilde{S}^*) \in \mathfrak{F}$ that is an optimal solution of problem (\ref{eqn:intro_mot}). Assume on the other hand that $(\Omega^*, \mathcal{F}^*, \mathbb{Q}^*, S^*) \in \mathfrak{F}$ is an optimal solution of problem (\ref{eqn:intro_mot}). Since the two problem values are equal and $\mathfrak{F} \subseteq \mathfrak{F}'$, it immediately follows that then $(\Omega^*, \mathcal{F}^*, \mathbb{Q}^*, S^*)$ is an optimal solution also of problem  (\ref{eqn:equiv_mot}). This completes the proof.
\end{proof}

\begin{remark} \label{remark:equiv_char_FMOT}
    Note that when $(S, X)$ is not a Markov process under $\mathbb{Q}$, it does not necessarily hold that 
    \begin{equation} \label{eqn:equiv_char_FMOT}
        \mathbb{E}_{\mathbb{Q}} [ S_t |\sigma(S_0, \dots, S_{t-1})]
        =
        \mathbb{E}_{\mathbb{Q}} [S_t | \sigma(S_{t-1} , X_{t-1} )] \quad \text{a.s.}, \quad t \in [T] \backslash 0 .
    \end{equation}
    Therefore, in general  $(\Omega, \mathcal{F}, \mathbb{Q}, S) \not \in \mathfrak{F}$ when $(\Omega, \mathcal{F}, \mathbb{Q}, S) \in \mathfrak{F}'$. Consequently $\mathfrak{F} \neq \mathfrak{F}'$. From this it can be noted that $\mathfrak{F}$ can be characterised as the subset of tuples $(\Omega, \mathcal{F}, \mathbb{Q}, S) \in \mathfrak{F}'$ such that \cref{eqn:equiv_char_FMOT} holds. We immediately realise that every tuple $(\Omega, \mathcal{F}, \mathbb{Q}, S) \in \mathfrak{F}'$ that is such that $(S,X)$ is a Markov process under $\mathbb{Q}$ also belongs to $\mathfrak{F}$. On the other hand, note that every tuple in $\mathfrak{F}'$ does not correspond to an $\mathbb{R}^2$-valued Markov process. However, one can show that the multi-marginal MOT problem that is subject to such an additional Markov constraint is equivalent to problem (\ref{eqn:equiv_mot}) in the same way as we have shown that problem (\ref{eqn:intro_mot}) is equivalent to problem (\ref{eqn:equiv_mot}); this means that there are in total three problems that are all equivalent to each other, given that the payoff function is of the form (\ref{eqn:intro_phi}). We choose to solve problem (\ref{eqn:equiv_mot}) which is convex, in contrast to the problem with an explicit Markov constraint. 
\end{remark}

\begin{remark} \label{rem:sol_markov}
    If we know that the optimal solution to an MOT problem is such that the price process is a Markov process, then the process $X$ can be omitted; see this by letting $X_t = S_t$ for $t \in [T]$ in the above results. This reduces the size of the problem. This holds in particular for all bi-marginal problems since a stochastic process defined over two discrete time points trivially is a Markov process. Along the same lines, we emphasise that MOT problems characterised by a payoff of the form $\phi(S_0, \dots , S_T) = \sum_{t \in [T] \backslash 0 } \phi_t(S_{t-1}, S_t)$, $\phi_t : \mathbb{R}^2 \rightarrow \mathbb{R}$, can be shown to be equivalent to an OT problem subject to the constraint $\mathbb{E}_{\mathbb{Q}}[S_t | \sigma(S_{t-1})] = S_{t-1}$ for $t \in [T] \backslash 0$.
\end{remark}

\subsection{Formulation as a linear programming problem} \label{sec:tensors}
In order to develop a computational method, we now consider the case when the support of the given marginals is concentrated on a finite number of points --- this is the situation for information given by a real-world financial market. For practical convenience, we limit ourselves to only consider models where the support of the intermediate marginals is also concentrated on a discrete set. The objective function as well as the constraints of problem (\ref{eqn:equiv_mot}) can then be expressed in terms of tensors; we will now formalise this. With a slight abuse of terminology, we refer to \cref{eqn:equiv_mot_mtgconstr} as `martingale constraints' from now on.

Let a family $\{ \mathcal{S}_t \}_{t \in [T]}$ of discrete subsets of $\mathbb{R}$ be given by $\mathcal{S}_t = \{ s_t^1, \dots, s_t^{n^{\mathcal{S}}_t} \}$ with $n_t^{\mathcal{S}} \in \mathbb{N}$ and $ t \in [T]$,
where $\mathcal{S}_t$ for $t \in \mathcal{T}$ is the support of the given marginal $\mu_t$. We assume that the resulting grid $\mathcal{S}_0 \times \dots \times \mathcal{S}_T$ is such that it is for $t \in [T] \backslash \mathcal{T}$ possible to construct intermediate marginals $\mu_t$, with support $\mathcal{S}_t$, respecting the convex order. Then the set of martingales, supported on $\mathcal{S}_0 \times \dots \times \mathcal{S}_T$, that respects the given marginals is non-empty. Let for $t \in [T]$ the vector $s_t \in \mathbb{R}^{n_t^{\mathcal{S}}}$ be the vector consisting of the elements in $\mathcal S_t$, that is $s_t(i) := s^i_t$ for $i \in  \mathcal{I}^{\mathcal{S}}_t$,
where $\mathcal{I}^{\mathcal{S}}_t := \{ 1, \dots, n^{\mathcal{S}}_t \}$. For each given marginal $\mu_t$ we define a vector $m_t \in \mathbb{R}^{n_t^{\mathcal{S}}}$ where each element is given by 
\begin{equation} \label{eqn:tensors_mu}
    m_t(i) := \mu_t\left( \{ s_t(i) \} \right), \quad i \in  \mathcal{I}_t^{\mathcal{S}}, \quad t \in \mathcal{T}.
\end{equation}

Let for $t \in [T] \backslash 0$ the set $\mathcal{X}_t$ be such that $h_t(s_t, s_{t-1}, x_{t-1}) \in \mathcal{X}_t$ for $(s_t, s_{t-1}, x_{t-1}) \in \mathcal{S}_t \times \mathcal{S}_{t-1} \times \mathcal{X}_{t-1}$ and $\mathcal{X}_0$ be such that $h_0(s_0) \in \mathcal{X}_0$ for $s_0 \in \mathcal{S}_0$. Then each $\mathcal{X}_t$ is of the form $\mathcal{X}_t = \{ x_t^1, \dots, x_t^{n_t^{\mathcal{X}}} \}$ for some $n_t^{\mathcal{X}} \in \mathbb N$. Let for $t \in [T]$ the vector $x_t \in \mathbb{R}^{n_t^{\mathcal{X}}}$ be the vector consisting of the elements of $\mathcal X_t$, that is $ x_t(i) := x_t^i$ for $i \in \mathcal{I}_t^{\mathcal{X}}$,
where $\mathcal{I}_t^{\mathcal{X}} :=  \{ 1, \dots, n_t^{\mathcal{X}} \}$. Note that depending on the choice of $h_t$, $t \in [T]$, we could have that $\mathcal{X}_t \equiv \mathcal{S}_t$ --- this is the case for Examples \ref{ex:examples_max} and \ref{ex:examples_memory}. We emphasise that the joint state space $\mathcal{S}_t \times \mathcal{X}_t$ contains in total $n_t = n_t^{\mathcal{S}} n_t^{\mathcal{X}}$ states and is indexed by joint indices $(i_t^{\mathcal{S}}, i_t^{\mathcal{X}}) \in\mathcal{I}_t^{\mathcal{S}} \times \mathcal{I}_t^{\mathcal{X}}$. By letting $\mathcal{I}_t := \{1, \dots , n_t\}$ and ordering the joint index set  $\mathcal{I}_t^{\mathcal{S}} \times \mathcal{I}_t^{\mathcal{X}}$ according to
\begin{equation} \label{eqn:tensors_idxorder}
    (1, 1), \dots , (n_t^{\mathcal{S}}, 1), (1,2), \dots, (n_t^{\mathcal{S}}, 2), {\dots} \hspace{0.1cm}{\dots}, (1, n_t^{\mathcal{X}}), \dots, (n_t^{\mathcal{S}}, n_t^{\mathcal{X}}),
\end{equation}
we can identify each element $i_t \in \mathcal{I}_t$ with an element of $(i_t^{\mathcal{S}}, i_t^{\mathcal{X}}) \in \mathcal{I}_t^{\mathcal{S}} \times \mathcal{I}_t^{\mathcal{X}}$; we will write $i_t$ for short for $i_t(i_t^{\mathcal{S}}, i_t^{\mathcal{X}})$. Similarly, it is understood that $i_t^{\mathcal{S}} \in \mathcal{I}_t^{\mathcal{S}}$ and $i_t^{\mathcal{X}} \in \mathcal{I}_t^{\mathcal{X}}$ when $i_t(i_t^{\mathcal{S}}, i_t^{\mathcal{X}}) \in \mathcal{I}_t$. For notational convenience we also define $\mathcal{I} := \bigtimes_{t \in [T]} \mathcal{I}_t$, thus  $(i_0, \dots , i_T) \in \mathcal{I}$ indicates that $i_t \in \mathcal{I}_t$ for $t \in [T]$.

We will now see how the payoff can be represented on tensor form within this framework. Let  $\Phi_t \in \mathbb{R}^{n_{t-1} \times n_t}$ for $t \in [T ] \backslash 0$ be given by 
\begin{align*}
    \Phi_t(i_{t-1}, i_t)
    :=
    \phi_t\big(  s_{t-1}(i_{t-1}^{\mathcal{S}}), x_{t-1}(i_{t-1}^{\mathcal{X}}), s_t(i_t^{\mathcal{S}}), x_t(i_t^{\mathcal{X}}) \big)
    , \quad i_{t-1} \in \mathcal{I}_{t-1}, i_t \in \mathcal{I}_t
\end{align*}
and define $\mathbf{\Phi} \in \mathbb{R}^{n_0 \times \dots \times n_T}$ as $\mathbf{\Phi} (i_0, \dots, i_T) := \sum_{t \in [T] \backslash 0} \Phi_t(i_{t-1}, i_t)$ for $(i_0, \dots i_T) \in \mathcal{I}$ analogously to \cref{eqn:intro_phi}. Then $\mathbf{\Phi}(i_0, \dots, i_T) = \phi ( s_0(i_0^{\mathcal{S}}), \dots, s_T(i_T^{\mathcal{S}}) )$ for $(i_0, \dots , i_T) \in \mathcal{I}$ and represents the payoff associated with the price evolution $(s_0(i_0^{\mathcal{S}}), x_0(i_0^{\mathcal{X}})), \dots, (s_T(i_T^{\mathcal{S}}), x_T(i_T^{\mathcal{X}}))$. In order to enforce a zero probability for every index tuple that does not respect \cref{eqn:intro_x}, let $\mathbf{C} \in \mathbb{R}^{n_0 \times \dots \times n_T}$ be a cost given by $\mathbf{C}(i_0, \dots, i_T) := \sum_{t \in [T] \backslash 0} C_t(i_{t-1}, i_t)$ for $(i_0, \dots, i_T) \in \mathcal{I}$, with
\begin{equation*}
    \begin{aligned}
        C_t (i_{t-1}, i_t)
        :=
        \begin{cases}
            \Phi_t (i_{t-1}, i_t)  & \text{if  } x_t(i_t^{\mathcal{X}}) = h_t\big(s_t(i_t^{\mathcal{S}}), s_{t-1}(i_{t-1}^{\mathcal{S}}), x_{t-1}(i_{t-1}^{\mathcal{X}})\big) \\
            \infty & \text{else}
        \end{cases} \\
    \end{aligned}
    , \quad t = 2, \dots, T 
\end{equation*}
and
\begin{equation*}
    \begin{aligned}
    C_1(i_{0}, i_1)
        :=
        \begin{cases}
            \Phi_1 (i_0, i_1) & \text{if  } x_0(i_0^{\mathcal{X}}) = h_0\big(s_0(i_0^{\mathcal{S}})\big) \text{ and } 
                          x_1(i_1^{\mathcal{X}}) = h_1\big(s_1(i_1^{\mathcal{S}}), s_0(i_0^{\mathcal{S}}), x_0(i_0^{\mathcal{X}})\big) \\
            \infty & \text{else}
        \end{cases} 
    \end{aligned}
\end{equation*}
being matrices in $\mathbb{R}^{n_{t-1} \times n_t}$ and $\mathbb{R}^{n_0 \times n_1}$, respectively. By defining the cost $\mathbf{C}$ in this way, we penalise forbidden index tuples; note that $\langle \mathbf{C} , \mathbf{Q} \rangle = \langle \mathbf{\Phi} , \mathbf{Q} \rangle$ whenever $\langle \mathbf{C} , \mathbf{Q} \rangle < \infty$. Also note that $\mathbf{C}$ inherits the structure of $\mathbf{\Phi}$ and therefore it decouples sequentially, as in \cref{eqn:bg_C}.

We now present some notation needed for formulating problem (\ref{eqn:equiv_mot}) as a linear programming problem. Define for $t \in [T] \backslash 0 $ the matrix $\Delta_t \in \mathbb{R}^{n_{t-1} \times n_t}$ as
\begin{equation} \label{eqn:tensors_delta}
    \Delta_t(i_{t-1},i_t) :=  s_t(i_t^{\mathcal{S}}) - s_{t-1}(i_{t-1}^{\mathcal{S}})
    , \quad i_{t-1} \in \mathcal{I}_{t-1}, i_t \in \mathcal{I}_t,
\end{equation}
where we recall the connection between $i_t$ and $i_t^{\mathcal{S}}$ from (\ref{eqn:tensors_idxorder}). Recall the definitions of the projection operators $P_{t}$, for $t \in [T]$, and $P_{t_1, t_2}$, for $t_1, t_2 \in [T]$ and note that they can be identified with maps from $\mathbb{R}_+^{n_0 \times \dots \times n_T}$ to $\mathbb{R}_+^{n_t}$ and $\mathbb{R}_+^{n_{t_1} \times n_{t_2}}$, respectively. Then define for $t \in [T]$ a family of projection operators $P^{\mathcal{S}_t} : \mathbb{R}_+^{n_t} \rightarrow \mathbb{R}_+^{n_t^{\mathcal{S}}}$ via 
\begin{equation*}
    P^{\mathcal{S}_t}(m)(j) := \sum_{k \in \mathcal{I}_t^{\mathcal{X}}}m \big( i_t(j, k )\big), \quad j \in \mathcal{I}_t^{\mathcal{S}}, \quad  m \in \mathbb R_+^{n_t}
\end{equation*}
 --- then $P^{\mathcal{S}_t}$ is the discrete analogue to integrating over the $\mathcal{X}_t$-component of a function (or distribution) on $\mathcal{S}_t \times \mathcal{X}_t$. Finally, let $P^{\mathcal{S}}_t := P^{\mathcal{S}_t} \circ P_t$ for $t \in [T]$.

Given the framework introduced above, consider the linear programming problem 
\begin{subequations} \label{eqn:tensors_mot_unreg}
\begin{align}
    \underset{\mathbf{Q} \in \mathbb{R}_+^{n_0 \times \dots \times n_T}}{\min} \quad 
    & \langle \mathbf{C} , \mathbf{Q} \rangle && \label{eqn:tensors_mot_unreg_objfcn} \\
    \text{subject to} \quad   
            & \left( P_{t-1, t}(\mathbf{Q}) \odot \Delta_t \right) 1_{n_t} =  0_{n_{t-1}}, && t \in [T] \backslash 0 \label{eqn:tensors_mot_unreg_mtgconstr} \\
    & P_t^{\mathcal{S}}(\mathbf{Q}) = m_t, && t \in \mathcal{T} . \label{eqn:tensors_mot_unreg_margconstr}
\end{align}
\end{subequations}
Note that it has an optimal solution since its feasible set is bounded and non-empty, but that it may not be unique. The next result shows that that problems (\ref{eqn:equiv_mot}) and (\ref{eqn:tensors_mot_unreg}) are equivalent --- the proof is deferred to \cref{sec:appendix}.

\begin{proposition} \label{prop:equiv_lp}
     Suppose that we restrict problem (\ref{eqn:equiv_mot}) to models such that the support of the price process at each time point $t \in [T]$ is $\mathcal{S}_t$. Then problems (\ref{eqn:equiv_mot}) and (\ref{eqn:tensors_mot_unreg}) are equivalent.
\end{proposition}

\section{Solving the structured problem via regularisation}\label{sec:algo}
We have now arrived at a version of the problem that can be approached computationally and we therefore proceed by deriving a generalisation of the algorithm described in \cref{sec:bg_ot}, modified to also take the martingale constraint into account. As in \cref{sec:bg_ot}, this is done by adding a regularising entropy term\footnote{Our definition of the discrete entropy $D(\mathbf{Q})$ corresponds to using the uniform measure as reference measure in the Kullback-Leibler divergence. Other choices of such measures could be made though --- since the reference measure can be interpreted as the distribution of some reference stochastic process, this would be natural in a context where such a process is available. Martingale couplings that minimise the Kullback-Leibler divergence with respect to this distribution would then be favored. However, using another reference measure would affect the objective function of the regularised problem (\ref{eqn:tensors_mot}) only slightly when the regularisation parameter $\varepsilon$ is small. Indeed, let the tensor associated with a general reference measure in the case of a finite state space be denoted by $\mathbf{R}$. Then it is straightforward to show that the new objective function of the regularised problem can be written $\langle \mathbf{C} - \varepsilon \log \mathbf{R}, \mathbf{Q} \rangle + \varepsilon D(\mathbf{Q})$.} $D(\mathbf{Q}) = \langle \log (\mathbf{Q}) - \mathbf{1}_{n_0 \times \dots \times n_T}, \mathbf{Q} \rangle $, scaled by some small number $\varepsilon > 0$, to the objective function of problem (\ref{eqn:tensors_mot_unreg}). The entropy regularised version of problem (\ref{eqn:equiv_mot}) written on tensor form thus becomes
\begin{subequations} \label{eqn:tensors_mot}
\begin{align}
    \underset{\mathbf{Q} \in \mathbb R_+^{n_0 \times \dots \times n_T}}{\min}  \quad
    &\langle \mathbf{C} , \mathbf{Q} \rangle + \varepsilon D(\mathbf{Q}) && \label{eqn:tensors_mot_onjfcn} \\
    \text{subject to} 
    \quad & \left( P_{t-1, t}(\mathbf{Q}) \odot \Delta_t \right) 1_{n_t}
        =  0_{n_{t-1}},
        && t \in [T] \backslash 0 \label{eqn:tensors_mot_mtgconstr} \\
    & P_t^{\mathcal{S}}(\mathbf{Q}) = m_t, && t \in \mathcal{T} , \label{eqn:tensors_mot_margconstr}
\end{align}
\end{subequations}
where we recall that the tensor $\mathbf{C}$ is of the form $\mathbf{C}(i_0, \dots , i_T) = \sum_{[T] \backslash 0} C_t(i_{t-1}, i_t)$, $(i_0, \dots , i_T) \in \mathcal{I}$,  for matrices $C_t \in \mathbb{R}^{n_{t-1} \times n_t}$. Note that a minimum exists for the above problem, since the objective function is continuous in $\mathbf{Q}$ and the feasible set is compact\footnote{The marginal constraints ensures that the feasible set of problem (\ref{eqn:tensors_mot}) is bounded. Compactness then follows from the continuity of the constraints (\ref{eqn:tensors_mot_mtgconstr}) and (\ref{eqn:tensors_mot_margconstr}).}, and that it is unique since the problem is strictly convex. We now show convergence of optimal solutions of the regularised problem as the regularisation parameter vanishes. The proof can be found in \cref{sec:appendix}. It follows that a solution of problem (\ref{eqn:tensors_mot}) serves as an approximate solution of problem (\ref{eqn:tensors_mot_unreg}) when $\varepsilon$ is small.

\begin{proposition} \label{prop:convergence}
    Let $(\varepsilon_k)_k$ be a decreasing sequence of positive regularisation parameters such that $\lim_{k \rightarrow \infty} \varepsilon_k = 0$ and let $\mathbf{Q}_k$ denote the optimal solution of  the regularised problem (\ref{eqn:tensors_mot}) with $\varepsilon = \varepsilon_k$. Then the sequence $(\mathbf{Q}_k)_k$ of minimisers has at least one convergent subsequence and every limit point of $(\mathbf{Q}_k)_k$ is a minimiser of the linear programming problem (\ref{eqn:tensors_mot_unreg}). Moreover, the value of the regularised problem (\ref{eqn:tensors_mot}) converges to the value of problem (\ref{eqn:tensors_mot_unreg}) as the regularisation parameter vanishes.
\end{proposition}

Next, we will see that strong duality holds for problem (\ref{eqn:tensors_mot}), and that under certain additional assumptions on the given marginals a dual maximiser exists and problem (\ref{eqn:tensors_mot}) can then be solved by considering the dual problem. We will also see that the optimality conditions for the dual variables correspond to relatively simple equations, where explicit expressions can be obtained for the dual variables corresponding to marginal constraints, while for the dual variables corresponding to martingale constraints it turns out that $T$ sets of equations must be solved numerically. This is a generalisation of the strategy used by \cite{deMarch2018} for the bi-marginal MOT problem. We then exploit the structure inherent in the problem to simplify the equations obtained; the resulting algorithm is a multi-marginal version of Sinkhorn's algorithm where some of the variables are updated by use of Newton's method.

\subsection{Strong duality for the regularised problem} \label{sec:algo_prel}
In order to provide the dual of problem (\ref{eqn:tensors_mot}), we introduce some new notation that allows for writing the expressions in a more compact format reducing the number of direct applications of the exponential function later on. Therefore, let $\lambda = \{ \lambda_t \}_{t \in \mathcal{T}}$ and $\gamma = \{ \gamma_t \}_{t \in [T-1]}$ be families of vectors with  $\lambda_t \in \mathbb{R}^{n_t^{\mathcal{S}}}$ and $\gamma_t \in \mathbb{R}^{n_t}$, respectively. Given $\lambda, \gamma$ and the cost matrices $C_t$ for $t \in [T] \backslash 0$, define another family $u^{\lambda} = \{ u_t^{\lambda_t} \}_{t \in \mathcal{T}}$ of vectors with $u_t^{\lambda_t} \in \mathbb{R}^{n_t^{\mathcal{S}}}$ given by
\begin{equation}  \label{eqn:algo_u}
    \begin{aligned}
        u_t^{\lambda_t}(i_t^{\mathcal{S}}) := e^{\lambda_t(i_t^{\mathcal{S}})/ \varepsilon}
    , \quad i_t^{\mathcal{S}} \in \mathcal{I}_t^{\mathcal{S}}
    \end{aligned}
\end{equation}
and let $K_t, G^{\gamma_{t-1}}_t \in \mathbb{R}^{n_{t-1} \times n_t}$ for $t \in [T] \backslash 0$ be two families of matrices, defined by
\begin{equation}  \label{eqn:algo_G}
    K_t(i_{t-1}, i_t) := e^{-C_t(i_{t-1}, i_t)/ \varepsilon}, 
    \quad 
    G^{\gamma_{t-1}}_t(i_{t-1}, i_t):= e^{ \gamma_{t-1}(i_{t-1})\Delta_t(i_{t-1}, i_t) / \varepsilon } ,
\end{equation}
for $i_{t-1} \in \mathcal{I}_{t-1}, i_t \in \mathcal{I}_t$. Let $\mathbf{K}, \mathbf{U}^{\lambda}, \mathbf{G}^{\gamma} \in \mathbb{R}^{n_0 \times \dots \times n_T}$ be tensors defined by
\begin{equation} \label{eqn:algo_tensors_KUG}
\begin{aligned}
    & \mathbf{K}(i_0, \dots , i_T) := \prod_{t \in [T] \backslash 0} K_t(i_{t-1}, i_t),
    && \mathbf{U}^{\lambda}(i_0, \dots , i_T) := \prod_{t \in \mathcal T} u_t^{\lambda_t}(i_t^{\mathcal{S}}),
    \\
    & \mathbf{G}^{\gamma}(i_0, \dots , i_T) := \prod_{t \in [T] \backslash 0} G^{\gamma_{t-1}}_t(i_{t-1}, i_t), 
    && (i_0, \dots , i_T) \in \mathcal{I},
\end{aligned}
\end{equation}
and also define, for $j \in \mathcal{T}$ and $k \in  [T] \backslash 0$,
\begin{equation*}
    \mathbf{U}^{\lambda}_{-j}(i_0, \dots , i_T) := \prod_{t \in \mathcal T \backslash j } u_t^{\lambda_t}(i_t^{\mathcal{S}}),  
    \quad 
    \mathbf{G}^{\gamma}_{-k}(i_0, \dots , i_T) := \prod_{t \in [T] \backslash \{ 0, k \}} G^{\gamma_{t-1}}_t(i_{t-1}, i_t), \quad (i_0, \dots , i_T) \in \mathcal{I}.
\end{equation*}
Using this notation we now state our duality result. In order to show existence of dual optimal solutions, we require that every pair $(\mu_{t_1}, \mu_{t_2})$ of two consecutive given marginals with $t_1 < t_2$ is \emph{irreducible} according to \cite[Definition A.3]{Beiglbock-Juillet2016}; in our discrete setting this is equivalent to $|s_{t_2} - z 1_{n^{\mathcal{S}}_{t_2}}|^{\top}m_{t_2} - |s_{t_1} - z 1_{n^{\mathcal{S}}_{t_1}}|^{\top}m_{t_1}  > 0$ for $z \in \mathcal{S}_{t_1} \cup \mathcal{S}_{t_2}^0$, where $\mathcal{S}_{t_2}^0$ denotes $\mathcal{S}_{t_2}$ without its largest and smallest elments.

\begin{theorem}\label{thm:algo}
    A dual of problem (\ref{eqn:tensors_mot}) is given by 
    \begin{align}\label{eqn:dualproblem} 
        \sup_{\lambda, \gamma } \quad \sum_{t \in \mathcal{T}} \lambda_t^{\top} m_t - \varepsilon \langle \mathbf{K}, \mathbf{U}^{\lambda} \odot \mathbf{G}^{\gamma} \rangle,
    \end{align}
    where $\lambda$ and $\gamma$ are families of vectors of the above given form, and the values of the two problems are equal --- that is, strong duality holds.

    Moreover, suppose that the given marginals $\{\mu_t\}_{t \in \mathcal{T}}$ are such that there exist probability measures $\mu_t$ with support $\mathcal{S}_t$ for $t \in [T] \backslash \mathcal{T}$ such that every pair $(\mu_{t-1}, \mu_t)$ for $t \in [T] \backslash 0$ is irreducible. Then the supremum in the dual problem (\ref{eqn:dualproblem}) is attained. Any dual maximiser $(\lambda^*, \gamma^*)$ is given as the solution of the equations 
        \begin{subequations} \label{eqn:algo_update1}
            \begin{align}
                & u_t^{\lambda_t}
                = 
                m_t \oslash P_t^{\mathcal{S}}(\mathbf{K} \odot \mathbf{U}^{\lambda}_{-t} \odot \mathbf{G}^{\gamma}) , && t \in \mathcal{T}, \label{eqn:algo_updateu1} \\
                & \big(
            P_{t,t+1}( \mathbf{K} \odot \mathbf{U}^{\lambda} \odot \mathbf{G}^{\gamma} ) \odot \Delta_{t+1} \big) 1_{n_{t+1}} = 0_{n_t} , && t \in [T-1], \label{eqn:algo_updategamma1}
            \end{align}
        \end{subequations} 
        and the unique martingale transport plan that minimises problem (\ref{eqn:tensors_mot}) is given by
            \begin{equation} \label{eqn:algo_Popt2}
                \mathbf{Q}^{\lambda^*, \gamma^*} = \mathbf{K} \odot \mathbf{U}^{\lambda^*} \odot \mathbf{G}^{\gamma^*} ,
            \end{equation}
        for any optimal dual variables $(\lambda^*, \gamma^*)$. It assigns a non-zero probability mass to every trajectory that respects \cref{eqn:intro_x}, that is, $\mathbf{Q}^{\lambda^*, \gamma^*}(i_0, \dots , i_T) > 0$ for $(i_0, \dots , i_T) \in \mathcal{I}$ such that $\mathbf{K}(i_0, \dots , i_T) > 0$.
\end{theorem}

We refer to \cref{eqn:algo_update1} as \emph{the dual optimality conditions of problem (\ref{eqn:tensors_mot})}. In order to prove the second part of \cref{thm:algo} we need the following lemma. The proofs of \cref{thm:algo} and \cref{lemma:algo_interiorpt} are deferred to \cref{sec:appendix}.

\begin{lemma} \label{lemma:algo_interiorpt}
    Suppose that the grid $\bigtimes_{t \in [T]} \mathcal{S}_t$ and the given marginals $\{ \mu_t \}_{t \in \mathcal{T}}$ satisfy the assumptions of the second part of \cref{thm:algo}. Then the domain of problem (\ref{eqn:tensors_mot}) has a non-empty relative interior, that is, Slater's condition holds. 
\end{lemma}

We remark that the second part of \cref{thm:algo} is formulated under additional assumptions on the grid and on the given marginals compared to what have been imposed earlier in this article. We emphasise that it is always possible to construct intermediate marginals as described in \cref{thm:algo} whenever each consecutive pair of marginals from the given family $\{ \mu_t \}_{t \in \mathcal{T}}$ is irreducible and when for $t \in [T] \backslash 0$ the convex hull of the set $\mathcal{S}_{t-1}$ is a proper subset of the convex hull of $\mathcal{S}_t$, given that the grid is fine enough. Indeed, one can then do so by forming convex combinations of the given marginals and moving some of the probability masses. In fact, we could in theory relax the irreducibility assumption on the given marginals by noting that this would result in trajectories in the path space $\mathcal{S}_0 \times \dots \times \mathcal{S}_T$ that are of zero-support for all martingale transport plans \cite[pp. 3044--3045]{BeiglbockNutzTouzi2017} and thus simply omit the representation of such trajectories from the problem. This would however be unpractical to implement programmatically, for what reason we have chosen the approach given here.

At this point, we can formulate a high-level method for solving problem (\ref{eqn:tensors_mot}). It is summarised in \cref{algo:highlevel}.
\begin{algorithm}
\begin{algorithmic}[h!] 
\State Initialise:  $u_t^{\lambda_t} \gets 1_{n_t^{\mathcal{S}}}$  for $t \in  [T ]$ 
\State \phantom{Initialise:} $\gamma_t \gets 1_{n_t}$ for $t \in  [ T-1 ]$
\While{not converged} 
    \For{ $t \in \mathcal{T} $ }
            \State $u_t^{\lambda_t} \gets m_t 
                                \oslash 
                                P_t^{\mathcal{S}}
                                \left(
                                     \mathbf{K} \odot \mathbf{U}^{\lambda}_{-t} \odot \mathbf{G}^{\gamma}
                                \right) $
    \EndFor 
    \For{ $t = T-1, \dots, 0$ }
        \State find $\hat{\gamma}_t$ such that $( P_{t, t+1}(\mathbf{K} \odot \mathbf{U}^{\lambda} \odot \mathbf{G}^{\gamma}_{-(t+1)} ) \odot G_{t+1}^{\hat{\gamma}_t} \odot \Delta_{t+1} )1_{n_{t+1}} = 0_{n_t}$
        \State $\gamma_t \gets \hat{\gamma}_t$ 
    \EndFor 
\EndWhile 
\State $\mathbf{Q}^{\lambda, \gamma} \gets \mathbf{K} \odot \mathbf{U}^{\lambda} \odot \mathbf{G}^{\gamma}$ \\
\Return{$\mathbf{Q}^{\lambda, \gamma}$}
\end{algorithmic}
\caption{High-level method for solving problem (\ref{eqn:tensors_mot}).}
\label{algo:highlevel}
\end{algorithm}

\begin{remark}
    We conclude this section with a note on the connection to robust hedging duality. The linear programming problem (\ref{eqn:tensors_mot_unreg}) has a dual
    \begin{subequations} \label{eqn:tensors_mot_unreg_dual}
    \begin{align}
        \underset{\lambda, \gamma }{\max}  \quad
        & \sum_{t \in \mathcal{T}} \lambda_t^{\top} m_t && \label{eqn:tensors_mot_unreg_dual_objfcn} \\
        \text{subject to} 
        \quad & \sum_{t \in [T] \backslash 0} \gamma_{t-1}(i_{t-1}) \Delta_t(i_{t-1}, i_t)
        + \sum_{t \in \mathcal{T}}\lambda_t(i_t^{\mathcal{S}}) 
        \le \mathbf{C}(i_0, \dots, i_{T})  , && (i_0, \dots , i_T) \in \mathcal{I}, \label{eqn:tensors_mot_unreg_dual_subhedge} 
    \end{align}
    \end{subequations}
    that can be viewed as the robust subhedging problem of a derivative whose payoff is given by the tensor $\mathbf{C}$. Note that an optimal solution to problem (\ref{eqn:tensors_mot_unreg_dual}) exists, since primal boundedness and feasibility implies dual feasibility and boundedness via strong linear programming duality. The dual variables $(\lambda, \gamma)$ then has an interpretation as a hedging strategy, where the variables $\lambda$ corresponds to static trading in European options and the variables $\gamma$ corresponds to dynamic trading in the underlying: the element $\lambda_t(j)$, $j \in \mathcal{I}_t^{\mathcal{S}}$, provides the payoff of European options expiring at time $t \in \mathcal{T}$ when the price of the underlying at time $t$ is $s_t(j)$ and the element $\gamma_{t}(k)$ for $k \in \mathcal{I}_t$ provides the position in the underlying held between time $t \in [T-1]$ and time $t+1$, respectively. For the case of continuous marginal constraints, the corresponding duality was established by \cite{DolinskySoner2014} for problems with a trivial initial marginal and a general terminal marginal in a continuous-time setting and by \cite{CheriditoKupperTangpi2017} for problems written on multiple assets with several given marginals in a discrete-time setting. See also \cite{BeiglbockNutzTouzi2017, ChereditoKiiskiPromelSoner2021, GuoTanTouzi2017, HouObloj2018} for further work along these lines; see in particular \cite{Beiglbock-HL-Penkner2013, NutzStebeggTan2020} for problems with multiple given marginals.

    The dual (\ref{eqn:dualproblem}) of the entropy-regularised problem (\ref{eqn:tensors_mot}) has no constraints --- the subhedging constraint (\ref{eqn:tensors_mot_unreg_dual_subhedge}) of problem (\ref{eqn:tensors_mot_unreg_dual}) is incorporated in the objective function of problem (\ref{eqn:dualproblem}). It follows from \cref{prop:convergence} and strong duality that the value of problem (\ref{eqn:dualproblem}) converges to the value of problem (\ref{eqn:tensors_mot_unreg_dual}) as the regularisation parameter $\varepsilon$ vanishes, however, showing convergence of optimal dual variables is less straightforward. It is our belief that an optimal solution of problem (\ref{eqn:dualproblem}) serves as an approximation of an optimal subheding strategy of problem (\ref{eqn:tensors_mot_unreg_dual}) for $\varepsilon > 0$ small, but the matter should be investigated further.     
\end{remark}

\subsection{Exploiting the structure for computing the projections}
The computational bottleneck of \cref{algo:highlevel} is the evaluation of the projections; this is because the number of elements in the tensors grows exponentially in the number of marginals. A key ingredient for efficiently computing the optimal dual variables is therefore to reduce the work required for this part. We will now show how this can be done by exploiting the structure inherent in the problem, hence simplifying equations (\ref{eqn:algo_updateu1}) and (\ref{eqn:algo_updategamma1}).

Our key result is summarised in the following theorem.

\begin{theorem} \label{thm:algo_formulas}
     Let $\lambda=\{\lambda_t\}_{t\in \mathcal T}$ and  $\gamma=\{\gamma_t\}_{t \in [T-1]}$ be given families of vectors with $\lambda_t \in \mathbb{R}^{n_t^{\mathcal{S}}}$ and $\gamma_t \in \mathbb{R}^{n_t}$, respectively. Let $u_t^{\lambda_t}=\exp(\lambda_t/\varepsilon)$ for $t\in {\mathcal T}$ and let $K_t=\exp(-C_t/\varepsilon)$ and $G^{\gamma_{t-1}}_t=\exp(\diag(\gamma_{t-1})\Delta_{t}/\varepsilon)$, where the matrices $C_t$ are as in problem (\ref{eqn:tensors_mot}) and the matrices $\Delta_{t}$ are as given in \cref{eqn:tensors_delta}, for $t \in [T] \backslash 0$. Define two families of vectors, $\hat{\psi} = \{\hat{\psi}_t\}_{t\in [T]}$ and $\psi = \{\psi_t\}_{t\in [T]}$,  via the recursions
    \begin{equation} \label{eqn:algo_formulas_f}
        \begin{aligned}
            \hat{\psi}_t = 
            \begin{cases}
                1_{n_0} , & t = 0\\
                ( K_t \odot G^{\gamma_{t-1}}_t )^{\top} ( \hat{\psi}_{t-1} \odot (1_{n_{t-1}^{\mathcal{X}}} \otimes u_{t-1}^{\lambda_{t-1}})  ), &  (t-1)  \in \mathcal T, \quad t \in [T] \backslash 0  \\
                ( K_t \odot G^{\gamma_{t-1}}_t )^{\top} \hat{\psi}_{t-1}, & (t-1)  \not \in \mathcal T, \quad   t \in [T] \backslash 0
            \end{cases}
        \end{aligned}
    \end{equation}
    and
    \begin{equation} \label{eqn:algo_formulas_b}
            \begin{aligned}
                \psi_t = 
                \begin{cases}
                    1_{n_T} , & t = T\\
                    ( K_{t+1} \odot G^{\gamma_t}_{t+1} ) ( \psi_{t+1} \odot (1_{n_{t+1}^{\mathcal{X}}} \otimes u_{t+1}^{\lambda_{t+1}})  ), &  (t+1) \in \mathcal T, \quad t \in [T-1] \\
                    ( K_{t+1} \odot G^{\gamma_t}_{t+1} ) \psi_{t+1}, & (t+1)  \not \in \mathcal T, \quad t \in [T-1].
                \end{cases}
            \end{aligned}
        \end{equation} 
     Suppose that the grid $\bigtimes_{t \in [T]} \mathcal{S}_t$ and the given marginals $\{ \mu_t \}_{t \in \mathcal{T}}$ satisfy the assumptions of the second part of \cref{thm:algo}. Then $(\lambda, \gamma)$ are optimal variables for the dual problem (\ref{eqn:dualproblem}) if and only if the following equations hold 
    \begin{equation*}
        u_t^{\lambda_t} =  m_t \oslash P^{\mathcal{S}} ( \hat{\psi}_t \odot \psi_t ) , \quad t \in \mathcal{T},
    \end{equation*}
    and  
    \begin{align*}
              &\hat{\psi}_t \odot \big(1_{n_t^{\mathcal{X}}} \otimes u_{t}^{\lambda_{t}} \big) \odot  \big( K_{t+1} \odot G^{\gamma_t}_{t+1} \odot \Delta_{t+1} \big) \big( \psi_{t+1} \odot (1_{n_{t+1}^{\mathcal{X}}} \otimes u_{t+1}^{\lambda_{t+1}}) \big) = 0_{n_t}, 
              & t\in \mathcal{T}, \;
              & t+1 \in  \mathcal T \\
              &\hat{\psi}_t \odot  \big( K_{t+1} \odot G^{\gamma_t}_{t+1} \odot \Delta_{t+1} \big) \big( \psi_{t+1} \odot (1_{n_{t+1}^{\mathcal{X}}} \otimes u_{t+1}^{\lambda_{t+1}}) \big) = 0_{n_t}, 
              & t\notin \mathcal{T}, \; 
              & t+1 \in  \mathcal T \\
             & \hat{\psi}_t \odot (1_{n_t^{\mathcal{X}}} \otimes u_{t}^{\lambda_{t}}) \odot \left( K_{t+1} \odot G^{\gamma_t}_{t+1} \odot \Delta_{t+1} \right) \psi_{t+1}  = 0_{n_t}, 
              & t\in \mathcal{T}, \;
              & t+1 \notin \mathcal T \\
               & \hat{\psi}_t  \odot \left( K_{t+1} \odot G^{\gamma_t}_{t+1} \odot \Delta_{t+1} \right) \psi_{t+1}  = 0_{n_t}, 
              & t\notin \mathcal{T}, \;
              & t+1 \notin \mathcal T.
    \end{align*}
\end{theorem}

\begin{remark}
    According to the above theorem, we can compute dual variables satisfying the optimality conditions from equations (\ref{eqn:algo_updateu1}) and (\ref{eqn:algo_updategamma1}), by simply performing a number of matrix-vector products. This greatly reduces the computational work required, especially since the help vectors $\hat{\psi}$ (resp. $\psi$) can be computed inductively (resp. recursively); therefore,  we do not have to recalculate the full chains for each individual variable update. We also emphasise that according to the above theorem, there is no need to explicitly form the tensors $\mathbf{K}$, $\mathbf{U}^{\lambda}$ and $\mathbf{G}^{\gamma}$.
\end{remark}

In order to prove \cref{thm:algo_formulas}, we will need the following lemma. Its proof can be found in \cref{sec:appendix}.

\begin{lemma} \label{lemma:algo_filipisabell_applied}
    Let $\lambda=\{\lambda_t\}_{t\in \mathcal T}$ and  $\gamma=\{\gamma_t\}_{t\in [T-1]}$ be given families of vectors with $\lambda_t \in \mathbb{R}^{n_t^{\mathcal{S}}}$ and $\gamma_t \in \mathbb{R}^{n_t}$, respectively. Define $u_t^{\lambda_t}$ for $t \in \mathcal{T}$ and $K_t$ and $G^{\gamma_{t-1}}_t$ for $t \in [T] \backslash 0$ as in \cref{thm:algo_formulas} and let $\mathbf{K} (i_0, \dots, i_T) = \sum_{t \in [T] \backslash 0} K_t(i_{t-1}, i_t) $, $\mathbf{U}^{\lambda}(i_0, \dots, i_T ) = \sum_{t \in \mathcal{T}} u_t^{\lambda_t}(i_t^{\mathcal{S}})$ and $\mathbf{G}^{\gamma} (i_0, \dots, i_T) = \sum_{t \in [T] \backslash 0} G^{\gamma_{t-1}}_t(i_{t-1}, i_t) $ for $(i_0, \dots, i_T) \in \mathcal{I}$. Then 
    \begin{equation*}
        \begin{aligned}
            P_t &  \left( \mathbf{K} \odot \mathbf{U}^{\lambda} \odot \mathbf{G}^{\gamma}  \right) 
            =  
             \hat{\psi}_t \odot \bar{u}_t  \odot \psi_t , & t \in [T],
        \end{aligned}
    \end{equation*}
    where $\hat{\psi}$ and $\psi$ are functions of $\lambda$ and $\gamma$ as given in equations (\ref{eqn:algo_formulas_f}) and (\ref{eqn:algo_formulas_b}), respectively, and $\bar{u}_t \in \mathbb{R}^{n_t}$ is given by 
    \begin{align} \label{eqn:algo_formulas_lemma_pf-1}
        \bar{u}_t 
        =
        \begin{cases}
            1_{n_t^{\mathcal{X}}} \otimes u_t^{\lambda_t} , & t \in \mathcal T \\
            1_{n_t}, & t \in [T] \backslash \mathcal T .
        \end{cases}
    \end{align}
    Moreover, for $t_1, t_2 \in [T]$ such that $t_1 < t_2$, 
     \begin{equation*} \label{eqn:algo_filipisabell_applied_1}
     \begin{aligned}
        P_{t_1, t_2} &( \mathbf{K} \odot  \mathbf{U}^{\lambda} \odot  \mathbf{G}^{\gamma} ) \\
        &=  
         \diag( \hat{\psi}_{t_1} \odot \bar{u}_{t_1} )  \big(
         (K_{t_1 +1} \odot G^{\gamma_{t_1}}_{t_1 + 1}) 
         \diag(\bar u_{t_1+1}) 
         \dots 
         (K_{t_2} \odot G^{\gamma_{t_2-1}}_{t_2}) \diag( \bar u_{t_2} ) \big) \diag (\psi_{t_2}) .
    \end{aligned}
    \end{equation*}
\end{lemma}

We now prove \cref{thm:algo_formulas}.

\begin{proof}[Proof of \cref{thm:algo_formulas}]
    \Cref{thm:algo} states that the (representation of) the optimal dual variables $u^{\lambda}$ and $\gamma$ should satisfy the optimality conditions (\ref{eqn:algo_updateu1}) and (\ref{eqn:algo_updategamma1}). Start by considering the former, whose right-hand side is equal to
    \begin{equation} \label{eqn:algo_formulas_pf1}
        ( m_t \odot u_t^{\lambda_t} ) \oslash P_t^{\mathcal{S}} ( \mathbf{K} \odot  \mathbf{U}^{\lambda} \odot  \mathbf{G}^{\gamma}), \quad t \in \mathcal T .
    \end{equation}
   We must find an expression for the projection of the tensor $ \mathbf{K} \odot  \mathbf{U}^{\lambda} \odot  \mathbf{G}^{\gamma}$. Fix some $t \in \mathcal T$ and start by recalling that $P_t^{\mathcal{S}} (  \mathbf{K} \odot  \mathbf{U}^{\lambda} \odot  \mathbf{G}^{\gamma} ) =  ( P^{\mathcal{S}_t} \circ P_t ) (  \mathbf{K} \odot  \mathbf{U}^{\lambda} \odot  \mathbf{G}^{\gamma} )$, where \cref{lemma:algo_filipisabell_applied} then yields an expression for $P_t( \mathbf{K} \odot  \mathbf{U}^{\lambda} \odot  \mathbf{G}^{\gamma})$. Therefore
    \begin{equation} \label{eqn:algo_formulas_pf2}
        \begin{aligned}
            P_t^{\mathcal{S}} (  \mathbf{K} \odot  \mathbf{U}^{\lambda} \odot  \mathbf{G}^{\gamma} ) 
            =
            P^{\mathcal{S}_t} \big( \hat{\psi}_t \odot ( 1_{n_t^{\mathcal{X}}} \otimes u_t^{\lambda_t} ) \odot \psi_t \big)
            = P^{\mathcal{S}_t} ( \hat{\psi}_t \odot \psi_t ) \odot  u_t^{\lambda_t} 
        \end{aligned}
    \end{equation}
    where the second equality follows from the fact that $( 1_{n_t^{\mathcal{X}}} \otimes u_t^{\lambda_t} )(i_t) = u_t^{\lambda_t}(i_t^{\mathcal{S}})$ for every $i_t \in \mathcal{I}_t$ given by the order (\ref{eqn:tensors_idxorder}). Inserting \cref{eqn:algo_formulas_pf2} into \cref{eqn:algo_formulas_pf1} yields the expression stated in the theorem.

    Moving on to \cref{eqn:algo_updategamma1}, we apply \cref{lemma:algo_filipisabell_applied} to obtain an expression for the bi-marginal projections, $ P_{t, t+1} (  \mathbf{K} \odot  \mathbf{U}^{\lambda} \odot  \mathbf{G}^{\gamma} ) =  \diag( \hat{\psi}_{t} \odot \bar{u}_t )  (K_{t+1} \odot G^{\gamma_t}_{t+1}) \diag(\psi_{t+1} \odot \bar u_{t+1})$ for $t \in [T-1]$.
     The optimality condition from \cref{eqn:algo_updategamma1} thus becomes
    \begin{equation*}
        \begin{aligned}
            \hat{\psi}_{t} 
        \odot \bar u_{t}  
        \odot ( K_{t+1} \odot G^{\gamma_t}_{t+1} \odot \Delta_{t+1} ) (\psi_{t+1} \odot \bar u_{t+1})
        &=
        0_{n_t} , \quad t \in [T-1] .
        \end{aligned}
    \end{equation*}
    Inserting the definition of $\bar{u}_t$ and $\bar{u}_{t+1}$ from \cref{eqn:algo_filipisabell_applied_1} into the above expression proves the second part of the assertion and thus completes the proof.  
\end{proof}

We end this section by noting that \cref{lemma:algo_filipisabell_applied} provides an expression that allows for fast computation of the sub-transport between any two marginals, without explicitly having to form the full transport $ \mathbf{Q}$. One type of sub-transport that is of particular interest is transportation between adjacent marginals; later we will see that such sub-transports allow for recovering the price without explicitly forming the full transport $ \mathbf{Q}$.
\begin{corollary} \label{cor:algo_price}
     Let $\lambda=\{\lambda_t\}_{t\in \mathcal T}$ and  $\gamma=\{\gamma_t\}_{t\in [T-1]}$ be given families of vectors with $\lambda_t \in \mathbb{R}^{n_t^{\mathcal{S}}}$ and $\gamma_t \in \mathbb{R}^{n_t}$, respectively. Define $u_t^{\lambda_t}$ for $t \in \mathcal{T}$ and $K_t$ and $G^{\gamma_{t-1}}_t$ for $t \in [T] \backslash 0$ as in \cref{thm:algo_formulas} and let a corresponding transport plan $\mathbf{Q}^{\lambda, \gamma}$ be given by \cref{eqn:algo_Popt2}. Then the bi-marginal sub-transport $P_{t-1, t}( \mathbf{Q}^{\lambda, \gamma})$ for $t \in [T] \backslash 0$ is given by 
    \begin{align*}
         \begin{cases}
         \begin{aligned}
             &\diag \big( \hat{\psi}_{t-1} \odot (1_{n_{t-1}^{\mathcal{X}}} \otimes u_{t-1}^{\lambda_{t-1}}) \big)
             ( K_t \odot G^{\gamma_{t-1}}_t )
             \diag\big( \psi_{t} \odot (1_{n_t^{\mathcal{X}}} \otimes  u_{t}^{\lambda_{t}} ) \big), 
             & t-1  \in \mathcal{T} , \quad & t \in \mathcal{T} \\
             &\diag \big( \hat{\psi}_{t-1} \odot (1_{n_{t-1}^{\mathcal{X}}} \otimes u_{t-1}^{\lambda_{t-1}}) \big)
             ( K_t \odot G^{\gamma_{t-1}}_t )
             \diag ( \psi_{t}),
             & t-1 \in \mathcal T, \quad & t \notin  \mathcal T \\
             &\diag ( \hat{\psi}_{t-1}  )
             ( K_t \odot G^{\gamma_{t-1}}_t )
             \diag\big( \psi_{t} \odot ( 1_{n_t^{\mathcal{X}}} \otimes  u_{t}^{\lambda_{t}} )\big),
             & t-1 \notin \mathcal T, \quad & t \in \mathcal T\\
             &\diag ( \hat{\psi}_{t-1} )
             ( K_t \odot G^{\gamma_{t-1}}_t )
             \diag ( \psi_{t} ), 
             & t-1 \notin \mathcal T, \quad &t \notin \mathcal T ,
            \end{aligned}
         \end{cases}
    \end{align*}
    where $\hat{\psi}$ and $\psi$ are given as functions of $\lambda$ and $\gamma$ by equations (\ref{eqn:algo_formulas_f}) and (\ref{eqn:algo_formulas_b}).

\end{corollary}

\begin{proof}
    Fix some $t \in [T] \backslash 0$ and note that by \cref{eqn:algo_Popt2}, $P_{t-1, t}( \mathbf{Q}^{\lambda, \gamma}) = P_{t-1, t}( \mathbf{K} \odot  \mathbf{U}^{\lambda} \odot  \mathbf{G}^{\gamma})$, where $ P_{t-1,t}( \mathbf{K} \odot  \mathbf{U}^{\lambda} \odot  \mathbf{G}^{\gamma}) = \diag ( \hat{\psi}_{t-1} \odot \bar u_{t-1} )( K_t \odot G^{\gamma_{t-1}}_t )\diag ( \psi_{t} \odot \bar u_{t}^{\lambda_{t}} )$
    follows from application of \cref{lemma:algo_filipisabell_applied} with $t_1 = t-1$ and $t_2 = t$. Here $\bar u_{t-1}$ and $\bar u_t$ are defined as in \cref{eqn:algo_formulas_lemma_pf-1}. This proves the assertion.  
\end{proof}

\subsection{Summary of the full algorithm}
We have now arrived at a form of the optimality conditions that allows for efficient computation of the optimal dual variables and we can therefore construct a coordinate dual ascent method by cyclically fixing all but one variable and optimise over the remaining variable by selecting it so that it satisfies the corresponding optimality condition. We then move on to the next variable until all variables have been optimised over; one such cycle defines one iteration in the method.

Each iteration starts with the update of the variables $u^\lambda$ corresponding to the marginal constraints (\ref{eqn:tensors_mot_margconstr}). The closed-form formula for doing so is given by \cref{thm:algo_formulas} as
\begin{equation*}
    u_t^{\lambda_t} \leftarrow m_t \oslash P^{\mathcal{S}} ( \hat{\psi}_t \odot \psi_t ) , \quad t \in \mathcal{T}.
\end{equation*}
From the above we note that is suffices to manipulate vectors in order to update $u^\lambda$, something that enables fast computation.

We then proceed to updating the dual variables $\gamma$ representing the martingale constraints (\ref{eqn:tensors_mot_mtgconstr}). This time, no closed-form formula exists for $\gamma$ --- the optimality condition thus has to be solved numerically. \Cref{thm:algo_formulas} provides a form of the equations that allows for efficiently doing so by application of Newton's method. In order to do so, define a family of functions $\alpha_t : \mathbb{R}^{n_{t-1}} \rightarrow \mathbb{R}^{n_{t-1}}$ for $t \in [T] \backslash 0$ via
\begin{equation*}
    \alpha_t(z) 
    :=
    \begin{cases}
        \big( K_t \odot \exp( \diag(z) \Delta_t / \varepsilon)\odot \Delta_t \big) \big(\psi_t \odot (1_{n_t^{\mathcal{X}}} \otimes u_t^{\lambda_t})\big) , & t \in  \mathcal T\\
        \big( K_t \odot \exp( \diag(z) \Delta_t / \varepsilon) \odot \Delta_t \big) \psi_t , & t \notin   \mathcal T.
    \end{cases} 
\end{equation*}
Then \cref{eqn:algo_updategamma1} can be written 
\begin{align*}
    \begin{cases}
     \hat{\psi}_{t} \odot (1_{n_t^{\mathcal{X}}} \otimes u_{t}^{\lambda_{t}}) \odot \alpha_{t+1}(\gamma_{t}) = 0_{n_t}, & t \in [T-1] \cap \mathcal{T} \\
      \hat{\psi}_{t} \odot  \alpha_{t+1}(\gamma_{t}) = 0_{n_t}, & t \in [T-1] \backslash \mathcal{T}.
      \end{cases}
\end{align*}
Application of Newton's method to the above equation requires iterating according to 
\begin{equation} \label{eqn:algo_newton1}
    \gamma_{t}^{(k+1)} 
    \leftarrow 
    \gamma_{t}^{(k)} - \theta^{(k)} \odot J^{\alpha}_{t+1}(\gamma_{t}^{(k)})^{-1}\alpha_{t+1}(\gamma_{t}^{(k)} ), \quad k = 1, 2, \dots ,
\end{equation}
where $\theta^{(k)} \in \mathbb{R}^{n_t}$ is a vector containing the respective step lengths to be used for each variable, such that $\theta^{(k)}(j) \in [0,1]$ for $j \in \mathcal{I}_t$. It is determined by vectorised line search. In the above, $J^{\alpha}_{t+1}$ denotes the Jacobian matrix associated with $\alpha_{t+1}$; note that $\alpha_{t+1}$ is a vector-valued function whose $j^{\text{th}}$ component is independent of any other components of $\gamma_{t}$ than $\gamma_{t}(j)$, $j \in \mathcal{I}_t$. There are no cross dependencies. Therefore, the Jacobian is a diagonal matrix and given by $J^{\alpha}_{t}(\gamma_{t-1}) = \diag(\beta_{t}(\gamma_{t-1}))$, where $\beta_t : \mathbb{R}^{n_{t-1}} \rightarrow \mathbb{R}^{n_{t-1}}$ for $t \in [T] \backslash 0$ is defined as 
\begin{equation*}
     \beta_t(z) 
    :=
    \begin{cases}
        \varepsilon^{-1} \big( K_t \odot \exp( \diag(z) \Delta_t / \varepsilon) \odot \Delta_t \odot \Delta_t \big) \big(\psi_t \odot (1_{n_t^{\mathcal{X}}} \otimes u_t^{\lambda_t})\big) , & t \in  \mathcal T \\
        \varepsilon^{-1} \big( K_t \odot \exp( \diag(z) \Delta_t / \varepsilon) \odot \Delta_t \odot \Delta_t \big) \psi_t , &  t \notin  \mathcal T .
    \end{cases} 
\end{equation*}
The iterative scheme (\ref{eqn:algo_newton1}) can therefore be simplified to 
\begin{equation*} \label{eqn:algo_newton2}
    \gamma_{t}^{(k+1)} 
    \leftarrow 
    \gamma_{t}^{(k)} 
    - 
     \theta^{(k)} \odot
     \alpha_{t+1}(\gamma_{t}^{(k)})
    \oslash \beta_{t+1}(\gamma_{t}^{(k)})
    , \quad k = 1, 2, \dots 
\end{equation*}
The inversion of the Jacobian and the matrix-vector multiplication used in (\ref{eqn:algo_newton1}) are thus avoided and replaced by the elementwise division of one vector with another, allowing for efficient evaluation of the dual optimality condition (\ref{eqn:algo_updategamma1}). Newton's method is therefore particularly well suited for this step in the method.

Putting things together, we have completed the derivation of our algorithm --- see \cref{algo:main} for the full method. Convergence of the algorithm can be shown given that there, by \cref{lemma:algo_interiorpt}, exists a feasible solution that is positive for all variables 
(cf. \cite[Theorem 4.1 and Proposition 1]{HaaslerRinghChenKarlsson2023}).  
Convergence of the coordinate dual ascent then follows from  \cite[Theorem 2 and Section 5]{LuoTsengQ1992}.

\begin{algorithm}
\begin{algorithmic}[] 
\State Given: matrices $\{K_t \}_{t\in [T] \backslash 0}$ and $\{\Delta_t\}_{t \in [T] \backslash 0}$, index set $\mathcal{T}$, vectors $\{ m_t \}_{t \in \mathcal{T}}$, scalar $\varepsilon > 0$
\State Initialise: $u_t^{\lambda_t} \gets 1_{n_t^{\mathcal{S}}}$  for $t \in [T]$ 
\State \phantom{Initialise:} $\gamma_t \gets 1_{n_t}$ for $t \in [T-1]$ 
\State \phantom{Initialise:} $G^{\gamma_{t-1}}_t \gets \exp( \diag(\gamma_{t-1}) \Delta_t / \varepsilon )$ for $t \in [T] \backslash 0$
\State \phantom{Initialise:} $\hat{\psi}_0  \gets 1_{n_0}$ 
\State \phantom{Initialise:} $\psi_T \gets 1_{n_T}$ 
\State  \phantom{Initialise:} $\psi_t \gets (K_{t+1} \odot G_{t+1})( \psi_{t+1} \odot ( 1_{n_{t+1}^{\mathcal{X}}} \otimes u_{t+1}^{\lambda_{t+1}})  )$ for $t \in [T-1]$
\While{Sinkhorn not converged} \\
    \emph{   \#Update marginal constraints:}
    \If{$0 \in \mathcal{T}$}
        \State $u^{\lambda_0}_0 \gets m_0 \oslash P^{\mathcal{S}_0}(\psi_{0})$ 
    \EndIf
    \For{ $t = 1, \dots , T-1 $ }
        \State $\hat{\psi}_t \gets (K_{t} \odot G^{\gamma_{t-1}}_{t})^{\top}( \hat{\psi}_{t-1} \odot ( 1_{n_{t-1}^{\mathcal{X}}} \otimes u_{t-1}^{\lambda_{t-1}}))$
        \If{$t \in \mathcal{T}$}
            \State $u_t^{\lambda_t} \gets m_t 
                                \oslash 
                                P^{\mathcal{S}_t}
                                (
                                     \hat{\psi}_t \odot \psi_t 
                                ) $
        \EndIf
    \EndFor 
    \State $\hat{\psi}_T \gets (K_{T} \odot G^{\gamma_{T-1}}_{T})^{\top}( \hat{\psi}_{T-1} \odot (1_{n_{T-1}^{\mathcal{X}}} \otimes u^{\lambda_{T-1}}_{T-1}))$
    \If{$T \in \mathcal{T}$}
        \State $u^{\lambda_T}_T \gets m_T \oslash P^{\mathcal{S}_T}( \hat{\psi}_T ) $ 
    \EndIf \\
    \emph{   \#Update martingale constraints:}
    \For{ $t = T, \dots, 1$ }
        \If{$t \neq T$}
            \State $\psi_t \gets (K_{t+1} \odot G^{\gamma_t}_{t+1})( \psi_{t+1} \odot (1_{n_{t+1}^{\mathcal{X}}} \otimes u_{t+1}^{\lambda_{t+1}}))$
        \EndIf
        \State $k \gets  0$
        \While{Newton not converged}
            \State $k \gets  k + 1$
            \State $\theta^{(k)} \gets 1$
            \State $\gamma_{t-1}^{(k+1)} \gets \gamma_{t-1}^{(k)} -   \theta^{(k)} \odot \alpha_t(\gamma_{t-1}^{(k)}) \oslash  \beta_t(\gamma_{t-1}^{(k)}) $
            \While{$\gamma_{t-1}^{(k+1)} > \gamma_{t-1}^{(k)}$}
                \State $\theta^{(k)} \gets \theta^{(k)} / 2$
                \State $\gamma_{t-1}^{(k+1)} \gets \gamma_{t-1}^{(k)} -   \theta^{(k)} \odot \alpha_t(\gamma_{t-1}^{(k)}) \oslash  \beta_t(\gamma_{t-1}^{(k)}) $
            \EndWhile
            \State $G^{\gamma_{t-1}}_{t} \gets \exp( \diag (\gamma_{t-1}^{(k+1)})\Delta_t / \varepsilon )$
        \EndWhile
    \EndFor 
    \State $\psi_{0} \gets (K_1 \odot G^{\gamma_0}_1)( \psi_1 \odot (1_{n_1^{\mathcal{X}}} \otimes u^{\lambda_1}_1)  )$
\EndWhile \\
\Return{$u$, $\hat{\psi}$, $\psi$}, $\{G^{\gamma_{t-1}}_t\}_{t=1}^T$
\end{algorithmic}
\caption{Method for solving problem (\ref{eqn:tensors_mot}). }
\label{algo:main}
\end{algorithm}

\begin{remark}
Storing the tensor $\mathbf{Q}$ quickly becomes impossible as the number of marginals grows and consequently we cannot access the full transport $\mathbf{Q}$ in practice. Fortunately, it is often enough to obtain a set of sub-transports of the form  $P_{t_1, t_2}(\mathbf{Q})$, where $t_1, t_2 \in [T]$ are such that $t_1 < t_2$. For example, the corresponding problem value --- the lower bound on the price --- can be obtained without explicitly forming and storing the full tensor. See this by recalling that $\langle \mathbf{\Phi}, \mathbf{Q} \rangle = \sum_{t \in [T] \backslash 0} \langle \Phi_t , P_{t-1,t}(\mathbf{Q}) \rangle$,
where the sub-transports $P_{t-1,t}(\mathbf{Q})$ for $t \in [T] \backslash 0$ are given by \cref{cor:algo_price}.
\end{remark}

\section{Computational results} \label{sec:sol}
In this section, we apply our computational framework to solve a number of MOT problems: First, we consider problems for which the optimal solution is known analytically; comparison of the computed solution with the known solution allows for verifying that the method works as intended. We then conclude by solving an MOT problem where the optimal solution is not known. In all examples the regularisation parameter $\varepsilon$ was chosen by trial and error to be as small as possible.

\subsection{Monotone transport}
A well-known example from classical OT theory is the monotone transport plan, which is concentrated on the graph of an increasing function (see, for example, \cite[p. 75]{Villani2003}). It is also known as the Hoeffding-Fréchet coupling and solves the upper bound bi-marginal OT problem when the cost function $\phi : \mathbb{R} \times \mathbb{R} \rightarrow \mathbb{R}$ satisfies certain conditions\footnote{The cost function $\phi$ should be upper semi-continuous and of linear growth.} --- in particular, the cross derivative $\partial_{s_0 s_1}\phi( s_0 , s_1)$ should exist and be strictly positive  \cite[Theorem 2.2]{HL-T2016}. Similarly, the anti-monotone transport plan is concentrated on the graph of a decreasing function; it solves the corresponding lower bound OT problem. There exists a martingale analogue to the monotone coupling; the \emph{left-monotone transport plan} was introduced in \cite[Definition 1.4]{Beiglbock-Juillet2016} and shown to solve the bi-marginal MOT problem for two types of cost functions \cite[Theorems 6.1 and 6.3]{Beiglbock-Juillet2016}. It is concentrated on the graphs of two measurable functions and exhibits a V-shape when visualised in the plane. It was shown in \cite[Theorem 5.1]{HL-T2016} that the left-monotone coupling solves the upper bound bi-marginal MOT problem, given that the derivative $\partial_{s_0 s_1 s_1}\phi(s_0, s_1)$ exists and is strictly positive. Since the monotone couplings are visually easily recognised, they serve as a clear first demonstrating example.

\begin{example} [Robust pricing of a variance swap] \label{ex:sol_monotone}
    Let $\phi(s_0, s_1) = ( \log(s_1/s_0))^2$ and let the initial marginal $\mu_0$ be uniform with support on $n_0^{\mathcal{S}} = 600$ atoms, evenly distributed over the interval $[1.25,1.75]$. Similarly, let the terminal marginal $\mu_1$ be uniform with support on $n_1^{\mathcal{S}} = 1200$ atoms, evenly distributed over the interval $[1,2]$. It is easily verified that $\partial_{s_0 s_1 s_1}\phi > 0$ and $\partial_{s_0 s_1}\phi < 0$ on the supports of $\mu_0$ and $\mu_1$ and that the marginals are in convex order; the optimal solutions for the corresponding upper-bound MOT and OT problems are thus the left-monotone and the anti-monotone couplings, respectively. 
    \begin{figure}
    \centering
    \begin{subfigure}[t]{0.45\textwidth}
        \centering
        \includegraphics[width=6cm]{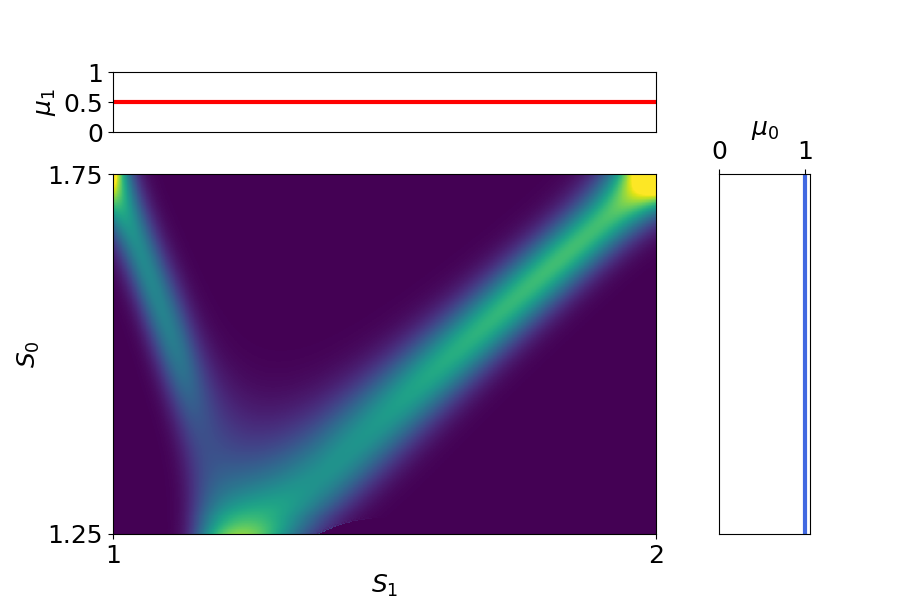}
        \caption{Optimal martingale coupling.}
        \label{fig:sol_monotone_mtg}
    \end{subfigure}\hspace{1cm}
    \begin{subfigure}[t]{0.45\textwidth}
        \centering
        \includegraphics[width=6cm]{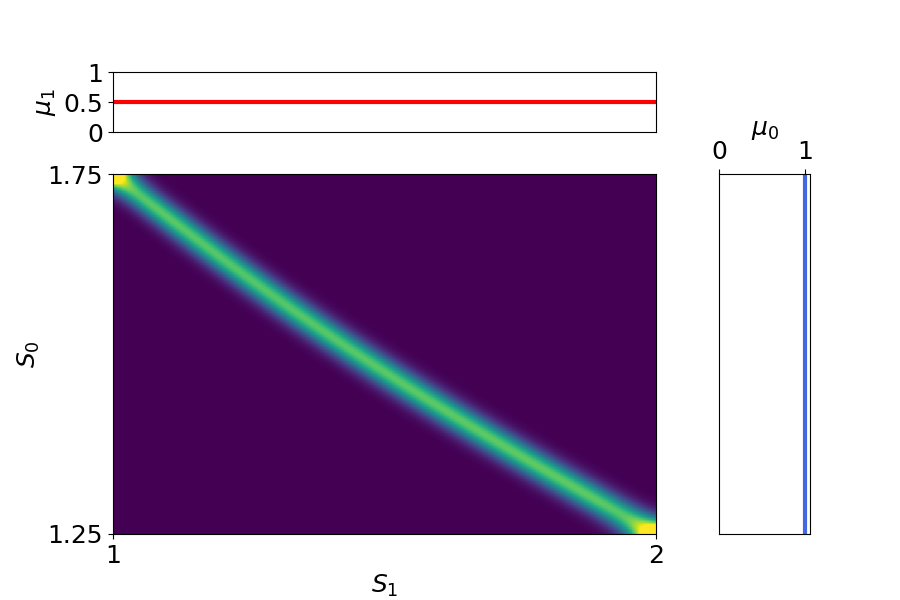}
        \caption{Optimal non-martingale coupling.}
        \label{fig:sol_monotone_nonmtg}
    \end{subfigure}
    \caption{Optimal solutions to the upper bound bi-marginal MOT and OT problems corresponding to the payoff of a variance swap when the given marginals are uniform, as in \cref{ex:sol_monotone}. The computed solution of the MOT problem is displayed in \cref{fig:sol_monotone_mtg}; it exhibits the shape of the left-monotone transport. For comparison, the computed solution of the corresponding OT problem is displayed in \cref{fig:sol_monotone_nonmtg}; it displays the shape of the anti-monotone coupling.}
    \label{fig:sol_monotone}
\end{figure}
    \Cref{fig:sol_monotone} shows the computed optimal solutions obtained for $\varepsilon = 4.5 \cdot 10^{-4}$. A dark blue colour indicates a close-to-zero probability for that specific transport to occur, whilst a yellow or green colour indicates a non-zero probability that this transport occurs. Note the V-shape of the optimal martingale coupling. 
\end{example}

\subsection{Late and early transports} \label{sec:sol_lateearly}
Next we consider a multi-marginal ($T> 1$) example where the payoff function is of the form
    \begin{equation} \label{eqn:sol_lateearly_fcn}
        \phi(s_0, \dots, s_T) = \frac{1}{T+1} \sum_{t \in [T]} f(s_t), \quad (s_0, \dots s_T) \in \mathcal{S}_0 \times \dots \times \mathcal{S}_T,
    \end{equation}
for some convex function $f:  \mathbb{R} \rightarrow \mathbb{R}$, and where the initial and terminal marginals $\mu_0$ and $\mu_T$ are the only given marginals. As always, it is assumed that $\mu_0$ and $\mu_T$ are in convex order. Two particular martingale couplings are of interest here, the coupling that corresponds to not moving at all before $t = T-1$, then performing the full transport between $t=T-1$ and $t=T$, and the coupling that corresponds to performing the full transport between $t=0$ and $t=1$, then keeping the process constant until $t=T$. We refer to these couplings as late and early transport, respectively. The below result follows immediately from linearity of the expected value, since for any model $(\Omega, \mathcal{F}, \mathbb{Q}, S)$ we have that $ \mathbb{E}_{\mathbb{Q}} [ f(S_T)]
        \ge
        \mathbb{E}_{\mathbb{Q}} [ f(S_{T-1})]
        \ge 
        \dots
        \ge 
        \mathbb{E}_{\mathbb{Q}} [ f(S_1) ]
        \ge
        \mathbb{E}_{\mathbb{Q}} [ f(S_0)]$,
by convexity of $f$. See \cite{Stebegg2014} for a similar result.

\begin{proposition} \label{prop:sol_lateearly}
    Let $T > 1$ and consider a payoff function $\phi : \mathbb{R}^{T+1} \rightarrow \mathbb{R}$ of the form (\ref{eqn:sol_lateearly_fcn}). Let $\mu_0$ and $\mu_T$ be two given probability measures on $(\mathbb{R}, \mathcal B(\mathbb R))$ and assume that they are in convex order.  Then the optimal solution of the associated lower bound MOT problem corresponds to the late transport defined above, whilst the early transport solves the upper bound MOT problem.
\end{proposition}

\begin{example} \label{ex:sol_lateearly}
    Let $f(x) = x^2$, $T = 50$ and let the two given marginal distributions $\mu_0$ and $\mu_{50}$ be as in \cref{fig:sol_lateearly}; they are in convex order. Since the payoff function is of the form $\phi(s_0, \dots , s_T) = \phi_1(s_0, s_1) + \sum_{t \in [T] \backslash \{0, 1\} } \phi_t(s_t)$  for 
    \begin{equation*}
        \begin{aligned}
            \phi_1(s_0, s_1) &= \frac{1}{T+1} \left( f(s_0) + f(s_1) \right) , && (s_0 , s_1) \in \mathcal S_0 \times \mathcal{S}_1, \\
            \phi_t(s_t) & = \frac{1}{T+1} f(s_t), && s_t \in \mathcal S_t, \quad t = 2, \dots, T ,
        \end{aligned}
    \end{equation*}
    it is by \cref{rem:sol_markov} possible to omit the memory process $X$ from the representation of the problem.
    
    \begin{figure}
        \centering
        \begin{subfigure}[t]{0.45\textwidth}
            \centering
            \includegraphics[width=6cm]{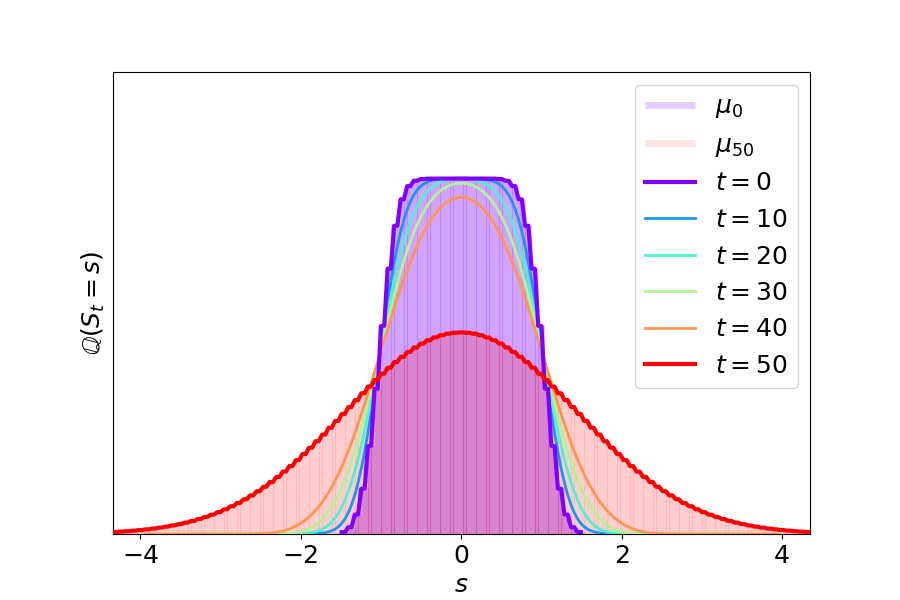}
            \caption{Marginals of the lower bound coupling.}
            \label{fig:sol_late}
        \end{subfigure}\hspace{1cm}
        \begin{subfigure}[t]{0.45\textwidth}
            \centering
            \includegraphics[width=6cm]{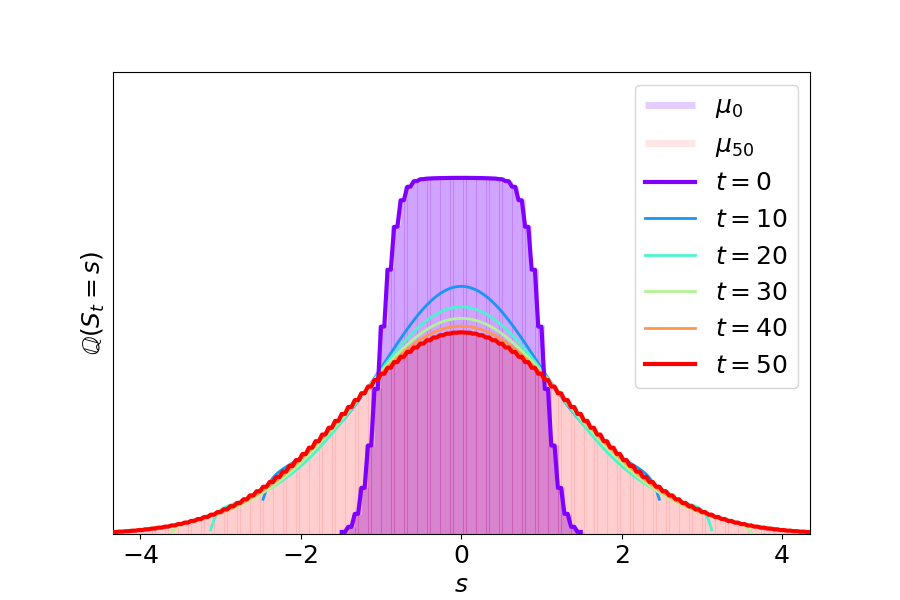}
            \caption{Marginals of the upper bound coupling.}
            \label{fig:sol_early}
        \end{subfigure}
        \caption{Marginal distributions corresponding to the optimal solutions of the upper- and lower bound MOTs from \cref{ex:sol_lateearly}. The given initial marginal $\mu_0$ is displayed in light blue bars, while the given terminal marginal $\mu_T$ is shown in light red bars. The marginals corresponding to the computed transport plan are shown for $t \in \{0, 10, 20, 30, 40, 50\}$ with solid lines. \Cref{fig:sol_late} corresponds to the lower bound martingale transport plan while \cref{fig:sol_early} corresponds to the upper bound martingale transport plan.}
        \label{fig:sol_lateearly}
    \end{figure}

    A subset of the marginal distributions of the computed optimal martingale transport plans are displayed in \cref{fig:sol_lateearly}, along with the given marginals $\mu_0$ and $\mu_{50}$. The number of gridpoints used ranges from $n_0^{\mathcal{S}} = 74$ to $n_{50}^{\mathcal{S}} = 214$. The lower bound solution (computed using $\varepsilon = 6.5 \cdot 10^{-3}$) is shown in \cref{fig:sol_late}; since the computed intermediate marginals are close to the initial marginal this solution corresponds to a late transport, as predicted by \cref{prop:sol_lateearly}. The upper bound solution (computed using $\varepsilon = 0.02$) is shown in \cref{fig:sol_early}; for this solution, we note that the intermediate marginals are close to the terminal marginal. It is therefore an example of early transport, as given by \cref{prop:sol_lateearly}. For both solutions, the small deviations of the intermediate marginals from the initial or terminal marginal are due to the regularisation. 

\end{example}

\subsection{The maximum of the maximum} \label{sec:sol_maxofmax}
Suppose that $\mu_T$ is a given marginal distribution, centered in $s_0 \in \mathbb R$, and assume that $\mu_0 = \delta_{s_0}$. Let the process $X$ be the rolling maximum of the price process $S$ --- that is, let $X$ be as in \cref{ex:examples_max}. We here consider the problem of finding the martingale model that maximises the maximum while respecting the terminal marginal $\mu_T$, that is, that solves the upper bound MOT problem with $\phi(S_0, \dots, S_T) = X_T$ and $\mathcal{T} = \{0, T\}$. Note that this payoff is of the form (\ref{eqn:intro_phi}); the problem can thus be addressed within our framework.

Before we compute its optimal solution for a specific choice of $\mu_T$, we will have a look at the corresponding continuous-time solution. Fix a probability space $(\Omega, \mathcal{F}, \mathbb{Q})$ and let $S^{\text{AY}}$ denote the martingale on $[0,T]$ whose maximum has a first-order stochastic dominance over the maximum of any other continuous martingale with terminal marginal $\mu_T$. It exists, and is related to the martingale constructed by Azéma and Yor \cite{AzemaYor1979} as a solution of the Skorokhod embedding problem (SEP) with terminal marginal $\mu_T$ \cite[Lemma 2.2]{BrownHobsonRogers2001_maxofmax}; we will refer to this martingale as the Azéma--Yor martingale. It follows that for the continuous-time case, an optimal solution of the above MOT problem is given by the Azéma--Yor martingale $S^{\text{AY}}$ \cite{Hobson1998}. A solution of the problem in continuous-time is thus known. It follows from \cite[Lemma 2.4]{BrownHobsonRogers2001} that said distribution can be obtained as 
\begin{equation} \label{eqn:sol_maxofmax_distrmax}
    \mathbb{Q}(X_{T}^{\text{AY}} \ge B) = \underset{0 \le y \le B}{\min \text{  }} f_B(y), \quad B > s_0 ,
\end{equation}
where $f_B : [0,B] \rightarrow \mathbb{R}$ is given by $f_B(y) := (B - y)^{-1} \int (s - y)^+ \mathrm{d}\mu_T(s)$ for $B > s_0$.
\Cref{eqn:sol_maxofmax_distrmax} thus directly provides the law of the maximum $X_T$ that corresponds to an optimal solution of the problem, when an `infinite number of time steps' is used. We will now compute this law for a specific choice of terminal marginal $\mu_{T}$. We then solve the problem computationally for $T$ increasing and compute the distribution of $X_T$ corresponding to an optimal solution. We will see that the resulting distribution approaches the distribution obtained from \cref{eqn:sol_maxofmax_distrmax} as $T$ increases.

\begin{example}[Comparison with the Azéma--Yor martingale] \label{ex:sol_maxofmax}    
     Let the terminal marginal $\mu_{T}$ be as in \cref{fig:sol_AY_marginal} with $s_0 = 0.5$ and $n_T^{\mathcal{S}} = 67$. The probability $\mathbb{Q}(X_{T}^{\text{AY}} \ge B)$ is obtained by minimising $f_B$ over $[0,B]$ --- note that its minimum exists for each $B$. By doing this for all $B \in (0.5, 1]$ we are able to recover the distribution of $X_{T}^{\text{AY}}$. It is displayed using a bold red line in \cref{fig:sol_maxofmax_cdf}. The corresponding problem value, the robust price, is $ \mathbb{E}_{\mathbb{Q}} [ X_{T}^{\text{AY}}] = 0.6133$.

    We now solve the corresponding upper bound MOT problem for different values of $T$, ranging from $2$ to 29, by application of our framework with $\varepsilon = 0.01$. The corresponding law of the maximum of the price process is displayed in \cref{fig:sol_maxofmax_cdf}.  The corresponding robust prices, $\mathbb{E}_{\mathbb{Q}}[X_T]$, are shown in \cref{fig:sol_maxofmax_price} as a function of the number of marginals used. We note that the computed distribution approaches the known continuous-time solution and that the cumulative distribution function decreases as the number of time steps used increases. Similarly the computed problem values, the robust prices, approaches the problem value of the corresponding continuous-time problem.

    For the problem with $T = 29$ we computed the residuals of all constraints, prior to updating the corresponding dual variable. The maximum element of the residuals is shown in \cref{fig:sol_maxofmax_residuals} as a function of the Sinkhorn iteration count, for both marginal constraints and for some of the martingale constraints, along with the threshold level used in Newton's method. Note that a high level of accuracy was used in this example --- the tolerated marginal and martingale residuals were set to $10^{-6}$ and $10^{-8}$, respectively. This is why the number of iterations in \cref{fig:sol_maxofmax_residuals} is quite large.

    \begin{figure}
        \centering\begin{subfigure}[t]{0.45\textwidth}
            \centering
            \includegraphics[width=6cm]{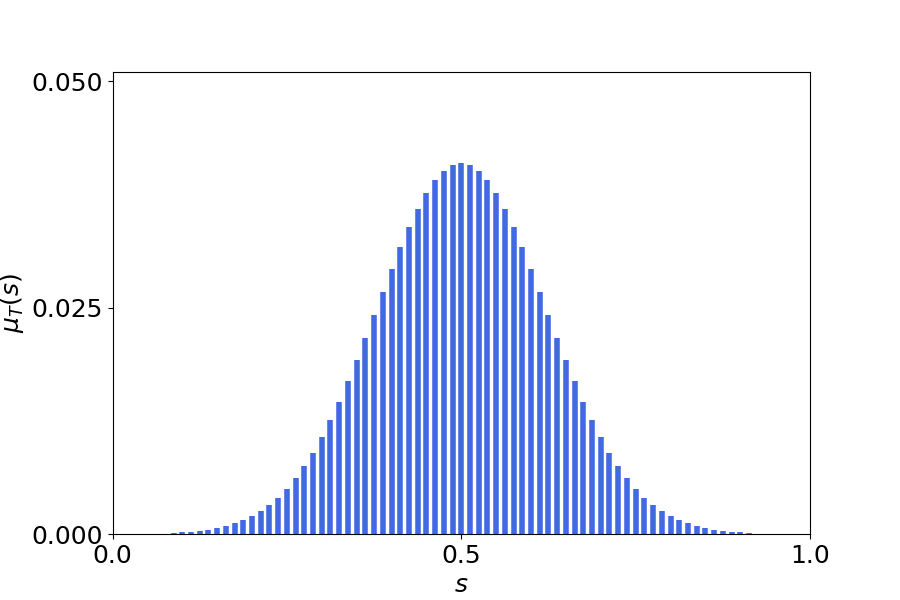}
            \caption{Given terminal marginal $\mu_{T}$.}
            \label{fig:sol_AY_marginal}
        \end{subfigure}\hspace{1cm}
        \begin{subfigure}[t]{0.45\textwidth}
            \centering
            \includegraphics[width=6cm]{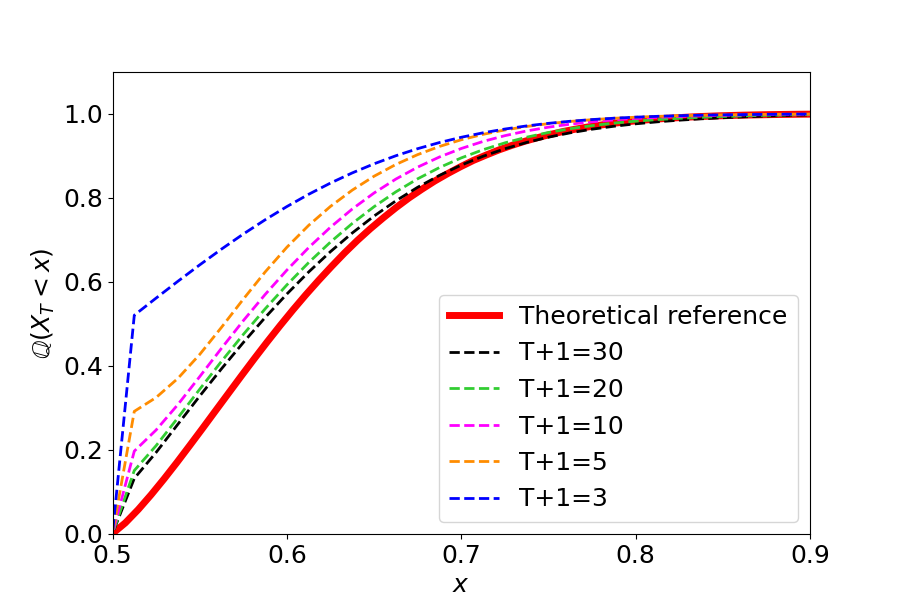}
            \caption{Computed cumulative distribution functions of the maximum of the maximum for an increasing number of marginals.}
            \label{fig:sol_maxofmax_cdf}
        \end{subfigure}\\
        \begin{subfigure}[t]{0.45\textwidth}
            \centering
            \includegraphics[width=6cm]{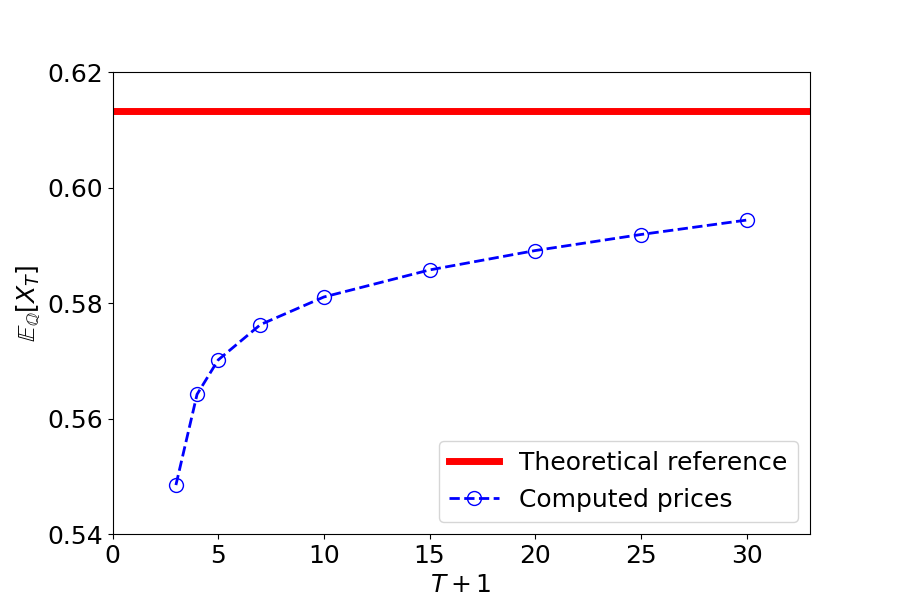}
            \caption{The computed robust price as a function of the number of marginals used.}
            \label{fig:sol_maxofmax_price}
        \end{subfigure}\hspace{1cm}
        \begin{subfigure}[t]{0.45\textwidth}
            \centering
            \includegraphics[width=6cm]{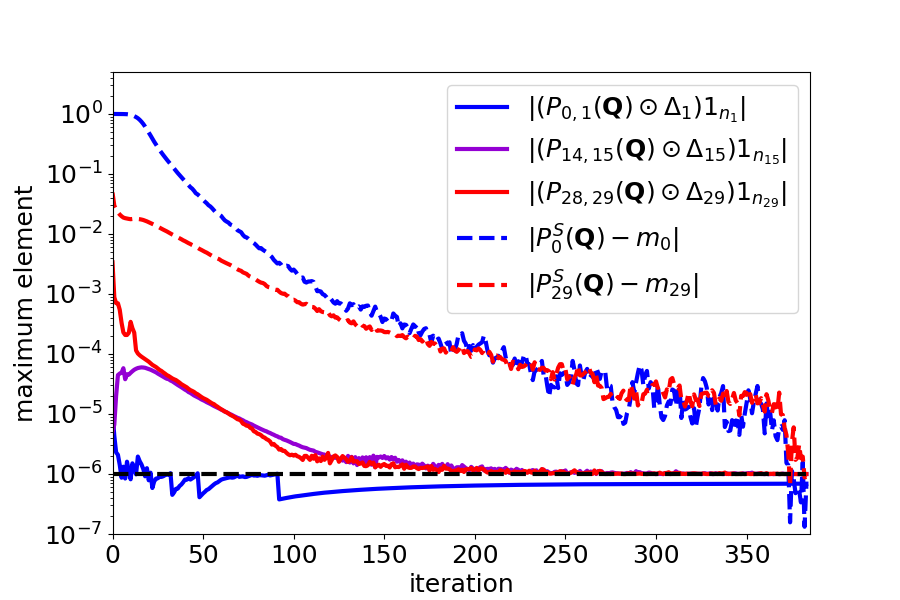}
            \caption{Computed residuals as a function of the Sinkhorn iteration count when $T = 29$.}
            \label{fig:sol_maxofmax_residuals}
        \end{subfigure}
        \caption{The given terminal marginal used in \cref{ex:sol_maxofmax} is displayed in \cref{fig:sol_AY_marginal}. It is centered in $0.5$ and its support is contained within $[0,1]$. The corresponding distribution of the maximum of the maximum, solved approximately for an increasing number of marginals $T+1$, is shown in \cref{fig:sol_maxofmax_cdf}. The corresponding distribution for the continuous-time case is included for reference (bold red). We note that the discrete-time computed solutions approach the continuous-time solution as the number of time steps used increases, and that the cumulative distribution function decreases. The corresponding problem values are shown as a function of the number of marginals $T+1$ in \cref{fig:sol_maxofmax_price}; the problem value for the continuous-time case is included for comparison (bold red). The maximum element of the computed residuals, obtained prior to updating the corresponding dual variable, for the marginal constraints (dashed line) and for some of the martingale constraints (bold line) are displayed in a linlog scale in \cref{fig:sol_maxofmax_residuals}, along with the threshold level used in Newton's method (black dashed line). Note the high level of accuracy used in the computations.}
        \label{fig:sol_maxofmax}
    \end{figure}

\end{example}

\begin{remark}[Other applications to computation of optimal Skorokhod embeddings] \label{rem:sol_SEP}
    In the above we saw that the Azéma-Yor martingale maximises the law of the maximum among all continuous martingales that respects the terminal law $\mu_T$, and that our framework can be used to approximately obtain this law by computing an upper bound MOT in discrete-time for $T$ large. There is another martingale, constructed from the Perkins solution of the same SEP, that minimises the law of the maximum among all continuous martingales that has terminal law $\mu_T$. We could analogously approximately compute this minimal law by solving the corresponding lower bound MOT problem. 
    
    If we instead let the auxiliary process $X$ be the realised variance of the price process $S$ (cf. \cref{ex:examples_var}) we can in a similar manner for $T$ large approximate the martingale that is induced by yet another solution of the SEP with terminal marginal $\mu_T$, namely, the Root solution. Indeed, it follows from results by Dupire and by Carr and Lee that it solves the corresponding continuous-time lower bound MOT problem with $\phi(S_0, \dots, S_T) = f(X_T)$ for $f : \mathbb{R} \rightarrow \mathbb{R}$ convex. See \cite{Hobson2011} for further details. 
\end{remark}

\subsection{The robust price of a digital option}
An example related to the maximum of the maximum is the robust upper bound on the fair price of a \emph{digital option} --- a claim that pays one unit of money if the underlying price process $S$ has exceeded the barrier $B$ before maturity, else it expires worthless. Yet again, let $\mu_T$ be a given terminal distribution, centered in $s_0 \in \mathbb{R}$, and assume that the initial distribution is given by $\mu_0 = \delta_{s_0}$. Fix some $B > s_0$, let $T > 1$ and let $X$ be the indicator process keeping track of whether the price process $S$ has exceeded the barrier $B$ so far, that is, let $ X_t := \chi_{[B, \infty )}( \max_{r \in  [t]}  S_r )$ for $t \in [T]$. Then the payoff of a digital option with barrier $B$ is given by $\phi(S_0, \dots, S_T) = X_T$; we are interested in solving the corresponding upper bound MOT problem.

We have seen that in continuous-time there exists a martingale $S^{\text{AY}}$ that optimises all digital options simultaneously, that is, the Azéma--Yor martingale $S^{\text{AY}}$ solves the problem for any barrier $B > s_0$. Indeed, since $\chi_{[B, \infty)}(\cdot)$ is an increasing function and the maximum of the martingale $S^{\text{AY}}$ has a first-order stochastic dominance over the maximum of any other continuous martingale $S$ with marginals $\mu_0$ and $\mu_T$, 
    \begin{equation*}
        \mathbb{E}_{\mathbb{Q}}\big[  \chi_{[B, \infty)}\big ( \max_{t \in [0,T]} S_t^{\text{AY}} \big) \big] 
        \ge 
        \mathbb{E}_{\mathbb{Q}} \big[ \chi_{[B, \infty)}\big( \max_{t \in [0,T]} S_t \big) \big] 
        , \quad B > s_0.
    \end{equation*}
 For a specific digital option though --- that is, for $B$ fixed --- one can find an (optimal) coupling with the same objective value by using just one intermediate time step \cite[pp. 416--419]{FollmerSchied2016}. In the below examples, we solve problems with $T=2$ computationally; we will recover the coupling constructed in \cite{FollmerSchied2016} as well as the distribution of the maximum of the Azéma--Yor martingale.

\begin{example}[Recovering the optimal coupling from \cite{FollmerSchied2016} when $T=2$] \label{ex:sol_mot_dig1}
    Let $s_0 = 0.5$ and $\mu_T = 0.5(\delta_{0} + \delta_{1})$, then it follows from 
    \cref{eqn:sol_maxofmax_distrmax} that the the robust price of a digital option with barrier $B \in (0.5, 1]$ for the continuous-time problem is $(2B)^{-1}$. On the other hand, \cite[Theorem 7.27]{FollmerSchied2016} yields the same value when taking only discrete time models with $T \ge 2$ into account. Repeating the construction from the proof of the theorem allows for constructing a martingale coupling that attains this upper bound. For this example with $T=2$, it corresponds to take at $t=1$ any of the two values 0 or $B$ with probability $1- (2B)^{-1}$ and  $(2B)^{-1}$, respectively, before moving to the terminal marginal. We note that the three marginals fully specifies a martingale coupling.

     We approximately solve the corresponding tri-marginal MOT problem by application of our framework for $B = 0.75$, $\varepsilon = 0.02$ and $n_1^{\mathcal{S}} = n_2^{\mathcal{S}} = 100$. The corresponding intermediate marginal distribution ($t = 1$) of the price process is $0.27 \delta_{0}  + 0.05 \delta_{0.01} + 0.01 \delta_{0.02} + 0.67 \delta_{0.75}$; this can be compared with $0.33 \delta_0 + 0.67 \delta_{0.75}$, which is given by the theoretical construction. The computed distribution is subject to some smoothing around zero, which is due to regularisation.

     We then repeat the computations for 2, 3, 5, 10 and 15 marginals and investigate what happens with the problem value, the robust price, as the number of marginals increases.
     For $T=1$, the computed price is 0.5, while for larger $T$ it is $0.66$ --- this confirms that it is not possible to increase the problem value by adding further time steps when $T \ge 2$.

\end{example}

\begin{example}[The law of the maximum of the Azéma--Yor martingale via digital options] \label{ex:sol_mot_dig_ay}
    Let $T = 2$ and let the initial and terminal marginals be as in \cref{ex:sol_maxofmax}. We solve the corresponding MOT problem computationally for a set of barriers ranging from 0.5 to 0.92; a regularisation parameter of $\varepsilon = 0.01$ was used and a number of $n_1^{\mathcal{S}} = n_2^{\mathcal{S}} = 67$ points of support. By doing so we obtain the maximum of the value $\mathbb{E}_{\mathbb{Q}} [ \chi_{[B, \infty)} ( \max_{ t \in [T]} S_t ) ] = \mathbb{Q}( \max_{ t \in [T]} S_t \ge B ) $ for each such barrier $B$; the result is shown as a function of the barrier in \cref{fig:dig_maxofmax}, along with the cumulative distribution function of the maximum of $S^{\text{AY}}$ (cf. \cref{ex:sol_maxofmax}). We see that by successively solving for the upper bounds on the probability that digital options with barriers $B \in (0.5, 1]$ are not activated, we can approximately reproduce the cumulative distribution function of the maximum of the Azéma--Yor martingale.

    \begin{figure}
    \centering
    \includegraphics[width=6cm]{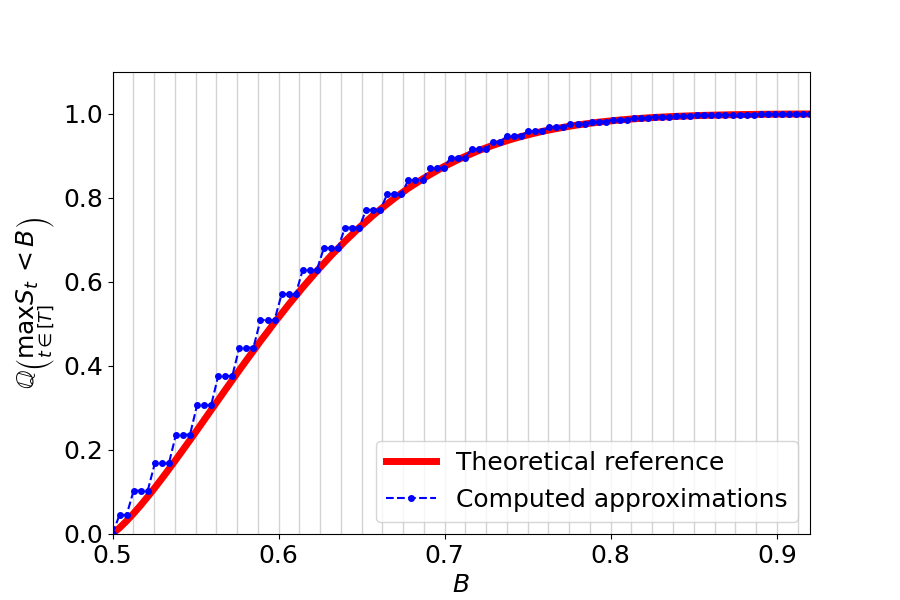}
    \caption{The computed upper bound on the probability that the digital option with barrier $B$ is not activated as a function of the barrier level $B$, as given in \cref{ex:sol_mot_dig_ay}. The cumulative distribution function of the maximum of the corresponding Azéma--Yor martingale is displayed in bold red for reference. The grey vertical lines illustrate the support of the price process contained within $[0.5, 0.92]$. Note that the resolution of the computed approximations follows this grid.}
    \label{fig:dig_maxofmax}
\end{figure}

\end{example}

\subsection{Robust pricing of an Asian option}
We conclude by considering an example where the optimal solution is not known, namely, the robust pricing of an Asian option subject to several marginal constraints. Asian options are a class of financial derivatives that are of great practical financial interest; they are characterised by their payoff being of the form $\phi(S_0, \dots , S_T) = f((T+1)^{-1} \sum_{t \in [T]} S_t )$ for some function $f : \mathbb{R} \rightarrow \mathbb{R}$. The dependence on the arithmetic mean makes the corresponding MOT problem a difficult one to solve analytically --- this was discussed in \cite{Stebegg2014}, who provided the optimal solutions for a problem with two or three given marginals and $f$ convex. As far as we know, it is still unclear what the optimal solution would look like for an MOT problem subject to more than four marginal constraints. By letting the stochastic process $X$ be as given in \cref{ex:examples_mean}, we can however solve these problems computationally. We will now do so for an example where we successively increase the number of marginal constraints. In order to provide visually clear results, all given marginals will be rather simple.

\begin{example}[Pricing a straddle] \label{ex:sol_asian_straddle}
Let the payoff function $\phi$ be as given above with $f(x) = |x - 30|$. Then the payoff corresponds to the payoff of an Asian straddle with strike 30. Let $T = 11$. Consider a trivial initial marginal $\mu_0 = \delta_{30}$ and a uniformly distributed terminal marginal $\mu_{11}$ whose support is contained within $n_{11}^{\mathcal{S}} = 41$ atoms, evenly distributed over the interval $[25, 35]$. In order not to run into problems with rounding errors, we will here use an increasing grid for representing the support of the process $X$; we use a number of points of support ranging from $n_0^{\mathcal{X}} = 1$ to $n_{11}^{\mathcal{X}} = 481$ while $n_t^{\mathcal{S}} = 41$ for $t \ge 1$.

The optimal solution of the corresponding lower bound MOT problem, subject to marginal constraints on the initial and the terminal marginals, was computed using $\varepsilon = 6 \cdot 10^{-3}$. The marginals of the computed optimal coupling are displayed in \cref{fig:sol_asian_couplings1}, where the marginals that are subject to a constraint are marked in green. We note that it corresponds to a late transport, as given by \cite{Stebegg2014}. The smoothing is due to regularisation. A third marginal constraint was then added on the fourth marginal, by requiring that  $\mu_4 = \frac{1}{2}(\delta_{29} + \delta_{31})$. The problem was then solved again; the marginals of the computed optimal coupling are shown in \cref{fig:sol_asian_couplings2}. We note that the transport from the initial marginal $\mu_0$ to the given intermediate marginal $\mu_4$ corresponds to a late transport, as conjectured by \cite{Stebegg2014}. The problem was then solved a third time with a fourth marginal constraint imposed on the eighth marginal, with $\mu_8 = \frac{1}{4}(\delta_{28} + 2 \delta_{30} + \delta_{32})$. The marginals of the optimal coupling are shown in \cref{fig:sol_asian_couplings3}. 

\begin{figure}
    \centering
    \subfloat[$\mathcal{T} = \{0, 11 \}$]{\label{fig:sol_asian_couplings1}\includegraphics[width=4.5cm]{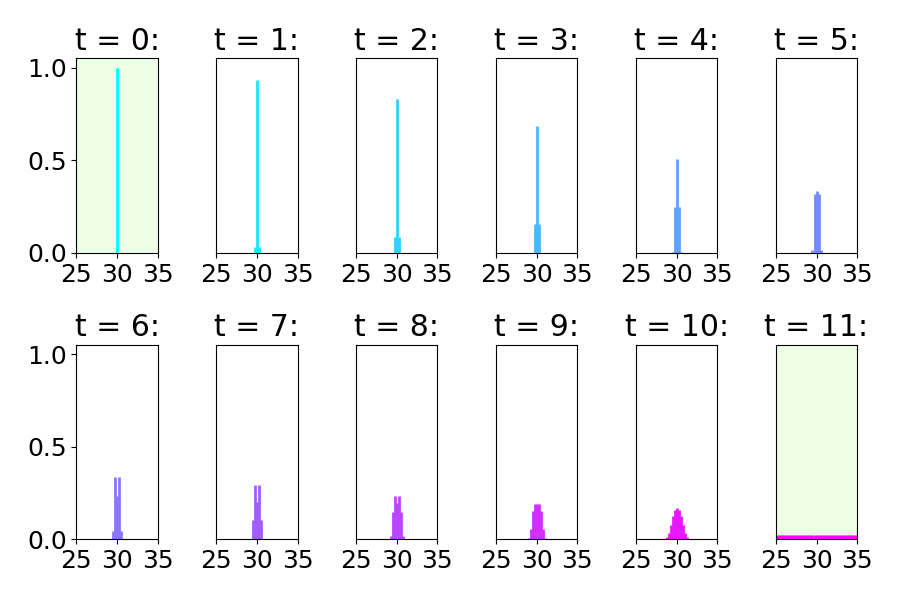}}%
    \subfloat[$\mathcal{T} = \{0, 4, 11 \}$]{\label{fig:sol_asian_couplings2}\includegraphics[width=4.5cm]{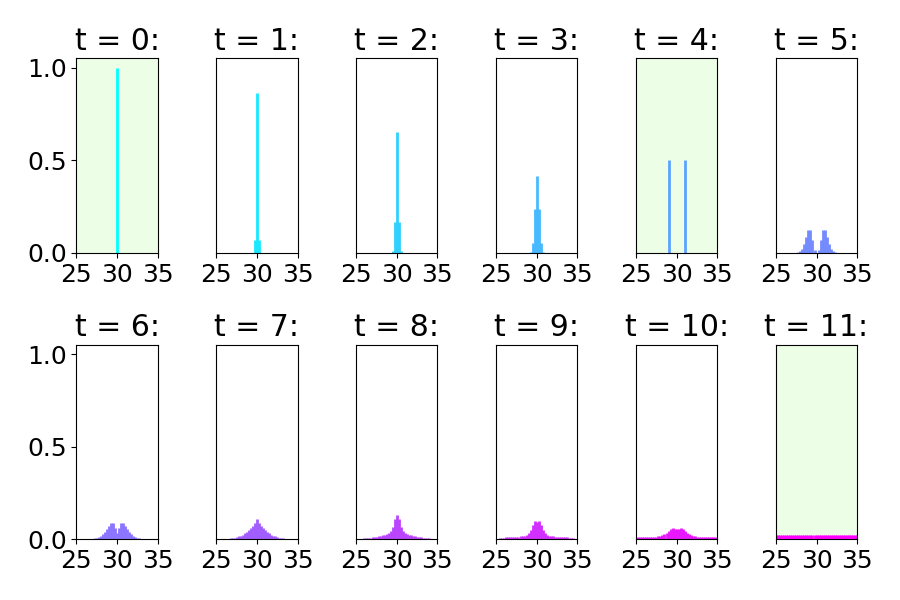}}%
    \subfloat[$\mathcal{T} = \{0, 4, 8, 11 \}$]{\label{fig:sol_asian_couplings3}\includegraphics[width=4.5cm]{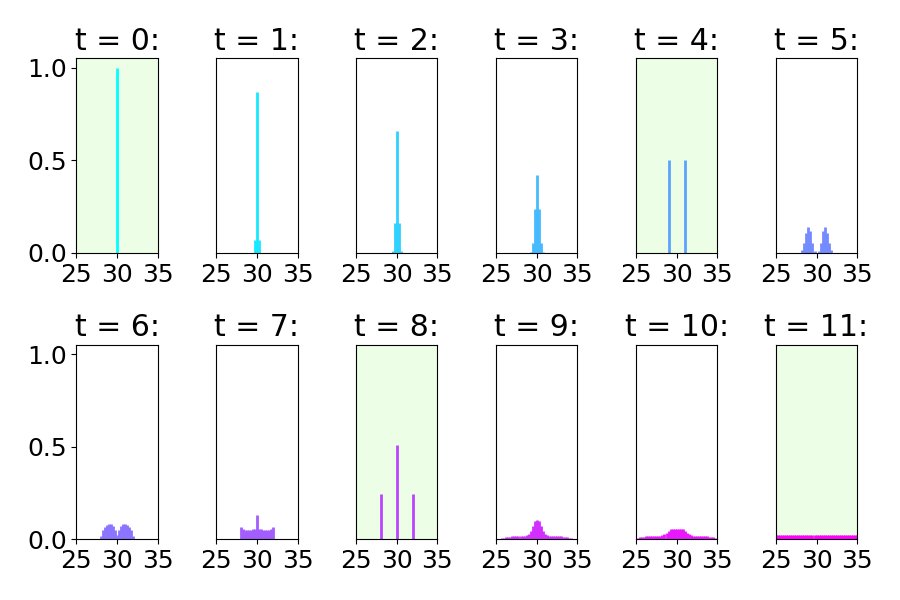}}
    \caption{The marginals of the optimal couplings corresponding to the lower bound MOT problem of \cref{ex:sol_asian_straddle}, subject to marginal constraints on the marginals belonging to $\mathcal T$. Marginals subject to constraints are marked with a green background.}
    \label{fig:sol_asian_couplings}
\end{figure}

\end{example}

\section{Summary and future work}
We have proposed a method for approximately solving a class of multi-marginal MOT problems computationally and demonstrated its utility on a number of different examples. Our examples show that we are able to solve problems over a \emph{large} number of marginals and that the computed solutions are aligned with theoretical results. We have not seen other examples of methods that are able to handle this many marginals in a martingale transport context --- our method is thus, to the best of our knowledge, unique in this sense. One bottleneck is to numerically solve the equations characterising the optimal dual variables corresponding to martingale constraints in the primal problem; if this sub-routine is computationally very heavy, the full algorithm becomes intractably slow. We have here chosen an approach based on Newton's method, since it is in general faster than e.g. gradient methods for relatively simple problems. Our implementation avoids the inversion of the Jacobian matrix, which is typically the most time consuming part of the vectorised version of Newton's method, and it is therefore our belief that our approach is competitive in terms of speed, even though there is a possibility that it could be improved further by exploiting properties of specific problems. Another, related, concern is the propagation of the non-zero errors in the constraints. Since the number of martingale constraints are typically much larger than the number of marginal constraints, an important question is how the residuals in Newton's method accumulates as the number of time steps increases. Our numerical experiments suggests that 30 marginals can be used without this being an issue, but the matter should be investigated further.

\medskip

\bibliographystyle{siam}
\bibliography{References}

\appendix

\section*{Appendix: Additional proofs}  \label{sec:appendix}

\begin{proof}[Proof of \cref{prop:equiv_lp}]
    Fix $\Omega = \bigtimes_{t \in [T]} (\mathcal{S}_t \times \mathcal{X}_t)$ and let the joint process $(S,X)$ be the canonical process on $\Omega$ --- the problem of finding a tuple that optimises problem (\ref{eqn:equiv_mot}) is then equivalent to finding a probability measure $\mathbb{Q}$ on $(\Omega, \bigtimes_{t \in [T]}( 2^{\mathcal{S}_t} \times 2^{\mathcal{X}_t}))$ such that the corresponding tuple is optimal among all tuples of this form. We can then identify an element $\omega$ in $\Omega$ with an index tuple in $\mathcal{I}$ via $(S_t, X_t)(\omega) = (s_t(i_t^{\mathcal{S}}), x_t(i_t^{\mathcal{X}}) )$ for $i_t^{\mathcal{S}} \in \mathcal{I}_t^{\mathcal{S}}, i_t^{\mathcal{X}} \in \mathcal{I}_t^{\mathcal{X}}$ and $t \in [T]$
    and use the probability measure $\mathbb{Q}$ to define a tensor $\mathbf{Q} \in \mathbb{R}_+^{n_0 \times \dots \times n_T}$ as
    \begin{equation} \label{eqn:tensors_Q}
        \mathbf{Q}(i_0, \dots , i_T)
        :=
        \mathbb{Q} \left( \bigcap_{t=0}^T \bigg \{ S_t = s_t(i_t^{\mathcal{S}}), X_t = x_t(i_t^{\mathcal{X}}) \bigg \} \right), 
        \quad (i_0, \dots , i_T) \in \mathcal{I}.
    \end{equation}
    Starting with the objective function,  $\mathbb{E}_{\mathbb{Q}} [ \sum_{t \in [T] \backslash 0} \phi(S_{t-1}, X_{t-1}, S_t, X_t ) ] = \langle \mathbf{\Phi}, \mathbf{Q} \rangle  = \langle \mathbf{C}, \mathbf{Q} \rangle $,
    where the first equality follows from the definition of the tensors $\mathbf{Q}$ and $\mathbf{\Phi}$ and the last equality from the fact that $\langle \mathbf{C}, \mathbf{Q} \rangle < \infty$ for any tensor $\mathbf{Q}$ defined as in \cref{eqn:tensors_Q}. As for the marginal constraints (\ref{eqn:equiv_mot_margconstr}), it follows immediately from the definition of $\mathbf{Q}$ that
    $\mathbb{Q} (S_t = s_t(j) ) = P_t^{\mathcal{S}}(\mathbf{Q})(j)$ for $j \in \mathcal{I}_t^{\mathcal{S}}$ and $t \in [T]$.
    Combining this elementwise with equation (\ref{eqn:tensors_mu}) and the constraint (\ref{eqn:equiv_mot_margconstr}) yields the constraint (\ref{eqn:tensors_mot_unreg_margconstr}). Similarly, for $t \in [T] \backslash 0$, it holds that 
    \begin{equation*} \label{eqn:tensors_mtg}
        \begin{aligned}
            \mathbb{E}_{\mathbb{Q}} \left [  S_t | S_{t-1} = s_{t-1}(i_{t-1}^{\mathcal{S}}), X_{t-1} = x_{t-1}(i_{t-1}^{\mathcal{X}}) \right]
            &=
            \frac{\big( P_{t-1, t}(\mathbf{Q}) 
            ( 1_{n_t^{\mathcal{X}}} \otimes s_t ) \big)
            \left( i_{t-1} \right)}{P_{t-1} \left( \mathbf{Q} \right)\left( i_{t-1} \right)},
            \quad i_{t-1} \in \mathcal{I}_{t-1},
        \end{aligned}
    \end{equation*}
    since $s_t(i_t^{\mathcal{S}}) = (1_{n_t^{\mathcal{X}}} \otimes s_t)(i_t)$. By combining the martingale constraint  (\ref{eqn:equiv_mot_mtgconstr}) elementwise with this we obtain $ P_{t-1, t}(\mathbf{Q}) ( 1_{n_t^{\mathcal{X}}} \otimes s_t ) = ( 1_{n_{t-1}^{\mathcal{X}}} \otimes s_{t-1} ) \odot P_{t-1} ( \mathbf{Q} )$ for $t \in [T] \backslash 0$.
    This is the martingale constraint (\ref{eqn:equiv_mot_mtgconstr}) formulated on vector form. By subtracting the right-hand side from the left-hand side and identifying the matrix $\Delta_t \in \mathbb{R}^{n_{t-1} \times n_t}$ the constraint (\ref{eqn:tensors_mot_unreg_mtgconstr}) follows. 

    Conversely, let $\mathbf{Q} \in \mathbb{R}_+^{n_0 \times \dots \times n_T}$ be feasible to problem (\ref{eqn:tensors_mot_unreg}). Then $(\mathcal{I}, \bigtimes_{t\in [T]} ( 2^{\mathcal{I}_t^{\mathcal{S}}} \times  2^{\mathcal{I}_t^{\mathcal{X}}} ), \mathbf{Q})$ is a probability space, on which we can identify $\mathbf{\Phi}$ as a random variable. Assume that $\mathbf{Q}$ is such that $\langle \mathbf{C}, \mathbf{Q} \rangle < \infty$.  Then $\langle\mathbf{C}, \mathbf{Q} \rangle = \langle \mathbf{\Phi}, \mathbf{Q} \rangle = \mathbb{E}_\mathbf{Q} [ \mathbf{\Phi}]$,
    which means that the values of the objective functions (\ref{eqn:tensors_mot_unreg_objfcn}) and (\ref{eqn:equiv_mot_objfcn}) coincide. Also note that the model, induced by the tensor $\mathbf{Q}$ satisfying the constraints (\ref{eqn:tensors_mot_unreg_mtgconstr}) and (\ref{eqn:tensors_mot_unreg_margconstr}) and by the mapping $(i,t) \mapsto s_t(i_t^{\mathcal{S}})$ for $t \in [T]$ and $i \in \mathcal{I}$, satisfies the constraints (\ref{eqn:equiv_mot_mtgconstr}) and (\ref{eqn:equiv_mot_margconstr}). It remains to argue that there always exists a $\mathbf{Q}$ such that $\langle \mathbf{C}, \mathbf{Q} \rangle < \infty$ whenever the feasible set of problem (\ref{eqn:tensors_mot_unreg}) is nonempty. Do so by noting that the feasible set of problem (\ref{eqn:equiv_mot}) is nonempty and that every feasible model $(\Omega, \mathcal F, \mathbb{Q}, S)$ induces $\mathbf{Q}$ such that $\langle \mathbf{C}, \mathbf{Q} \rangle < \infty$. This completes the proof.
\end{proof}

\begin{proof}[Proof of \cref{prop:convergence}]
     Let the feasible set of problems (\ref{eqn:tensors_mot_unreg}) and (\ref{eqn:tensors_mot}) be denoted $\mathfrak{F}$, that is, let 
    \begin{align*}
       \mathfrak{F}
       := 
       \big \{ \mathbf{Q} \in \mathbb{R}_+^{n_0 \times \dots \times n_T} : \left( P_{t-1, t}(\mathbf{Q}) \odot \Delta_t \right) 1_{n_t} =  0_{n_{t-1}}, \; t \in [T], \quad  P_t^{\mathcal{S}}(\mathbf{Q}) = m_t, \; t \in \mathcal{T} \big \}.
    \end{align*}
    Then note that for each $k$ fixed the regularised problem (\ref{eqn:tensors_mot}) with $\varepsilon = \varepsilon_k$ is strictly convex. Therefore, the optimal solution $\mathbf{Q}_k$ is unique and hence the sequence $(\mathbf{Q}_k)_k$ is well-defined. Compactness of $\mathfrak{F}$ yields that it has a convergent subsequence whose limit belongs to $\mathfrak{F}$ --- let the limit be denoted $\mathbf{Q}_{\infty}$. 
    
    We now show that $\mathbf{Q}_{\infty}$ minimises the unregularised problem (\ref{eqn:tensors_mot_unreg}). In order to do so, let $\mathbf{Q}^*$ be a minimiser of problem (\ref{eqn:tensors_mot_unreg}),  fix $k$ and note that by definition of $\mathbf{Q}_k$
    \begin{align*}
        \langle \mathbf{C}, \mathbf{Q}_k \rangle  + \varepsilon_k D(\mathbf{Q}_k) 
        \le 
        \langle \mathbf{C}, \mathbf{Q} \rangle + \varepsilon_k D(\mathbf{Q}) 
        \le \langle \mathbf{C}, \mathbf{Q} \rangle, \quad \mathbf{Q} \in \mathfrak{F},
    \end{align*}
    where we have used that $-\prod_{t \in [T]} n_t \le D(\mathbf{Q}) \le 0$ for $\mathbf{Q} \in \mathfrak{F}$. Choosing $\mathbf{Q}^*$ as $\mathbf{Q}$ in the above gives that
    \begin{align*}
        \underset{k \rightarrow \infty}{\limsup \;}  \langle \mathbf{C}, \mathbf{Q}_k \rangle 
        = 
        \underset{k \rightarrow \infty}{\limsup \;}  \langle \mathbf{C}, \mathbf{Q}_k \rangle  
        + \varepsilon_k D(\mathbf{Q}_k) 
        \le 
        \langle \mathbf{C}, \mathbf{Q}^* \rangle.
    \end{align*}
    On the other hand, trivially, $\liminf_{k \rightarrow \infty} \langle \mathbf{C}, \mathbf{Q}_k \rangle \ge \langle \mathbf{C},\mathbf{Q}^* \rangle$ and hence we get $\langle \mathbf{C}, \mathbf{Q}_\infty \rangle = \langle \mathbf{C},\mathbf{Q}^* \rangle$ by continuity of the cost function $\langle \mathbf{C}, \cdot \rangle $ --- this shows that the limit $\mathbf{Q}_\infty$ is indeed a minimiser of the unregularised problem (\ref{eqn:tensors_mot_unreg}). Convergence of the value of the regularised problem (\ref{eqn:tensors_mot}) to the value of problem (\ref{eqn:tensors_mot_unreg}) as the regularisation parameter vanishes then follows immedately from continuity of the cost function and from continuity and boundedness on $\mathfrak{F}$ of the entropy term $D$.
\end{proof}

\begin{proof}[Proof of \cref{lemma:algo_interiorpt}]
    The domain of problem (\ref{eqn:tensors_mot}) is 
    \begin{align*}
        \mathfrak{D} := \big \{ \mathbf{Q} \in \mathbb{R}_+^{n_0 \times \dots \times n_T} : \langle \mathbf{C}, \mathbf{Q} \rangle < \infty \text{ and the constraints (\ref{eqn:tensors_mot_mtgconstr}) and (\ref{eqn:tensors_mot_margconstr}) hold} \big \}.
    \end{align*}
    To show that Slater's condition holds, we want to show that $\mathfrak{D}$ has a non-empty relative interior, that is, that there exists a tensor $\mathbf{Q^{\text{ri}}}$ feasible to problem (\ref{eqn:tensors_mot}) such that $\mathbf{Q^{\text{ri}}}(i_0, \dots , i_T) > 0$ whenever $\mathbf{C}(i_0, \dots, i_T) < \infty$ for $(i_0, \dots , i_T) \in \mathcal{I}$.

    Let $\{ \mu_t \}_{t \in [T] \backslash \mathcal{T}}$ be intermediate marginals satisfying the assumptions of the second part of \cref{thm:algo} and let for $t \in [T] \backslash \mathcal{T}$ the vector representing $\mu_t$ via \cref{eqn:tensors_mu} be denoted by $m_t$. Then define for any $t \in [T]\backslash 0$ 
    \begin{align*}
        \mathfrak{M}_{t} := \big \{ Q \in \mathbb{R}^{n_{t-1}^{\mathcal{S}} \times n_t^{\mathcal{S}}}_+ : Q s_t = s_{t-1} \odot (Q 1_{n_{t}^{\mathcal{S}}}), \; Q 1_{n_t^{\mathcal{S}}} = m_{t-1}, \; Q^{\top} 1_{n_{t-1}^{\mathcal{S}}} = m_t \big \}
    \end{align*}
    --- note that it is a non-empty convex set and that an element $Q \in \mathfrak{M}_t$ is the matrix representation of a bi-marginal martingale subtransport between the marginals $\mu_{t-1}$ and $\mu_t$. Since $(\mu_{t-1}, \mu_t)$ is irreducible, it follows from  \cite[Theorem 3.2]{BeiglbockNutzTouzi2017} that there exists no index tuple $(i,j) \in \mathcal{I}_{t-1}^{\mathcal{S}} \times \mathcal{I}_t^{\mathcal{S}}$ such that $Q(i,j) = 0$ for every $Q \in \mathfrak{M}_t$. Therefore there exists $Q^{\mathcal{S}}_t \in \mathfrak{M}_t$ strictly positive; see this by forming a convex combination of a set of matrices in $\mathfrak{M}_t$ whose zero elements do not coincide. 

    Now extend the above to bi-marginal martingale transports across the two-dimensional state space by defining, for $t \in [T] \backslash 0$, a matrix $Q_t \in \mathbb{R}^{n_{t-1} \times n_t}$ by
    \begin{align*}
        Q_t(i_{t-1}, i_t)
        :=
        \begin{cases}
             Q_1^{\mathcal{S}} (i_0^{\mathcal{S}}, i_1^{\mathcal{S}}), & x_0 (i_0^{\mathcal{X}}) = h_0( i_0^{\mathcal{S}} ) \text{ and } x_1( i_1^{\mathcal{X}} ) = h_1( s_1(i_1^{\mathcal{S}}), s_{0}(i_{0}^{\mathcal{S}}), x_{0}(i_{0}^{\mathcal{X}})), \; t = 1 \\
             Q_t^{\mathcal{S}} (i_{t-1}^{\mathcal{S}}, i_t^{\mathcal{S}}), \hspace{-.2cm}&  x_t( i_t^{\mathcal{X}} ) = h_t( s_t(i_t^{\mathcal{S}}), s_{t-1}(i_{t-1}^{\mathcal{S}}), x_{t-1}(i_{t-1}^{\mathcal{X}})), \; t = 2, \dots , T \\
             0, & \text{else}. 
        \end{cases}
    \end{align*}
    Note that the above definition implies that there is for every index pair $(i_{t-1}^{\mathcal{S}}, i_t^{\mathcal{S}}) \in \mathcal{I}_{t-1}^{\mathcal{S}} \times \mathcal{I}_t^{\mathcal{S}}$ one unique pair $(i_{t-1}^{\mathcal{X}}, i_t^{\mathcal{X}}) \in \mathcal{I}_{t-1}^{\mathcal{X}} \times \mathcal{I}_t^{\mathcal{X}}$ such that $Q_t(i_{t-1}, i_t)  > 0$. It follows from this, from  $Q_t^{\mathcal{S}} \in \mathfrak{M}_t$ and from the definition of the matrix $\Delta_t$ that the matrix $Q_t$ given above satisfies 
    \begin{equation}
        \begin{aligned} \label{eqn:algo_lemma_interiorpt_pf1}
            \big( Q_t \odot \Delta_t \big) 1_{n_t} = 0_{n_{t-1}}, \quad
            P^{\mathcal{S}_{t}} \big(Q_t^{\top} 1_{n_{t-1}} \big) = m_{t}, \quad
            P^{\mathcal{S}_{t-1}} \big(Q_t 1_{n_t} \big) = m_{t-1}
        \end{aligned}
    \end{equation}
    and hence the extension works as intended. Then define a tensor $\mathbf{Q^{\text{ri}}}$ by 
    \begin{align*}
        \mathbf{Q^{\text{ri}}}(i_0, \dots, i_T)
        :=  
        m_0(i_0^{\mathcal{S}})\prod_{t \in [T] \backslash 0} Q_t(i_{t-1}, i_t) / m_{t-1}(i_{t-1}^{\mathcal{S}}), \quad  (i_0, \dots, i_T) \in \mathcal{I},
    \end{align*}
    where we remark that it follows from the assumptions that the vectors $m_t$ for $t \in [T]$ have no zero elements. It is immediate from the construction that $\mathbf{Q^{\text{ri}}} \in \mathbb{R}_+^{n_0 \times \dots \times n_T}$ and that the bi-marginal projections evaluates to $ P_{t-1, t}(\mathbf{Q^{\text{ri}}})(i_{t-1}, i_t) = Q_t(i_{t-1}, i_t)$ for $i_{t-1} \in \mathcal{I}_{t-1}, i_t \in \mathcal{I}_t$ and $t \in [T] \backslash 0$.
    Combining this with \cref{eqn:algo_lemma_interiorpt_pf1} yields that the tensor $\mathbf{Q}^{\text{ri}}$ satisfies the constraints (\ref{eqn:tensors_mot_mtgconstr}) and (\ref{eqn:tensors_mot_margconstr}). It remains to show that the tensor $\mathbf{Q^{\text{ri}}}$ belongs to the relative interior of the domain $\mathfrak{D}$. Do so by noting that $\mathbf{Q^{\text{ri}}}(i_0, \dots , i_T) > 0$ if and only if $(i_0, \dots, i_T) \in \mathcal{I}$ is such that $x_0 (i_0^{\mathcal{X}}) = h_0( i_0^{\mathcal{S}} )$, $x_1( i_1^{\mathcal{X}} ) = h_1( s_1(i_1^{\mathcal{S}}), s_{0}(i_{0}^{\mathcal{S}}), x_{0}(i_{0}^{\mathcal{X}}))$ and $x_t( i_t^{\mathcal{X}} ) = h_t( s_t(i_t^{\mathcal{S}}), s_{t-1}(i_{t-1}^{\mathcal{S}}), x_{t-1}(i_{t-1}^{\mathcal{X}}))$ for $t \in [T] \backslash \{0,1\}$ --- that is, if and only if the index tuple $(i_0, \dots, i_T) \in \mathcal{I}$ is such that the corresponding realisations of the joint process respect \cref{eqn:intro_x}. Consequently, if $(i_0, \dots , i_T) \in \mathcal{I}$ is such that $\mathbf{Q^{\text{ri}}}(i_0, \dots , i_T) = 0$, then $\mathbf{C}(i_0, \dots , i_T) = \infty$. Therefore, $\mathbf{Q^{\text{ri}}}$ does indeed belong to the relative interior of $\mathfrak{D}$ and the result follows.
\end{proof}

\begin{proof}[Proof of \cref{thm:algo}]
We start by proving the second part of the assertion, that is, the part that is provided under additional assumptions on the given marginals. Since the entropy term $D(\textbf{Q})$ is strictly convex in $\mathbf{Q}$, problem (\ref{eqn:tensors_mot}) is a strictly convex optimisation problem whose domain, by \cref{lemma:algo_interiorpt}, has a non-empty relative interior. That is, Slater's condition holds and thus strong Lagrangian duality holds \cite[pp. 226--227]{BoydVandenberghe2004}; that is, the duality gap is zero, the dual supremum is attained and the primal minimum can be computed by first solving for a dual maximiser.

In order to derive the Lagrangian dual of problem (\ref{eqn:tensors_mot}), we start by forming the corresponding Lagrangian function. It reads
\begin{equation*}
\begin{aligned}
    \mathcal{L}(\mathbf{Q}, \lambda, \gamma ) 
    = 
    \langle \mathbf{C} , \mathbf{Q} \rangle + \varepsilon D(\mathbf{Q}) 
    + \sum_{t \in \mathcal T} \big( \lambda_t^{\top} ( m_t - P^{\mathcal{S}}_t(\mathbf{Q}) )  \big)
    - \sum_{t \in [T] \backslash 0} \big( \gamma_{t-1}^{\top}  
     ( P_{t-1, t}(\mathbf{Q}) \odot \Delta_t ) 1_{n_t}
      \big),
\end{aligned}
\end{equation*}
where $ \lambda = \{\lambda_t \}_{t \in \mathcal T}$ are Lagrangian multipliers (dual variables) corresponding to the marginal constraints (\ref{eqn:tensors_mot_margconstr}) and $\gamma = \{ \gamma_t \}_{t\in [T-1]}$ are Lagrangian multipliers corresponding to the martingale constraints (\ref{eqn:tensors_mot_mtgconstr}). Note that each $\lambda_t$ is a vector of length $n_t^{\mathcal{S}}$, while each $\gamma_t$ is a vector of length $n_t$. The dual of problem (\ref{eqn:tensors_mot}) is then
\begin{equation*} 
    \underset{\lambda, \gamma }{\max} \quad \varphi(\lambda, \gamma), 
\end{equation*}
 since dual variables corresponding to primal equality constraints are unconstrained (see, for example, \cite[p. 470]{Nash-Sofer}), where the dual functional is defined as $\varphi(\lambda, \gamma) := \inf \{\mathcal{L}(\mathbf{Q}, \lambda, \gamma) : \mathbf{Q} \in \mathbb{R}_+^{n_0 \times \dots \times n_T} \}$ --- note that for each fixed pair of dual variables $(\lambda, \gamma)$ the Lagrangian is continuous and coercive in $\mathbf{Q} \in \mathbb{R}_+^{n_0 \times \dots \times n_T}$, hence the infimum is attained.

We proceed by finding the minimising $\mathbf{Q}^{\lambda, \gamma}$ for given dual variables $(\lambda, \gamma)$. Since the Lagrangian function $\mathcal{L}$ is strictly convex in $\mathbf{Q}$ and the entropy term $D(\mathbf{Q})$ prevents optimal solutions at the boundary, the minimising transport plan satisfies
$\partial_{\mathbf{Q}(i_0, \dots , i_T)}\mathcal{L}(\mathbf{Q}, \lambda, \gamma) = 0$ for $(i_0, \dots , i_T) \in \mathcal{I}$. By differentiating the Lagrangian function we obtain that for any $(i_0, \dots , i_T) \in \mathcal{I}$,
\begin{equation*}
    \begin{aligned}
        \partial_{\mathbf{Q}(i_0, \dots , i_T)}& \mathcal{L}(\mathbf{Q}, \lambda, \gamma) \\
        =&
        \varepsilon \log \mathbf{Q}(i_0, \dots , i_T) 
        - \sum_{t \in \mathcal T} \lambda_t(i_t^{\mathcal{S}})
        +
        \sum_{t \in [T] \backslash 0} 
        \big( C_t(i_{t-1}, i_t)
        - \gamma_{t-1}(i_{t-1}) \Delta_t(i_{t-1}, i_t) \big)
    \end{aligned}
\end{equation*}
which yields a minimising transport plan 
\begin{equation*} \label{eqn:algo_Popt}
    \mathbf{Q}^{\lambda, \gamma }(i_0, \dots , i_T)
    =
     \Big( \prod_{t \in \mathcal{T}}  e^{\lambda_t(i_t^{\mathcal{S}})/ \varepsilon} \Big)
    \Big( \prod_{t \in [T] \backslash 0}  e^{-C_t(i_{t-1}, i_t)/ \varepsilon}
    e^{\gamma_{t-1}(i_{t-1})\Delta_t(i_{t-1}, i_t)/ \varepsilon} \Big)
    , \; (i_0, \dots , i_T) \in \mathcal{I},
\end{equation*}
which can be written as in \cref{eqn:algo_Popt2}. Note that the claim that the martingale transport plan $\mathbf{Q}^{\lambda^*, \gamma^*}$ assigns $\mathbf{Q}^{\lambda^*, \gamma^*}(i_0, \dots , i_T) > 0$ for $(i_0, \dots , i_T) \in \mathcal{I}$ such that $\mathbf{K}(i_0, \dots , i_T) > 0$ now automatically follows from the finiteness of the optimal dual variables $(\lambda^*, \gamma^*)$ in combination with the definition of the tensors $\mathbf{K}$, $\mathbf{U}^{\lambda^*}$ and $\mathbf{G}^{\gamma^*}$.

For this choice of $\mathbf{Q}$, the first term of the scaled entropy $\varepsilon D(\mathbf{Q}^{\lambda, \gamma})$ equals
\begin{equation*}
    \begin{aligned}
         \varepsilon 
         &\sum_{(i_0, \dots , i_T) \in \mathcal{I}} 
          \mathbf{Q}^{\lambda, \gamma}(i_0, \dots , i_T) \log  \mathbf{Q}^{\lambda, \gamma}(i_0, \dots , i_T) \\
        &=
        \sum_{(i_0, \dots , i_T) \in \mathcal{I}}
        \mathbf{Q}^{\lambda, \gamma}(i_0, \dots , i_T) 
            \Big(  
                \sum_{t \in \mathcal T }  \lambda_t(i_t^{\mathcal{S}}) 
                - \sum_{t \in [T] \backslash 0}
                    \left(
                      C_t(i_{t-1}, i_t)
                     -
                    \gamma_{t-1}(i_{t-1})
                    \Delta_t(i_{t-1}, i_t)
                    \right)
            \Big) \\
            &=
            \sum_{t \in \mathcal T} \lambda_t^{\top} P_t^{\mathcal{S}}(\mathbf{Q}^{\lambda, \gamma})
            - \langle \mathbf{C}, \mathbf{Q}^{\lambda, \gamma} \rangle
            + \sum_{t \in [T] \backslash 0} 
            \gamma_{t-1}^{\top}
              \big( P_{t-1, t}(\mathbf{Q}^{\lambda, \gamma}) \odot \Delta_t \big) 1_{n_t}   
            .
    \end{aligned}
\end{equation*}
Inserting $\mathbf{Q}^{\lambda, \gamma}$ into $\mathcal{L}( \cdot, \lambda, \gamma)$ thus yields that the objective of the dual of problem (\ref{eqn:tensors_mot}) becomes 
\begin{align*}
    \varphi(\lambda, \gamma) &= 
    \sum_{t \in \mathcal{T}}  \lambda_t^{\top} m_t 
    - \varepsilon 
     \sum_{(i_0, \dots , i_T) \in \mathcal{I}}  
        \Big(
        \prod_{t \in \mathcal{T}} e^{\lambda_t(i_t^{\mathcal{S}})/ \varepsilon} 
        \Big)
        \Big( 
        \prod_{t \in [T] \backslash 0}  e^{-C_t(i_{t-1}, i_t)/ \varepsilon}
    e^{\gamma_{t-1}(i_{t-1}) \Delta_t (i_{t-1}, i_t)/ \varepsilon}  \Big) \\
    &= \sum_{t \in \mathcal{T}} \lambda_t^{\top} m_t - \varepsilon \langle \mathbf{K}, \mathbf{U}^{\lambda} \odot \mathbf{G}^{\gamma} \rangle,
\end{align*}
and problem (\ref{eqn:dualproblem}) then follows.

We now move on to finding the strongest relaxation, or equivalently, to finding the dual variables $(\lambda^*, \gamma^*)$ such that the smooth and concave functional $\varphi$ is maximised. Such $(\lambda^*, \gamma^*)$ are stationary points of $\varphi$ and hence they satisfy the following set of equations
\begin{subequations} \label{eqn:algo_stationaryphi}
\begin{align}
    &\partial_{\lambda_t(j)} \varphi (\lambda, \gamma) = 0 , \quad t \in \mathcal{T}, \quad j \in \mathcal{I}_t^{\mathcal{S}}  \label{eqn:algo_stationaryphi_lambda}\\
    &\partial_{\gamma_t(j)} \varphi (\lambda, \gamma) = 0 , \quad t \in [T-1], \quad j \in \mathcal{I}_t. \label{eqn:algo_stationaryphi_gamma}
\end{align}
\end{subequations}
Starting with the dual variables corresponding to marginal constraints in the primal problem, we differentiate and obtain
\begin{equation*}
    \begin{aligned}
        \partial_{\lambda_t(j)} \varphi(\lambda, \gamma)
        =&
        m_t(j)
        -
        \sum_{\substack{(i_0, \dots , i_T) \in \mathcal{I} : \\ i_t^{\mathcal{S}} = j}}
        \mathbf{K}(i_0, \dots , i_T)
        \mathbf{G}^{\gamma}(i_0, \dots , i_T)
            \prod_{k \in \mathcal T}
              e^{\lambda_k(i_k^{\mathcal{S}})/ \varepsilon }
            ,  \quad t \in \mathcal{T}, \; j \in \mathcal{I}_t^{\mathcal{S}}.
            \quad 
    \end{aligned}
\end{equation*}
Solving for $u_t^{\lambda_t} = \exp(\lambda_t/ \varepsilon)$  thus yields \cref{eqn:algo_updateu1}. Moving on to the dual variables corresponding to martingale constraints in the primal problem, we differentiate and get
\begin{equation*}
    \begin{aligned}
        \partial_{\gamma_{t}(j)} \varphi(\lambda, \gamma)
        =
        - \sum_{\substack{(i_0, \dots , i_T) \in \mathcal{I} : \\ i_t = j}}
        \Delta_{t+1}(j, i_{t+1})
        \mathbf{K}(i_0, \dots , i_T)
        \mathbf{U}^{\lambda}(i_0, \dots , i_T)
        \prod_{k \in [T] \backslash 0}
            e^{\gamma_{k-1}(i_{k-1}) \Delta_k (i_{k-1}, i_k)/ \varepsilon },
    \end{aligned}
\end{equation*}
for $t \in  [T-1]$ and $j \in \mathcal{I}_t$. When solving for $\gamma$ such that $\partial_{\gamma_{t}(j)} \varphi(\lambda, \gamma) = 0$ for all $t$ and $j$ as given above, we do not get any explicit formulas. Instead we obtain for every $t \in [T-1]$,  $n_t$ equations in $\gamma_{t}(j)$, $j \in \mathcal{I}_t$. They can be written on vector form as in \cref{eqn:algo_updategamma1}.

We now turn to showing the first part of the assertion; as always, we assume that the given marginals $\{ \mu_t \}_{t \in \mathcal{T}}$ are in convex order. We proceed in a similar manner as in the proof of \cite[Theorem 1.1]{Beiglbock-HL-Penkner2013}. Let $\mathbf{Q} \in \mathbb{R}_+^{n_0 \times \dots \times n_T}$ be a tensor satisfying the marginal constraint (\ref{eqn:tensors_mot_margconstr}) and note that
    \begin{align*}
        \gamma_{t-1}^{\top} \big( P_{t-1,t}(\mathbf{Q}) \odot \Delta_t \big)1_{n_t} = 0, \quad \gamma_{t-1} \in \mathbb{R}^{n_{t-1}}, \quad t \in [T] \backslash 0
    \end{align*}
    holds if and only if $\mathbf{Q}$ also satisfies the martingale constraint (\ref{eqn:tensors_mot_mtgconstr}). This implies that if $\mathbf{Q}$ is not feasible to problem (\ref{eqn:tensors_mot}), there exists $t \in [T] \backslash 0$ and $j \in \mathcal{I}_{t-1}$ such that $\sum_{k \in \mathcal{I}_t}(P_{t-1,t}(\mathbf{Q}) \odot \Delta_t )(j,k) \neq 0$ and hence $-\gamma_{t-1}^{\top} ( P_{t-1,t}(\mathbf{Q}) \odot \Delta_t )1_{n_t}$ can be made arbitrarily large by scaling the vector $\gamma_{t-1}$. Therefore, the value of the problem 
    \begin{align} \label{eqn:tensors_mot_dualitypf1}
        \underset{\substack{ \mathbf{Q} \in \mathbb{R}_+^{n_0 \times \dots \times n_T} : \\ P_t^{\mathcal{S}}(\mathbf{Q}) = m_t, \; t \in \mathcal{T}}}{\inf \;} \underset{\gamma }{\sup \;} \langle \mathbf{C}, \mathbf{Q} \rangle + \varepsilon D(\mathbf{Q}) - \sum_{t \in [T] \backslash 0} \gamma_{t-1}^{\top} \big( P_{t-1, t}(\mathbf{Q}) \odot \Delta_t \big) 1_{n_t}
    \end{align}
     equals the value of problem (\ref{eqn:tensors_mot}). Moreover, it satisfies the assumptions of the minimax theorem \cite[Theorem 45.8]{strasser1985}; indeed, the subset of non-negative tensors in $\mathbb{R}^{n_0 \times \dots \times n_T}$ satisfying the marginal constraint (\ref{eqn:tensors_mot_margconstr}) is compact and convex and the objective is continuous and convex in $\mathbf{Q}$ and linear in $\gamma$. Hence the infimum and supremum can be interchanged and problem (\ref{eqn:tensors_mot_dualitypf1}) thus equals 
    \begin{align} \label{eqn:tensors_mot_dualitypf2}
         \underset{\gamma }{\sup \;} 
         \underset{\substack{ \mathbf{Q} \in \mathbb{R}_+^{n_0 \times \dots \times n_T} : \\ P_t^{\mathcal{S}}(\mathbf{Q}) = m_t, \; t \in \mathcal{T}}}{\inf \;}
         \langle \mathbf{C}, \mathbf{Q} \rangle + \varepsilon D(\mathbf{Q}) - \sum_{t \in [T] \backslash 0} \gamma_{t-1}^{\top} \big( P_{t-1, t}(\mathbf{Q}) \odot \Delta_t \big) 1_{n_t}.
    \end{align}
    Now consider the inner minimisation problem in (\ref{eqn:tensors_mot_dualitypf1}) for a fixed family $\gamma$ of real-valued vectors. It can be shown that Slater's condition holds for this problem --- the proof is analogous to the proof of \cref{lemma:algo_interiorpt} but invokes \cite[Proposition 3.1]{BeiglbockNutzTouzi2017} to show that there exists strictly positive matrices representing the bi-marginal (non-martingale) sub-transports between two adjacent marginals --- and hence strong Lagrangian duality holds \cite[pp. 226--227]{BoydVandenberghe2004}. The Lagrangian dual of the problem is
    $\max_{\lambda}\; \sum_{t\in \mathcal{T}} \lambda_t^{\top} m_t- \varepsilon \langle \mathbf{K}, \mathbf{U}^{\lambda} \odot \mathbf{G}^{\gamma} \rangle$
    --- the derivation is similar to the derivation of the Lagrangian dual of problem (\ref{eqn:tensors_mot}). It follows that problem (\ref{eqn:tensors_mot_dualitypf2}) equals problem (\ref{eqn:dualproblem}). Putting things together, we have thus shown that the value of problem (\ref{eqn:tensors_mot}) equals the value of problem (\ref{eqn:dualproblem}). Hence problem (\ref{eqn:dualproblem}) is a dual of problem (\ref{eqn:tensors_mot}) and strong duality between the two holds. 
\end{proof}

\begin{proof}[Proof of \cref{lemma:algo_filipisabell_applied}]
     Start by noting that $ \mathbf{U}^{\lambda}$ is of the form $\mathbf{U}^{\lambda}(i_0, \dots, i_T)
         = \prod_{t \in [T]} \bar u_t(i_t)$ for $(i_0, \dots, i_T) \in \mathcal{I}$
     since, when $i_t \in \mathcal{I}_t$ is given by the order (\ref{eqn:tensors_idxorder}), $\bar u_t(i_t) = u_t^{\lambda_t}(i_t^{\mathcal{S}})$ for $t \in \mathcal T$. Also note that $ \mathbf{K} \odot  \mathbf{G}^{\gamma}$ inherits the structure of $ \mathbf{K}$ and $ \mathbf{G}^{\gamma}$; indeed,
     $(  \mathbf{K} \odot  \mathbf{G}^{\gamma})(i_0, \dots, i_T)
        = \prod_{t \in [T] \backslash 0} (  K_t \odot G^{\gamma_{t-1}}_t ) (i_{t-1}, i_t)$ for $ (i_0, \dots, i_T) \in \mathcal{I}$.
    The requirements of \cref{prop:algo_filipisabell} are thus satisfied and it yields expressions for the projections. Specifically, taking  
    \begin{equation*}
        \begin{aligned}
            \hat{\psi}_t :=
            \begin{cases}
                1_{n_0}, & t = 0 \\
                ( K_t \odot G^{\gamma_{t-1}}_t )^{\top}
                \diag( \bar u_{t-1})
                \dots
                (  K_{1} \odot G^{\gamma_0}_{1} )^{\top}
                \bar u_{0}, & t \in [T] \backslash 0
            \end{cases}
        \end{aligned}
    \end{equation*}
    and 
    \begin{equation*}
        \begin{aligned}
            \psi_t :=
            \begin{cases}
                1_{n_T}, & t = T \\
                ( K_{t+1} \odot G^{\gamma_t}_{t+1} )
                \diag( \bar u_{t+1})
                \dots
                (  K_{T} \odot G^{\gamma_{T-1}}_{T} )
                \bar u_{T}, & t \in [T-1]
            \end{cases}
        \end{aligned}
    \end{equation*}
    we get  $P_t  (  \mathbf{K} \odot  \mathbf{U}^{\lambda} \odot  \mathbf{G}^{\gamma} ) = \hat{\psi}_t \odot \bar u_t \odot \psi_t$ for $t \in [T]$. It is easily verified that this choice of $\hat{\psi}$ and $\psi$ are given in equations (\ref{eqn:algo_formulas_f}) and (\ref{eqn:algo_formulas_b}), respectively, given the definition from  \cref{eqn:algo_formulas_lemma_pf-1}.

    Let $t_1, t_2 \in [T]$ such that $t_1 < t_2$. Applying the results from \cite[Lemma 2 and proof of Proposition 2]{ElvanderHaaslerJakobssonKarlsson2020} and making the same arguments as for the one-marginal projection gives that $ P_{t_1, t_2}( \mathbf{K} \odot  \mathbf{U}^{\lambda} \odot  \mathbf{G}^{\gamma}) = \diag( \hat{\psi}_{t_1} \odot  \bar{u}_{t_1} )  M_{t_1, t_2} \diag( \bar u_{t_2}   \odot \psi_{t_2} )$,
    where 
    \begin{equation*}
         M_{t_1, t_2}
         =
        (K_{t_1 +1} \odot G^{\gamma_{t_1}}_{t_1 + 1}) \diag(\bar u_{t_1+1}) 
         \dots 
         (K_{t_2-1} \odot G^{\gamma_{t_2 -2}}_{t_2-1}) \diag(\bar u_{t_2-1})
         (K_{t_2} \odot G^{\gamma_{t_2 -1}}_{t_2}),
    \end{equation*}
     from where the result follows.
\end{proof}

\end{document}